\documentclass[a4paper,12pt,english]{article}
\usepackage{authblk}

\usepackage[mathscr]{eucal}
\usepackage{amsmath,amsfonts,amssymb,amsthm}
\usepackage{times}
\usepackage{hyperref}

\numberwithin{equation}{section} 

\newtheorem{thm}{Theorem}[section]

\newtheorem{remark}[thm]{Remark}

\newtheorem{theorem}[thm]{\textbf{Theorem}}
 \newtheorem{prop}[thm]{Proposition}
 \newtheorem{lemma}[thm]{Lemma}
 \newtheorem{definition}[thm]{Definition}
 \newtheorem{example}[thm]{Example}
\newtheorem{cor}[thm]{Corollary}

\renewcommand{\tilde}{\widetilde}
\renewcommand{\hat}{\widehat}

\newcommand{\bref}[1]{\textbf{\ref{#1}}}

\newcommand{\gh}[1]{\mathrm{gh}(#1)}

\newcommand{\dx}{\mathrm{d}_X}

\renewcommand{\d}{\partial}
\renewcommand{\dh}{\mathrm{d_h}}

\renewcommand{\geq}{\,{\geqslant}\,}
\renewcommand{\leq}{\,{\leqslant}\,}

\newcommand{\binner}[2]{%
  {\langle}\kern-4.15pt{\langle}#1{,}\,#2{\rangle}\kern-4.15pt{\rangle}}
\newcommand{\commut}[2]{[#1{,}\,#2]}

\newcommand{\half}{\mathchoice{%
    \ffrac{1}{2}}{\frac{1}{2}}{\frac{1}{2}}{\frac{1}{2}}}

\newcommand{\ffrac}[2]{\raisebox{.5pt}%
  {\footnotesize$\displaystyle\frac{#1}{#2}$}\kern1pt}

\newcommand{\red}{\mathrm{red}}

\newcommand{\dl}[1]{\mathchoice{\ffrac{\d}{\d #1}}{\frac{\d}{\d #1}}{\ffrac{\d}{\d #1}}{\ffrac{\d}{\d #1}}}

\newcommand{\Liealg}{\mathfrak} 
\newcommand{\algg}{\Liealg{g}}

\newcommand{\cC}{\mathcal{C}}
\newcommand{\cD}{\mathcal{D}}

\newcommand{\fZ}{\mathbb{Z}}

 \def\cE{\mathcal{E}}

 \def\cI{\mathcal{I}}
 \def\cJ{\mathcal{J}}
\def\cK{\mathcal{K}}

\def\cO{\mathcal{O}}
 \def\cP{\mathcal{P}}
 \def\cT{\mathcal{T}}
 
\newcommand\blfootnote[1]{%
  \begingroup
  \renewcommand\thefootnote{}\footnote{#1}%
  \addtocounter{footnote}{-1}%
  \endgroup
}

\numberwithin{equation}{section}
\usepackage{cleveref}
\usepackage{cite}
\usepackage{tikz-cd}
\usepackage[left=1.5cm,right=1.5cm,
top=1cm,bottom=2cm,bindingoffset=0cm]{geometry}

\newcommand{\Co}{\mathrm{C}}
\newcommand{\We}{\mathrm{W}}

\newcommand{\dX}{\mathrm{d_X}}

\newcommand{\cY}{{\mathcal{Y}}}
\newcommand{\Ee}{E}
\newcommand{\Ered}{E_{\red}}
\newcommand{\Ecl}{E_{cl}}
\newcommand{\YM}{{\scriptscriptstyle \mathrm{YM}}}

\title{Boundary structure of gauge fields on \\ asymptotically AdS spaces}
\vspace{2em}

\author{Maxim Grigoriev$^{\dagger}$~~ and\, Mikhail Markov}
\affil{ \textsl{Service de Physique de l'Univers, Champs et Gravitation, \protect\\ Universit\'e de Mons, 20 place du Parc, 7000 Mons, 
Belgium \vspace{5pt}} }  \vspace{5em}

\date{}

\begin{document}
\maketitle

\begin{abstract}
We study the boundary structure of asymptotically AdS gravity and (gauge) fields defined on this background by employing the gauge PDE approach. The essential step of the construction is the incorporation of the boundary-defining function among the fields of the theory, which allows us to realise the asymptotic boundary as a space-time submanifold by employing the gauge PDE implementation of Penrose's concept of asymptotically-simple space. In so doing the gauge PDE describing the boundary structure is obtained by restricting to the boundary of spacetime and simultaneously restricting to the boundary of the field space by setting the boundary defining function to zero.  To implement this step systematically we introduce a notion of $Q$-boundary which seems to be new. The main concrete result of this work is the construction of an efficient boundary calculus, which gives a recursive procedure to obtain the explicit form of the equations satisfied by the boundary fields and their gauge transformations for boundary dimension $d \geq 3$. These include obstruction equations (such as Bach equation or Yang-Mills equation for $d=4$) and generalised conservation equations in the subleading sector. In particular, we derive the explicit form of the higher conformal Yang-Mills equation for $d=8$. The approach is very general and, in principle, applies to generic (gauge) fields on the Einstein gravity background  producing a conformally-invariant gauge theory on the boundary, which  describes their boundary structure. It can be considered an extension of the Fefferman-Graham construction, which takes into account both the leading and subleading sectors of the bulk fields.

\vfill\phantom{a}

\end{abstract}

\blfootnote{Both authors are supported by the ULYSSE Incentive
Grant for Mobility in Scientific Research [MISU] F.6003.24, F.R.S.-FNRS, Belgium.}
\blfootnote{${}^{\dagger}$ Also at Lebedev Physical Institute and Institute for Theoretical and Mathematical Physics, Lomonosov MSU, Moscow, Russia.}

\newpage
\tableofcontents

\
\section{Introduction}

Gauge theories on manifolds with asymptotic boundaries emerge in the context of AdS/CFT correspondence~\cite{Maldacena:1997re,Witten:1998qj,Aharony:1999ti} and flat-space holography~\cite{Strominger:2013jfa, Pasterski:2016qvg,Strominger:2017zoo }, see also 
\cite{Bondi:1962px,Sachs:1962zza,Penrose:1962ij,Geroch:1977big,Starobinsky:1982mr,Fefferman-Graham:1985ambient,Brown:1986nw,deHaro:2000xn,Skenderis:2002wp}
for important earlier developments and more recent contributions relevant in the present context. 
If the bulk theory is local, it naturally induces a number of geometrical structures on the boundary, which can be again interpreted as a local gauge theory defined therein. For instance, asymptotically AdS (AAdS) gravity\footnote{In the mathematical literature vacuum solutions of AAdS gravity are referred to as asymptotically Poincar\'e-Einstein manifolds.} induces a conformal metric (subject to conformal gravity equations in the case of even boundary dimension) along with the so-called subleading mode. 

The identification of the boundary structure is often performed by mapping the spacetime into the interior of a manifold with boundary so that the asymptotic boundary becomes a usual boundary. In so doing the metric diverges when approaching the boundary, but its conformal class remains well-defined so that one can pick a suitable representative, so-called ``unphysical metric'', which is well defined at the boundary. 
This idea was put forward by R.~Penrose~\cite{Penrose:1962ij}
and gives a precise meaning to the near-boundary analysis provided one has chosen coordinates and imposed a suitable gauge condition (if gauge fields are present). However, the subtlety is that the unphysical metric is defined modulo Weyl rescalings so that the boundary structure inherits Weyl symmetry and is naturally studied using the conformal geometry tools~\cite{Friedrich:1981wx,friedrich1983cauchy,Fefferman-Graham:1985ambient,Gover:2011rz,RodGover:2012ib,Herfray:2021qmp,Herfray:2021xyp},
see also~\cite{Kroon:2016ink,Curry:2014yoa}. In particular, the near-boundary analysis of AAdS gravity amounts to the celebrated Fefferman-Graham (FG) construction~\cite{Fefferman-Graham:1985ambient,Fefferman:2007rka} which was initially proposed as a tool to study conformal invariants. Strictly speaking, FG construction operates in terms of the ambient spacetime extension of the AAdS space-time so that the conformal boundary is identified with a suitable subquotient of the ambient spacetime.

Although gravity is a gauge theory, its gauge invariance is usually controlled by geometrical methods which allow one to study boundary structure without resorting to gauge-theoretical techniques. Moreover, even in the case of Yang-Mills (YM) theory on AAdS space a geometrical approach can be developed~\cite{Gover:2023rch}. Nevertheless, the general setup requires a specifically designed  formalism.  Of course, one can always get rid of the gauge invariance by imposing suitable gauge conditions but this complicates the technique and hides the critically important structures and symmetries.

As far as general local gauge theories on manifolds with boundaries are concerned, a systematic framework to study the boundary structure is based on representing the system as a so-called gauge PDE and was put forward in~\cite{Bekaert:2012vt,Bekaert:2013zya,Bekaert:2017bpy}, see also~\cite{Grigoriev:2022zlq,Grigoriev:2023kkk} for a more recent discussion. Gauge PDE (gPDE) is a very general and flexible generalisation of the Batalin-Vilkovisky (BV)~\cite{Batalin:1981jr,Batalin:1983wj} formulation on jet-bundles~\cite{Barnich:1995ap,Barnich:2000zw} known in the context of local BRST cohomology, see~\cite{Grigoriev:2019ojp} for further details and~\cite{Barnich:2010sw,Barnich:2004cr} for the earlier but less geometrical version of the gPDE concept.  As particular cases, the gPDE approach includes the AKSZ construction of topological theories~\cite{Alexandrov:1995kv}, unfolded formalism of higher spin theories~\cite{Vasiliev:1980as,Vasiliev:2005zu}, and the aforementioned jet-bundle BV formulation. Let us stress that here we refer to the formulations at the level of equations of motion. Moreover, all considerations of this work are limited to this level. In particular, presymplectic structures, Lagrangians, charges etc. are not discussed. For the Lagrangian version of gPDE formalism see e.g.~\cite{Grigoriev:2022zlq,Dneprov:2025eoi} and references therein. Alternative approaches to gauge theories on manifolds with boundaries, which are  based on one or another extension of the BV formalism can be found in~\cite{Barnich:2001jy,Cattaneo:2012qu,Mnev:2019ejh,Baulieu:2024oql}.

In the gPDE approach the structure of a local gauge field theory is encoded in the $\fZ$-graded fibre bundle over $T[1]X$ (shifted tangent bundle over the spacetime $X$), equipped with a compatible homological vector field $Q$. Solutions and their gauge transformations can be expressed in terms of sections of the underlying bundle. A crucial property that makes gPDEs useful in the context of spacetimes with boundaries is that they behave well with respect to restricting to spacetime submanifolds and boundaries in particular, giving a systematic way to describe boundary structure without the need to fix gauges or pick particular spacetime or field-space coordinates.  Of course,  using one or another coordinate system can be very convenient in analysing the boundary structure of a given gauge theory but the gauge PDE encoding this structure is defined in a coordinate-free way. The gauge PDE approach to boundary structure was successfully applied to generic gauge fields on AdS space in~\cite{Bekaert:2012vt,Bekaert:2013zya,Chekmenev:2015kzf} by employing the ambient space representation of the AdS space and its conformal boundary.

In this work we employ the version of gauge PDE approach to boundary structure of asymptotically flat/AdS gravity developed in~\cite{Grigoriev:2023kkk} in the study of asymptotic symmetries. In contrast to the ambient space FG construction, it operates in terms of the spacetime itself at the price of introducing an additional compensator field $\Omega$ playing a role of the boundary defining functions and the auxiliary Weyl symmetry $\delta_{\bar\lambda}\Omega=\bar\lambda \Omega$. Such a reformulation of gravity, called conformal-like GR,  allows us to reformulate Penrose's definition of asymptotically simple space in the gauge-theoretical and more specifically gauge PDE language, giving a systematic derivation of the gPDE describing the boundary structure of asymptotically AdS/flat gravity.  The extension of this prescription to systems involving (gauge) fields on the gravity background turns out to be rather natural and is constructed in this work. The equations appearing in the sector of gravity are often similar to those encountered in the context of conformal geometrical approach to boundary structure though their interpretation is somewhat different.

Although identifying the boundary gPDE for AAdS gravity (possibly with additional gauge fields) is straightforward, its structure is generally rather involved and the underlying gauge theory is to be reformulated in a human-readable form by e.g. giving a set of fields, the PDEs they satisfy, and the gauge transformations they are subject to. While for linearized systems this can be done in a more or less straightforward way, see e.g.~\cite{Bekaert:2013zya,Chekmenev:2015kzf}, the analysis of nonlinear systems requires a specific technique which we call boundary calculus. At a more technical level, we identify a certain boundary gauge PDE that describes unconstrained fields and their gauge transformations in a tractable way. This gauge PDE comes equipped with a system of functions and vector fields that can be computed recursively and in terms of which the boundary equations of motion and gauge transformations can be explicitly found. 

In this work we explicitly construct the boundary calculus for a number of typical examples such as AAdS gravity, a critical scalar field and YM theory on the generic AAdS gravity background. In particular, we explicitly present the, apparently new,  higher conformal YM equation induced on 8-dimensional boundary as well as the corresponding equations of motion and gauge transformations for the YM subleading. The construction of general boundary calculus for AAdS gravity and  scalar field and YM field defined on its background as well as the explicit form of the boundary structure in lower dimensions are the main results presented in this work. 

While the structure of leading modes of AAdS gravity and gauge fields on its background is relatively well studied in the existing literature, the subleading modes are often discussed less explicitly. Our work partially covers this gap by giving an exhaustive description of the equations of motion and gauge transformations for the subleading sector of AAdS gravity together with YM theory and scalar field. It turns out that in all these cases the boundary system has the structure of a linear subleading defined on the background of the generally nonlinear  system describing the leading mode. In particular, equations of motion and gauge transformations for the subleading depend on the leading mode and its gauge parameters. In the gPDE language this situation is described by a bundle of gauge PDEs where the base and the fibre describe, respectively,  the leading and the subleading. Such objects are known as gPDEs over background~\cite{Dneprov:2025eoi} and, surprisingly, they naturally emerge on the boundary of irreducible systems such as AAdS gravity or YM theory. Note that in the particular case of linear gauge fields on AdS space this was observed already in~\cite{Bekaert:2013zya}.

The paper is organized as follows: Section~\bref{sec:prelim} introduces the basics of the gPDE approach to boundary structure and recalls the gPDE formulation of conformal-like GR. In particular, the notion of $Q$-boundary is also proposed there. In  Section~\bref{sec:pre-min} we discuss preminimal model for conformal-like gravity  and use it to introduce a number of objects and techniques needed in the subsequent analysis. Section~\bref{sec: boundary system} contains the construction of the boundary calculus in the case of gravity and is central for this work. The boundary calculus is explicitly applied to study the boundary structure of gravity in Section~\bref{sec:boundary-GR} where we also explain how the explicit space-time form of the equations of motion and gauge transformations on the boundary are obtained. Finally, Section~\bref{sec:matter} is devoted to the boundary structure of (gauge) fields on the AAdS gravity background. The corresponding version of boundary calculus is constructed for Klein-Gordon field and YM theory and the explicit description of YM boundary structure for $d=4,6,8$ is obtained.

\section{Preliminaries}\label{sec:prelim}

\subsection{Gauge PDEs}\label{sec:gpde}

In this section we recall the basics of the gauge PDE approach to local gauge theories and its extension to the case of theories on manifolds with boundaries.

The main building block is a notion of $Q$-manifold which is a pair $(M,Q)$, where $M$ is a $\mathbb{Z}$-graded supermanifold, equipped with degree-one odd vector field $Q$ satisfying $Q^2=0$.  We denote the $\mathbb{Z}$-grading by $\gh{\cdot}$ and the $\mathbb{Z}_2$-parity by $|\cdot|$. The systems considered in this work do not involve physical fermionic fields and hence the two gradings are compatible in the sense that $|\cdot|=\gh{\cdot} \,\text{mod} \,2$, so that it is enough to keep track of the ghost degree only.

The basic example of a $Q$-manifold is the shifted tangent bundle $T[1]X$ of a smooth manifold $X$ or, more generally, a graded supermanifold. Its local coordinates can be chosen as $\{x^{\mu},\theta^\mu\}$, where $x^{\mu}$ are local coordinates on the base  $X$ and $\theta^\mu$ are the graded-anticommuting fibre coordinates in the coordinate frame. Functions on $T[1]X$ are isomorphic as an algebra to differential forms on $X$ and this identifies the de Rham differential of $X$ with the homological vector field $\dx \equiv \theta^\mu\frac{\partial}{\partial x^\mu}$. 

Natural morphisms between $Q$-manifolds, so-called $Q$-maps,  are degree-preserving maps compatible with the $Q$-structures. In terms of algebras of functions on $(M_1,Q_1)$ and $(M_2,Q_2)$ such a map $\mu:M_1 \to M_2$ is a homomorphism $\mu^*:\cC^\infty(M_2) \to \cC^\infty(M_1)$ such that $Q_1\circ \mu^{*}=\mu^* \circ Q_2$. In particular, there are natural notions of $Q$-submanifolds and $Q$-bundles. In the latter case both the total space and the base are $Q$-manifolds and the projection is required to be a $Q$ map~\cite{Kotov:2007nr}.

If one disregards space-time locality, $Q$-manifold is precisely a geometrical object corresponding to a gauge theory defined at the level of equations of motion
and formulated in terms of the non-Lagrangian version of BV formalism, see~\cite{Barnich:2004cr,Lyakhovich:2004xd}. If one prefers to maintain space-time locality in a manifest way, one can for instance resort to non-Lagrangian BV defined on a suitable jet-bundle~\cite{Barnich:2004cr}, see also~\cite{Barnich:2000zw} for a usual BV formalism on jet-bundles. However, a more flexible and geometrical approach which suits our needs in this work is provided by so-called gauge PDEs.

\begin{definition}
1) Gauge PDE $(E,Q,T[1]X)$ is a $Q$-bundle $\pi: (E,Q)\rightarrow(T[1]X,d_X)$, where $X$ is a space-time manifold. \\
2) Solutions to $(E,Q,T[1]X)$ are sections $\sigma: T[1]X \to E$ satisfying 
\begin{align}
    d_X\circ \sigma^{*}=\sigma^{*}\circ Q
\end{align}
3) Infinitesimal gauge transformations of section $\sigma$ are defined as 
\begin{align}
    \delta\sigma^{*}=\sigma^{*}[Q,Y],
\end{align}
where $Y$ is a vertical field of degree $-1$ on $E$, called a gauge parameter. In a similar way one defines gauge-for-gauge symmetries etc.
\end{definition}
This definition may be further restricted in order to exclude gPDEs corresponding to pathological field theoretical systems, see~\cite{Grigoriev:2019ojp,Grigoriev:2022zlq} for more details.

Given a local gauge theory defined at the level of equations of motion it can be reformulated as a local BV system, i.e. a graded fibre bundle $\cE\to X$, where $X$ denotes the space-time manifold, together with the evolutionary homological vector field $s$, called BRST differential,  defined on the infinite-jet bundle $J^\infty(\cE)$ of $\cE\to X$. In its turn, a local BV system naturally defines a gauge PDE, at least locally. Indeed, taking as $E$ the pullback of $J^\infty(\cE)$ by the projection $T[1]X \to X$ one can identify functions on $E$ with horizontal forms on $J^\infty(\cE)$. On the horizontal forms one has two anticommuting differentials: the horizontal differential $\dh$ and the Lie derivative $L_s$ with respect to $s$. Seen as a vector field on $E$ the total differential $\dh+L_s$ makes $(E,Q,T[1]X)$ into a gPDE. At least locally, the gauge theory encoded by this gPDE is equivalent to the starting point BV system~\cite{Barnich:2010sw}. 

An interesting phenomenon occurs when the original gauge theory is diffeomorphism invariant, i.e. when space-time diffeomorphisms are among the gauge transformations encoded in $s$. More precisely, the term in $s$ that is linear in the respective ghost variables $\xi^a$ has the form $\xi^aD_a+\ldots$, where $D_a$ are total derivatives and $\ldots$ denote terms not involving undifferentiated $\xi^a$. It turns out that in this case one can change fibre coordinates in such a way that the total $Q$-structure factorizes as $Q=s+\dX$, where $s$ is a vector field on the typical fibre. Geometrically, this means that the underlying $Q$-bundle is locally trivial.  Moreover, in a suitable coordinate system the expression of $s$ coincides with the initial $s$ restricted to a fibre.  This gives a very simple procedure to build  gPDE formulations for  diffeomorphism invariant systems.

When dealing with gauge theories one often finds that the same theory can be formulated using different sets of fields and gauge generators. Such equivalent formulations can be related via the elimination of auxiliary fields and gauge-fixing algebraic gauge symmetries. In the gPDE language this equivalence is implemented as follows:

\begin{definition}
\label{def:equiv-red}
Sub-gPDE $(E_{\red},Q_{\red},T[1]X)\subset (E,Q,T[1]X)$ is called an equivalent reduction of $E$
if $(E,Q,T[1]X)$ is a locally-trivial $Q$-bundle $\Pi:E\to E_\red$ with contractible fibre and $\Pi$ is compatible with the bundle structure $\pi:E\to T[1]X$, $\pi_{E_{red}}:E_{\red}\to T[1]X$, i.e. $\pi_E =\pi_{E_{red}}\circ \Pi$. Two gauge PDEs are said equivalent if they can be obtained as equivalent reductions of a third one.
\end{definition}
In coordinate terms this means that on $E$ there is a local coordinate system $w^\alpha,v^\alpha,\phi^i$ (here and below we omit the base space coordinates $x^\mu,\theta^\mu$) such that 
\begin{equation}
\label{contr-pairs}
Qw^\alpha=v^\alpha, \qquad Q\phi^i=Q^i(x,\theta,\phi)\,,
\end{equation}
and $E_\red$ is singled out by $w^\alpha=0, v^\alpha=0$
and the projection is given by $\pi^* (\phi^i|_{E_\red})=\phi^i$. Moreover, $w^\alpha,v^\alpha$ are global coordinates on the fibres of $E$ over $E_\red$. A useful criterium to identify an equivalent reduction is the existence of independent functions $w^\alpha$ such that $w^\alpha,Qw^\alpha$ are again independent and their zero locus is a subbundle. Locally and under some technical conditions this criterium actually gives a sufficient condition i.e. once can find $\phi^i$ that satisfy \eqref{contr-pairs} and complete $x^\mu,\theta^\mu,w^\alpha,Qw^\alpha$ to a local coordinate system on $E$. If one forgets about the bundle structure the notion reduces to the usual elimination of contractible pairs as it is known in the context of local BRST cohomology. For locally-trivial gPDEs describing diffeomorphism-invariant theories the equivalent reduction can be formulated in terms of the fibre $Q$-manifolds, at least locally.  
Further details on equivalence of gPDEs can be found in~\cite{Grigoriev:2019ojp}.

If it is clear from the context what the spacetime manifold of the system is, we tend to omit the explicit mention of $T[1]X$ when denoting gPDEs in order to keep conventions light, that is we simply write $(E,Q)$ instead of $(E,Q,T[1]X)$. This is especially natural for globally trivial gPDEs, whose structure is essentially encoded in the corresponding fibre $Q$-manifolds.

\subsection{Gauge PDEs on manifolds with boundaries} \label{sec:gpde-bound}

It turned out that gauge PDEs are especially useful if one is interested in the structure induced by a local gauge theory on a spacetime submanifold and, particularly,  boundary and/or low-dimensional space-time strata. The idea is very simple: consider a gauge PDE $(E, Q, T[1]X)$, where $X$ possesses a nontrivial boundary $\Sigma = \partial X$ and denote the embedding map by $i: \Sigma \hookrightarrow X$. To be concrete, here we talk about boundary $\Sigma$ but the general discussion applies to any submanifold $\Sigma \subset X$. This induces a new gPDE $i^{*}E$ on $\Sigma$, which is simply the restriction (also called pullback) of $E$ to $T[1]\Sigma \subset T[1]X$. Here, by abusing notations, we also denote by $i$ the corresponding pushforward $T[1]\Sigma \hookrightarrow T[1]X$. It is easy to verify that $i^{*}E$ is again a gPDE. In particular, $Q$ is tangent to $i^{*}E \subset E$.  In a slightly less general form this approach to boundary values was originally employed to study the boundary values of gauge fields in~\cite{Bekaert:2012vt,Bekaert:2013zya,Bekaert:2017bpy,Grigoriev:2018wrx}.

The pullback $i^{*}E$ of the original gPDE to the boundary $T[1]\Sigma$ describes the boundary values of the fields of the original $E$. In general, one in addition wants to be able to impose additional conditions on the boundary values of fields, gauge parameters, etc. This leads to the following more general notion:
\begin{definition}\cite{Grigoriev:2023kkk}
     1) By a gauge PDE with boundaries we mean the following data: $(E,Q,T[1]X,E_{\Sigma},T[1]\Sigma)$, where gPDE $(E,Q,T[1]X)$ is a gPDE on $X$ and $(E_{\Sigma},Q_{\Sigma},T[1]\Sigma)$ is a gPDE on the boundary $\Sigma=\partial X$, which is a sub-gPDE of $i^{*}E$. $Q_\Sigma$ denotes the restriction of $Q$ on $E_{\Sigma}$.
     
    2) Solutions to $(E,Q,T[1]X,E_{\Sigma},T[1]\Sigma)$ are sections $\sigma:T[1]X \to E$ such that the restriction of $\sigma$ to $T[1]\Sigma$  is a solution to  $(E_{\Sigma},Q_{\Sigma},T[1]\Sigma)$

    3) Gauge symmetries are generated by the vector field of the form $\commut{Q}{Y}$, where $Y$ is a vertical vector field on $E$ required to be tangent to $E_\Sigma$.
\end{definition}

Although the above definition can often be applied in a straightforward way,  its implementation in the case of asymptotic boundary requires some extra care. The point is that in this case the boundary defining function is an independent field and it has to be set to zero in a way compatible with $Q$. In other words the fibre of $E$ over the interior differs from the fibre over the boundary or, in other words, $E$ is slightly  not locally trivial. But the restrictions of $E$ to the interior and to the boundary are locally trivial. 

To be more precise, in the case of gravity we employ a gauge theoretical implementation of Penrose's notion of asymptotically-simple space~\cite{Penrose:1964ge}. Namely, one of the fibre coordinates  is the boundary defining function $\Omega$. In the interior of $X$ one requires $\Omega>0$ while on the boundary $\d X$ one sets $\Omega=0$ (in the gPDE setup one also sets $Q\Omega=0$). Roughly speaking, typical fibre $F$ of $E$ is a manifold with boundary and if we want to describe the system without asymptotic boundary we not only take the interior $Int(X)$ of the base but also take as fibre the interior $Int(F)$ of the initial fibre. At the same time the system induced on the boundary has $\d X$ as the base space and $\d F$ as the fibre.

A possible way to incorporate this setup in a more precise and geometrical way is as follows: introduce the notion of $Q$-manifold with $Q$-boundary. By definition, $Q$-boundary has codimension $2$ and is such that near the boundary  there is a coordinate system $\Omega,\lambda,z^I$, $\gh{\Omega}=p, \gh{\lambda}=p+1$ such that $Q\Omega=\lambda$ and the $Q$-boundary is singled out by $\Omega=0, \lambda=0$. 

Restricting for simplicity to the case of globally trivial $Q$-bundles we now take as the base $(T[1]X,\dx)$ whose $Q$-boundary is $(T[1](\partial X),\mathrm{d}_{\partial X})$ and the fibre to be a $Q$-manifold $(F,q)$ with $Q$-boundary $(\d^Q F,q)$.
Given such a gPDE there are two natural gauge PDEs associated to it. The first one is $(Int(F),q) \times (T[1] (Int(X)),\dX)$ that describes a gauge theory in the bulk. Here by some abuse of conventions we use $Int(F)$ to denote the interior of the manifold with $Q$-boundary. If $\Omega$ is the coordinate corresponding to the boundary defining function then the interior is locally given by the domain $\Omega>0$. The second, the boundary gPDE, is $(\d^Q F,q) \times (T[1](\partial X),\mathrm{d}_{\partial X})$ so that it is indeed obtained by setting $\Omega=0=Q\Omega$ and taking $\d X$ as the spacetime manifold.

\subsection{Implicit definition of gPDEs}\label{sec:implicit}

Although the language of gPDEs is very geometrical and concise it is often difficult to apply it in concrete examples because in realistic applications the explicit form of $Q$ in terms of local coordinates on $E$ is complicated or not even known. This generally happens if the underlying gauge theory is  nontopological and nonlinear, see e.g.~\cite{Grigoriev:2024ncm,Misuna:2024dlx} for a recent discussion.

A possible way out is to define $(E,Q)$ as a sub-gPDE of a bigger gPDE, e.g. the one arising from the jet-bundle formulation of the system. This can be naturally done by allowing $E$ to have negative ghost degree coordinates in such a way that the extended $Q$ incorporates a Koszul-Tate differential that implements restriction to the surface. This is the standard idea of Batalin-Vilkovisky formalism and its implementation in the gauge PDE setup leads to a relatively concise gPDE formulation of gauge theories. In particular, the fact that a given gauge theory formulated in BV terms admits an equivalent gPDE representation was shown in exactly this way~\cite{Barnich:2010sw}, see also \cite{Barnich:2004cr,Grigoriev:2019ojp}.

In this work we refrain from employing negative degree variables but simply allow our actual gauge PDE to be a subbundle in the ambient bundle $\pi:E \to T[1]X$ equipped with a $\pi$-projectable vector field $Q$ of degree $1$. If $\cI\subset \cC^{\infty}(E)$ is the ideal of functions vanishing on the subbundle we denote this data as $(E,Q,\cI,T[1]X)$. We assume that the zero locus of $\cI$ is a regular subbundle of $E$ and that $Q\cI\subset \cI$ (i.e. $Q$ is tangent to the zero-locus of $\cI$). Moreover, in all our examples the ambient $Q$ is nilpotent so that $(E,Q,\cI,T[1]X)$ is itself a gPDE. Of course, one could generally ask less, namely, that $Q^2 f\in \cI$ for any function $f$. This is enough to have a nilpotent $Q$ on the zero locus. 

If we define our gPDE as $(E,Q,\cI)$ and assume $Q^2=0$ then its solutions are  solutions to the underlying gPDE $(E,Q)$ that satisfy in addition $\sigma^*(\cI)=0$. Such formulation of the system is especially useful if $(E,Q,T[1]X)$ is off-shell. This means that its space of solutions is locally parametrised by a set of unconstrained fields. More formally, as a gPDE $(E,Q)$ is equivalent  to a nonegatively-graded jet-bundle equipped with an evolutionary homological vector field. More precisely, as we recalled earlier, a jet-bundle equipped with an evolutionary homological vector field $s$ can be pulled back to $T[1]X$ and equipped with $\dh+L_s$, giving a gPDE. If $(E,Q)$ is off-shell, all the nontrivial differential equations on our fields originate from $\cI$. Mention that somewhat similar formulation are known in the literature, see~e.g.~\cite{Vasiliev:2005zu,Grigoriev:2010ic,Grigoriev:2012xg}.

If $(E,Q)$ is off-shell, one can usually define a concept of prolongation. Namely, one can equip  $E$ with a distribution generated by vector fields $\nabla_A$, $\gh{\nabla_A}=0$, and $\commut{Q}{\nabla_A}=0$,  such that $\cI$ is generated by 
$\nabla_{A_1}\ldots \nabla_{A_k}\cY_\alpha$, $k\geq 0$, where $\cY_\alpha$ are some functions on $E$. If in addition $\sigma^*(\cY_\alpha)=0$ implies $\sigma^*(\cI)=0$ for any $\sigma$ solving $(E,Q)$, we refer to $\cY_\alpha$ as to the symbols of the equations of motion. 

If we work with implicitly defined gPDEs, it is natural to call two such implicit gPDEs equivalent if the corresponding zero loci are equivalent gPDEs. In particular, the equivalent reduction of the underlying gPDE $(E,Q,T[1]X)$ induces an equivalent reduction of the implicitly defined gPDE provided the corresponding contractible pairs $w^\alpha$ and $Qw^\alpha$ remain independent functions  when restricted to the zero locus of $\cI$.\footnote{Of course this is only a sufficient condition. The reduction could well affect the ideal without spoiling the equivalence.} More precisely, if $i: \tilde E \to E$ is an equivalent reduction then under the above condition $(\tilde E,Q,i^*(\cI))$ is an equivalent gPDE. 

An additional natural equivalence of implicit gPDEs relates $(E,Q,\cI)$ and
$(\tilde E,Q,\tilde \cI)$, where $(\tilde E,Q)\subset (E,Q)$ is a sub-gPDE and the zero locus of $\cI$ in $E$ coincides with the zero locus of $\tilde\cI$ in $\tilde E$. The passage from $E$ to $\tilde E$ corresponds to solving some of the constraints from $\cI$. Note that in this situation $\tilde\cI$ is a restriction of $\cI$ to $\tilde E$. 

Finally, let us mention that there are less obvious ways to encode a given gauge PDE in terms of an explicitly defined and often finite-dimensional geometrical object. This can be constructed as a sort of finite-dimensional consistent truncation of a gPDE and is known as a weak gPDE~\cite{Grigoriev:2024ncm}.

\subsection{Conformal-like gravity}\label{sec:conformal-like}
Our starting point is the equivalent reformulation of gravity (GR) in the bulk, obtained by adding artificial Weyl symmetry together with its associated compensator field. In the presence of boundaries this gives the gauge theoretical implementation of the Penrose description of the asymptotically-simple spaces. Throughout this work we assume that the space-time dimension $D\equiv \dim X$ satisfies $D \ge 4$ and for definiteness, we fix the metric signature to be Lorentzian $(+,-\dots,-)$, although none of the results in this work depend on the choice of signature. The boundary dimension is denoted by $d\equiv D-1$.

More precisely, in addition to the metric $g_{ab}$ one introduces a scalar field $\Omega$ subject to $\Omega>0$. The gauge transformations read as
\begin{equation}
\delta g_{bc}=  \xi^a \d_a g_{bc}+g_{ac}\d_{b}\xi^a+g_{ba}\d_c \xi^a+2\lambda g_{bc},\qquad  \delta \Omega=\xi^{a}\d_{a}\Omega+\lambda\Omega\,,
\end{equation}
where gauge parameters $\xi^a,\lambda$ are components of the vector field and scalar, respectively.  The above system is equivalent to the off-shell (i.e. with trivial equations of motion) gravity and is called off-shell conformal-like GR. The equivalence is easily seen by observing that the gauge $\Omega=1$ is reachable and completely kills the $\lambda$ parameter. In this gauge one stays with metric and its diffeomorphisms only, resulting in the off-shell GR.
Let us stress that this equivalence holds if $\Omega>0$.

We can now reformulate the above system as a gPDE $(\Ecl,Q)$ following the steps recalled in Section~\bref{sec:gpde}. More precisely, $(\Ecl,Q)$ is defined as follows: its typical fibre can be identified with the fibre of the jet-bundle for fields $g_{ab},\Omega$ and ghosts $\xi^a,\lambda$, whose fibres are coordinatised by 
\begin{equation}
D_{(c)}g_{ab}\,,\quad D_{(c)}\Omega\,, \quad D_{(c)} \xi^a \quad D_{(c)}\lambda\,.
\end{equation}
where $D_a$ denote the usual total derivatives on the jet-bundle and ${(a)}$ denotes a symmetric multiindex.  The action of $Q$ in these coordinates is given by:
\begin{equation}
\label{Q-GR}
\begin{gathered}
       Qx^\mu=\theta^{\mu},\quad Q g_{bc}=\xi^a D_a g_{bc}+g_{ac}D_{b}\xi^a+g_{ba}D_c \xi^a+2\lambda g_{bc},\quad
    Q\xi^b=\xi^a D_a \xi^b\,, \\
    Q\Omega=\xi^{a}D_{a}\Omega+\lambda\Omega,\qquad Q\lambda=\xi^a D_a\lambda,\qquad \commut{D_a}{Q}=0 \,,
\end{gathered}
\end{equation}
see~\cite{Grigoriev:2023kkk} and Refs. therein for further details. It is important to stress that although we have constructed $\Ecl$ in terms  of the starting point jet-bundle, the jet-bundle structure is not fully encoded in the gPDE structure of $\Ecl$. 

In the above gPDE formulation of conformal-like gravity and later in this work we systematically avoid discussing the global structure of the underlying bundles by assuming that all the bundles involved are globally trivial. Under this assumption all the constructions are coordinate independent in the sense that one is free to use any coordinates on the base but only $x,\theta$-independent coordinate change is allowed in the fibre.  For instance, the structure of $Q$ in~\eqref{Q-GR} is clearly not invariant under the naive orthogonal transformations with $x$-dependent parameters and these are to be adjusted to preserve $Q$. Of course, all the general notions such as gPDEs and their morphisms are perfectly global and the bundles encountered in this work can be defined globally. However, for simplicity we restrict ourselves to the case of     parallelisable spacetimes where all the relevant bundles can be assumed trivial.

So far we considered the off-shell system. In terms of the above gPDE formulation, the on-shell version can be obtained by restricting to a sub-gPDE determined by the prolongation of the Einstein equations for $\tilde g_{ab}=\Omega^{-2}g_{ab}$, rewritten in terms of $g,\Omega$. More precisely, 
\begin{prop}
The on-shell GR is described by the implicit gPDE $(\Ecl,Q,\cI_{cl})$, where the ideal $\cI_{cl}$ is generated by  
\begin{align}\label{Frho-def}
D_{(a)} G_{bc}=0\,, \quad S=0\,,
\end{align}
where
\begin{align}\label{einst-AE}
\begin{split}
        G_{bc}&\equiv D_b D_c \Omega- \Gamma_{bc}^{d}D_{d}\Omega+\Omega P_{bc} + \rho g_{bc},\\ \rho&\equiv-\dfrac{1}{D}g^{bc}(D_{b}D_{c}\Omega- \Gamma_{bc}^{d}D_{d}\Omega+P_{bc}\Omega)\\
    S&\equiv \Omega\rho+\dfrac{1}{2}D_a\Omega D^{a}\Omega-\dfrac{\Lambda}{(D-1)(D-2)}\,,
    \end{split}
\end{align}
and $\Gamma_{bc}^{d}$, $P_{bc}$ are, respectively, the Christoffel symbols and the Schouten tensor seen as functions in the jets of the metric. 
\end{prop}
The proof is based on the transformation properties of the Schouten tensor under Weyl transformations and imposing the gauge condition $\Omega=1$. Details can be found in~\cite{Grigoriev:2023kkk}. Equations \eqref{Frho-def} are known as the almost Einstein equation.

\section{Pre-minimal model of the conformal-like  GR}\label{sec:pre-min}
\subsection{Pre-minimal model}
Starting from the off-shell conformal GR represented by the gPDE $(\Ecl,Q)$ one can equivalently reduce it
in the sense of Definition \bref{def:equiv-red}. For a locally-trivial gPDEs a simple class of equivalent reductions originate from equivalent reductions of the fibre $Q$-manifolds. The one we are interested in is locally of this type, i.e. the base space plays a passive role.

Consider the sub-gPDE $(\Ered,Q)$, $r:\Ered\hookrightarrow \Ecl$ singled out by
\begin{align}
&D_{(a)}D_{(c}\Gamma_{b)d}^{e}=0,\qquad D_{(c)}D_{(a}P_{b)c}=0,
\\    &Q(D_{(a)}D_{(c}\Gamma_{b)d}^{e})=0,\qquad Q(D_{(c)}D_{(a}P_{b)c})=0\,.
\end{align}
The equations of the second line can be solved with respect to 
$D_{(c)}D_a D_b\xi^d$ and $D_{(c)}D_a D_b \lambda$ and hence this sub-gPDE is an equivalent reduction, at least locally. It is clear that this reduction does not affect the $\Omega$-sector of the system. In particular, it is perfectly defined even if we let $\Omega$ to vanish. In this case the underlying system is strictly speaking not equivalent to GR. Note that the total derivative $D_a$ is not tangent to this surface. 

As part of the coordinates on $\Ered$ we take the pullbacks of the degree $1$ coordinates $\xi^a,\lambda$ and their derivatives $C_a{}^b\equiv D_a\xi^b$, $\lambda_a \equiv D_a \lambda$ together with $g_{ab}$, $\Omega$. By some abuse of notation we use the same symbols to denote  the degree $1$ coordinates  and their pullbacks to $\Ered$. The same applies to $g_{ab}$. The action of $Q$ on these coordinates read as:
\begin{align}\label{gran-ish1}
    \begin{split}
        Q g_{bc}&=C_{b}{}^a g_{ac}+C_{c}{}^a g_{b a}+2\lambda g_{bc}\,,\qquad         Q\xi^b=\xi^a C_{a}{}^{b}\,,\\
        Q\lambda_a&= C_{a}{}^b\lambda_b +\half\xi^b \xi^c \Co_{abc}\,,\qquad\qquad\ Q\lambda=\xi^a\lambda_a\,,\\
        Q C_{b}{}^c&=C_{b}{}^a C_{a}{}^c+\lambda_b\xi^c-\lambda^c \xi_b+\delta_{b}^c \lambda_a \xi^a +\half\xi^a\xi^d \We^{c}{}_{b a d}\,, \\
   Qr^{*}\Omega&=\xi^a r^{*}D_a\Omega+\lambda r^{*}\Omega\,,
    \end{split}
\end{align}
where $\We^{c}{}_{b a d}$ and $\Co^{b}{}_{ad}$ are respectively the Weyl and Cotton tensors build out of $D_{(c)}g_{ab}$ and pulled back to $\Ered$. In particular, they possess the standard symmetry properties:
\begin{align}
        \We_{abcd}=-\We_{bacd}=-\We_{abdc}, \quad \We_{a[bcd]}=0, \quad \Co_{abc}=-\Co_{acb}, \quad \Co_{[abc]}=0.
    \end{align}
Functions  $\xi^a,\lambda,C_a{}^b,\lambda_a$ exhaust the independent coordinates of degree $1$ on $\Ered$. 

It is clear that the above reduction does not affect the sector of $D_{(a)}\Omega$ variables and in this sense amounts to the equivalent reduction of the gPDE $(\Ecl^{conf},Q)$ describing the conformal geometry, which can be identified with $D_{(a)}\Omega$-independent sector of $(\Ecl,Q)$.  More geometrically, the initial $(\Ecl,Q)$ is a $Q$-bundle  over the $(\Ecl^{conf},Q)$ with the fibre being a jet of $\Omega$ with coordinates $D_{(a)}\Omega$. Functions on the base can be identified with the functions on the total space that do not depend on $D_{(a)}\Omega$. The above equivalent reduction $(\Ered,Q) \subset (\Ecl,Q)$ originates from the equivalent reduction of $(\Ecl^{conf},Q)$ to $(\Ered^{conf},Q)$  and in this form was initially  constructed and employed in~\cite{Boulanger:2004eh,Boulanger:2004zf}, see also~\cite{Dneprov:2022jyn,Grigoriev:2023kkk}.  More precisely,
\begin{prop}\label{prop:bundle-over-conf}
$(\Ered,Q)$ is a $Q$-bundle over $(E^{conf}_{\red},Q)$. 
\end{prop}
Recall that both $(\Ered,Q)$ and $(E^{conf}_{\red},Q)$ are $Q$-bundles over $T[1]X$ and the projection
$(\Ered,Q)\to (E^{conf}_{\red},Q)$ is a map of bundles over $T[1]X$.

We refer to $(E^{conf}_{\red},Q)$ as to the pre-minimal model of conformal geometry. At the same time, it is natural to call  $(\Ered,Q)$ a pre-minimal model of conformal-like off-shell GR. The term pre-minimal reflects that one can still eliminate the metric $g_{ab}$ together with the symmetric part of the ghosts $C_a{}^b$ to obtain the minimal model, see~\cite{Dneprov:2022jyn,Grigoriev:2023kkk} for more details.

Finally, recall that the on-shell conformal-like GR is described by $(\Ecl,Q,\cI_{cl})$, where  $\cI_{cl}$ is generated by \eqref{Frho-def} which, in turn, originate from the Einstein equations. Pulling back $\cI_{cl}$ to $\Ered$ gives an ideal $\cI_{\red}$ which allows us to define the pre-minimal model of the on-shell conformal-like GR as an implicit gPDE $(\Ered,Q,\cI_{\red})$. 

\subsection{Covariant derivatives}
Our aim now is to introduce a set of functions on $\Ered$ that completes
the coordinates involved in~\eqref{gran-ish1} to a useful (possibly overcomplete) coordinate system on $\Ered$, compatible with the bundle structure $\Ered\to E^{conf}_{\red}$. To this end we need some additional ingredients. Let $C^M$ denote all the coordinates of ghost degree $1$ on $\Ered$. Consider the following vector fields: 
    \begin{align}
        \nabla_M\equiv  \Big[\frac{\partial}{\partial C^M},Q\Big].
    \end{align}
Note that because $C^M$ exhaust the degree $1$ fibre coordinates, vector fields $\nabla_M$ do not depend on the choice of the degree $0$ coordinates completing $C^M,x^\mu,\theta^\mu$ to a coordinate system on $\Ered$.

Vector fields $\nabla_M$ satisfy the following ``open-algebra'' relations:
\begin{align}
\label{nablacommut}
 [\nabla_M,\nabla_N]=\epsilon^{L}_{MN}\nabla_L+(Q\epsilon^{L}_{MN})\frac{\partial}{\partial C^{L}}\,,
 \end{align}
where $\epsilon^L_{NM}=-\epsilon^L_{MN}$ are degree zero functions on $\Ered$. To see this, note that  
\begin{align}
    [\nabla_M,\nabla_N]=[Q,r_{MN}], \qquad 
    r_{MN}\equiv \Big[\frac{\partial}{\partial C^M},\nabla_N\Big]\,.
\end{align}
Because $\gh{r_{MN}}=-1$ it can be represented as 
\begin{align}
    r_{MN}=\epsilon_{MN}^L\frac{\partial}{\partial C^{L}}\,,
\end{align}
for some degree $0$ functions $\epsilon^L_{NM}$. Note that $\epsilon^L_{NM}=-\epsilon^L_{MN}$ thanks to
\begin{align}
    r_{MN}=\Big[\frac{\partial}{\partial C^M},\nabla_N\Big]=\Big[\frac{\partial}{\partial C^M},\Big[Q,\frac{\partial}{\partial C^N}\Big]\Big]=\Big[\nabla_M,\frac{\partial}{\partial C^N}\Big]=-r_{NM}.
\end{align}
Of course,  these considerations apply to any $Q$-manifold whose coordinates are allowed to have degree $0$ and $1$ only. Such $Q$ manifolds are 1:1 with Lie algebroids.  The restriction of $\nabla_M$ to the base, i.e. a submanifold singled out by $C^M=0$, are nothing but the components of the anchor map of this algebroid.

It is convenient to introduce the following notations for $\nabla_M$ associated to different sectors of ghost variables:
\begin{align}\label{commutators}
\begin{split}
            \nabla_a \equiv\left[\frac{\partial}{\partial \xi^a}, Q\right],& \qquad \Delta^{a}{}_{b} \equiv\left[\frac{\partial}{\partial C_{a}{}^b}, Q\right], \\ 
        \Gamma^a \equiv\left[\frac{\partial}{\partial \lambda_a}, Q\right],& \qquad \Delta \equiv\left[\frac{\partial}{\partial \lambda}, Q\right]\,.
        \end{split}
\end{align}
We also employ collective notations $\Delta_A$ for $\Delta^{a}{}_{b}, \Gamma^a,\Delta$, i.e. all components of $\nabla_M$ save for $\nabla_a$.
\begin{prop}
Relations~\eqref{nablacommut} read explicitly as:
\begin{align}\label{commutrelations}
\begin{split}
&{\left[\Delta^{a}{}_{b},\Gamma^c\right] }  =-\delta_b^c\Gamma^a,  \qquad  {\left[\Gamma^a, \nabla_b\right] }  =-\mathcal{P}^{a c}_{d b} \Delta^{d}{}_{c}+\delta^{a}_{b}\Delta\,, \\ &{\left[\Delta^{a}{}_b, \nabla_c\right] }  =\delta_c^a \nabla_b,  \qquad  {\left[\Delta^{a}{}_{b}, \Delta^{c}{}_d\right]=\delta^{a}_{d} \Delta^{c}{} }_b-\delta^{c}_b \Delta^a{ }_d \\
     &{\left[\nabla_a, \nabla_b\right]=-\We^{d}{}_{cab} \Delta^{c}{ }_d-\Co_{dab}\Gamma^d -(Q\We^{d}{}_{cab})\frac{\partial}{\partial C_{c}{}^{d}}-(Q\Co_{dab})\frac{\partial}{\partial \lambda_d}} \,, \\
     \end{split}
\end{align}
where $\mathcal{P}^{d e}_{a b}\equiv (-g^{de}g_{ab}+\delta^{d}_{a}\delta^{e}_{b}+\delta^{d}_{b}\delta^{e}_{a})$ and all other commutators vanish.
\end{prop}
Note that vector fields \eqref{commutators} are well-defined on the subalgebra of degree $0$ functions where they coincide with the corresponding vector fields introduced in~\cite{Boulanger:2004eh} while their analogues in the case of GR and YM theory were discussed already  in~\cite{Brandt:1996mh}.

Note that functions of degree $0$ are closed under the action of $\nabla_M$. Moreover, on this subalgebra $\nabla_M$ form a representation of the extension of the conformal algebra $o(D,2)$. More precisely, the extension, where the Lorentz subalgebra is extended to a general linear algebra whose associated ghost variables are $C_a{}^b$. Alternatively, 
if one sets to zero all the terms involving curvatures $\We$ and $\Co$, relations \eqref{commutrelations} are just the relations of this Lie algebra. That a Lie algebra and its module appear naturally can be understood from the perspective of $L_{\infty}$-algebra. More precisely, the fibre of our gPDE $(\Ered,Q)$ (as well as that of $(E,Q)$) can be considered  a formal $Q$-manifold and hence is equivalent to an $L_{\infty}$-algebra. At the same time an $L_{\infty}$-algebra whose associated $Q$-manifold has coordinates of degree $0$ and $1$ only,  is known to describe deformations of a Lie algebra and its module. 

\begin{proof}
One reads off the commutation relations from the relation~\eqref{gran-ish1}. For instance, let $r_{ab}\equiv[\frac{\partial}{\partial\xi^a},\nabla_b]$. Because $\gh{r_{ab}}=-1$ it acts nontrivially on the degree $1$ coordinates only. Explicitly, one has
\begin{align}
      &r_{ab}C_{c}{}^{d}=\Big[\frac{\partial}{\partial\xi^a},\nabla_b\Big]C_{c}{}^{d}=\frac{\partial}{\partial\xi^a}\nabla_b C_{c}{}^{d}=\frac{\partial}{\partial\xi^a}\frac{\partial}{\partial\xi^b}Q C_{c}{}^{d}=\We^{d}{}_{cba},\\
      &r_{ab}\lambda_c=\frac{\partial}{\partial\xi^a}\frac{\partial}{\partial\xi^b}Q\lambda_c=\Co_{cba}\,, \qquad r_{ab}\lambda= r_{ab} \xi^c=0\,.
\end{align}
so that 
  \begin{align}\label{commutrelations-pr-1}
      r_{ab}=-\We^{d}{}_{cab}\frac{\partial}{\partial C_{c}{}^{d}}-\Co_{cab}\frac{\partial}{\partial\lambda_c}\,,
  \end{align}
and hence 
  \begin{align}
      [\nabla_a,\nabla_b]=[Q,r_{ab}]=-\We^{d}{}_{cab} \Delta^{c}{ }_d-\Co_{dab}\Gamma^d -(Q\We^{d}{}_{cab})\frac{\partial}{\partial C_{c}{}^{d}}-(Q\Co_{dab})\frac{\partial}{\partial \lambda_d}.
  \end{align}
The remaining relations are obtained in a similar way.
\end{proof}

\subsection{Ideals and coordinates}
In our studies of the boundary structures the following notion turns out to be useful:
\begin{definition}
\label{def-I}
Let $(E,Q,T[1]X)$ be a gPDE equipped with a set of functions $f_\alpha$ and a set of vector fields
$V_i\in Vect(E)$. We define the ideal $I^n(f,V)\subset \cC^\infty(E)\,,\,\, n \geq0$ to be the ideal generated by
\begin{align}
    V_{i_1}\dots V_{i_k}f^\alpha\,  \quad 0\leq k \leq n\,.
\end{align}
We assume that $I^n$ do not restrict the base space.\footnote{This can be formulated as the requirement that the image of the zero locus of $I^n(f,V)$ under the canonical projection $E \to T[1]X$ is $T[1]X$.} It is also useful to introduce $I^\infty(f,V)$ which is generated by $V_{i_1}\dots V_{i_k}f^\alpha$ with $k=0,1,2,\ldots$.   
\end{definition}
Note that $I^0(f,V)$ does not depend on $V$ and hence we also denote it as $I^0(f)$. There is a natural filtration on $I^\infty(f,V)$:
\begin{equation}
I^0(f)\subset I^1(f,V) \subset I^2(f,V)\subset \ldots
\end{equation}
Note that $V_jI^m(f,V)\subseteq I^{m+1}(f,V)$.

If $(E,Q)$ is the off-shell  GR introduced above, we say that $I^n(f,D_a)$ is compatible if $(Q-\xi^a D_a)I^n(f,D_a)\subset I^n(f,D_a)$. It is easy to check that if $(Q-\xi^a D_a)I^0(f)\subset I^0(f)$ then all $I^n(f,D_a)$ are also compatible. 

\begin{prop}\label{prostoe2}
Let $f_\alpha\in C^{\infty}(E)$, $\gh{f_\alpha}=0$ be such that $(Q-\xi^a D_a)I^0(f)\subset I^0(f)$ and 
$r:\Ered\hookrightarrow E$. It follows
\begin{align} 
    r^{*}I^n(f,D_a)=I^n(r^{*}f,\nabla_a),
\end{align}
where $I^n(r^{*}f,\nabla_a)$ is the ideal in $\cC^\infty(\Ered)$ determined (in the sense of Definition~\bref{def-I}) by functions $r^*f^\alpha$ and vector fields $\nabla_a\equiv \commut{Q}{\dl{\xi^a}}$. Recall, that we assume that the coordinate system on  $\Ered$ in the degree 1 is fixed to be $\xi^a,D_a\xi^b,\lambda,D_a\lambda$ pulled-back to $\Ered$.
\footnote{This defines $\dl{(r^*\xi^a)}$ in the unambiguous way}. 
\end{prop}
Note that the condition of the statement is equivalent to requiring $\Delta, \Delta_a{}^b,\Gamma^a$ to preserve $ I^0(f)\subset I^0(f)$.
\begin{proof}
The statement is obvious for $n=0$. 
For $n=1$ we have
\begin{align}
    Qf=\xi^{a}D_{a}f+I^0(f).
\end{align}
Applying  $\dfrac{\partial}{\partial\xi_b}r^{*}$ to the above we get
\begin{align}
\label{1-relation}
    \nabla_{b}r^{*}f=r^{*}D_b f+I^0(r^{*}f),
\end{align}
where we used $\nabla_b f\equiv\Big[\dfrac{\partial}{\partial\xi^b},Q\Big]r^{*}f=\dfrac{\partial}{\partial\xi^b}Qr^{*}f$  which in turn employs $\gh{f}=0$. Equation \eqref{1-relation} implies that $r^*(I^1(f,D_a))$ can be generated by 
$r^*f^\alpha$ and $\nabla_{a}r^{*}f^\alpha$ and hence the statement for $n=1$. 

Proceeding further by induction we assume that the statement holds for all $i<N$, i.e.
\begin{align}
   \nabla_{a_1}\dots \nabla_{a_{i}} r^{*}f=r^{*}D_{a_1}\dots D_{a_{i}}f+I_{i-1}(r^{*}f,\nabla_a), \quad i<N.
\end{align}
Applying $D_{a_2}\dots D_{a_{N}}$ to $(Q-\xi^a D_a)I^0(f)\subset I^0(f)$ and using $D_aQ=QD_a$ and  $D_a I^{n-1}(f,D_a)\subseteq I^n(f,D_a)$) one has:
\begin{align}
    QD_{a_2}\dots D_{a_N}f-\xi^{a_1}D_{a_1}\dots D_{a_N}f\in I^{N-1}(f,D_a).
\end{align}
Finally, applying $\dfrac{\partial}{\partial \xi^{a_1}}r^{*}$ gives
\begin{align}
    \nabla_{a_1}r^{*}D_{a_2}\dots D_{a_N}f-r^{*}D_{a_1}\dots D_{a_N}f\in r^{*}I^{N-1}(f,D_a)
\end{align}
which together with the induction assumption leads to
\begin{align}
    \nabla_{a_1}\dots\nabla_{a_N}r^{*}f=r^{*}D_{a_1}\dots D_{a_N}f+I^{N-1}(r^{*}f,\nabla_a)\,,
\end{align}
so that
$r^{*}I^{N}(f,D_a)=I^{N}(r^{*}f,\nabla_a)$. 
\end{proof}

\begin{prop}\label{prostoe-3}
    Let $f_\alpha\in \cC^\infty(E)$, $\gh{f_\alpha}=0$ be such that $(Q-\xi^aD_a)I^0(f)\subseteq I^0(f)$. Then
    \begin{align}
    \label{commut-preserve}
        [\nabla_{a},\nabla_{b}]I^n(f,\nabla_c)\subseteq I^n(f,\nabla_c).
    \end{align}
    In particular, $I^\infty(f,\nabla_a)$ is generated by $\nabla_{(a)} f^\alpha$.
\end{prop}
\begin{proof}
Let $C^M$ denote all the degree 1 coordinates but $\xi^a$ and $\Delta_M\equiv[\frac{\partial}{\partial C^M},Q]$. Note that $\commut{\dl{C^M}}{\xi^a\nabla_a}=-\xi^a\commut{\dl{C^M}}{\nabla_a}$ preserves $I^n(f,\nabla)$  because any vector field of degree $-1$ does so. It follows, $\Delta_M$ preserves the ideal 
as well because $\commut{\dl{C^M}}{Q-\xi^a\nabla_a}$ clearly does. Finally, according to \eqref{commutrelations}, $\commut{\nabla_a}{\nabla_b}$ is proportional to $\Delta_M$ so that~\eqref{commut-preserve} and 
    \begin{align}
        \nabla_{a_1}\dots \nabla_{a_m}I^n(f,\nabla_c)=\nabla_{(a_1}\dots \nabla_{a_m)}I^n(f,\nabla_c)+I^{m+n-1}(f,\nabla_c)\,.
    \end{align}
\end{proof}
\begin{cor}\label{cor-Omega}
Functions $\nabla_{(a)}(r^*\Omega)$ can be taken as a part of a coordinate system on $\Ered$. 
\end{cor}
Indeed, because $\Ecl$ is a bundle over $E^{conf}$ and the equivalent reduction $r:\Ered \to \Ecl$ does not  affect the fibres, $r^*(D_{(a)}\Omega)$ are coordinates on the fibres of $\Ered \to E^{conf}_{\red}$. Thanks to \bref{prostoe2}, the restriction of $r^*(I^\infty(\Omega,D_a))$ to a fibre is isomorphic to the restriction of $I^\infty(r^*(\Omega),\nabla_a)$ to the same fibre. Moreover, thanks to Proposition \bref{prostoe2} the ideals generated by $r^* D_{(a)}(\Omega-1)$ and $\nabla_{(a)}r^*(\Omega-1))$ are the same and so are their restrictions to a given fibre. This implies that $\nabla_{(a)} r^* \Omega$ restricted to the fibre define a coordinate system therein.

Using Proposition \bref{prostoe2} we can also characterize more explicitly the pullback $\cI_{\red}$ to $\Ered$ of the ideal $\cI_{cl}$ generated by constraints~\eqref{Frho-def} which, in turn, are  
determined by the Einstein equations. More precisely, 
\begin{cor}\label{ideal-bulk}
On $\Ered$ we have:
\begin{align}
    r^{*}G_{ab}&=\nabla_a\nabla_br^{*}\Omega+r^{*}(\rho g_{ab}),\\
    r^{*}\rho&=-\frac{1}{D}\nabla^{a}\nabla_a r^* \Omega\,
\end{align}
and hence
\begin{align}
        \cI_{\red}=r^{*}(\cI_{cl})=I^0(r^{*}S)\cup I^\infty(r^{*}G_{ab},\nabla_a)\,,
\end{align}
\end{cor}
\begin{proof}
To prove it is sufficient to observe that $r^{*}\Gamma_{ab}^{c}=r^{*}P_{ab}=0$ by the definition of $\Ered$ and that $r^{*}D_{a}D_{b}\Omega = \nabla_a \nabla_b r^{*}\Omega$.
\end{proof}
It was shown in~\cite{Boulanger:2004eh,Boulanger:2004zf} that $\nabla_{(a_1\ldots a_k}  \We^b{}_{a_{k+1}a_{k+2})c}$ for $k\geq 0$ can be taken as coordinates on $\Ered$ in the sector of $W$ variable. In fact this can be also  inferred as a Corollary of Proposition~\bref{prostoe-3}. All in all we have the following
\begin{prop}\label{prop:Ered-coord}
As coordinates on $\Ered$ in the sector of degree $1$ variables one can take $\xi^a,C_a{}^b,\lambda,\lambda_a$ pulled back to $\Ered$ while in the sector of degree $0$ one can take:
\begin{equation}
g_{ab}\,,\qquad 
\nabla_{(a_1\ldots a_k} \We^b{}_{a_{k+1}a_{k+2})c}\,, \qquad \nabla_{(a)}(r^*\Omega)\,, \qquad k\geq 0\,.
\end{equation}
In this coordinate system the action of $Q$ is determined by~$\commut{Q}{\nabla_a}=0$, $r^*\circ Q=Q\circ r^*$, relations~\eqref{gran-ish1}, 
\begin{align}\label{W-counted}
    Q  \We^{c}{}_{bad}=\xi^e\nabla_{e} \We^{c}{}_{bad}-C_{e}{}^{c} \We^{e}{}_{bad}+C_{b}{}^{e}\We^{c}{}_{ead}+C_{a}{}^{e} \We^{c}{}_{bed}+C_{d}{}^{e} \We^{c}{}_{bae}\,,
\end{align}
and
\begin{equation}
Q(r^*\Omega)=\xi^a\nabla_a  (r^*\Omega)+\lambda (r^*\Omega)\,.
\end{equation}
\end{prop}
Note that to determine explicitly the action of $Q$ in terms of this coordinate system it is useful to employ relation~\eqref{commutrelations}.
Recall that by some abuse of notations we systematically omit $r^*$ in the notations for degree $1$ coordinates, e.g. $\xi^a$ in the above Proposition mean $r^*(\xi^a)$, and denote by $Q$ the $Q$ vector field on $E$ and its restriction to $\Ered$ as well. In what follows we also omit $r^*$ in $r^*(\Omega)$, and $r^*(G_{ab})$, $r^{*}S$ as we will not employ coordinates on the initial $\Ecl \supset \Ered$.

Finally, starting from  $Q^{2}C_{a}{}^{b}=0$ and $Q^{2}\lambda_a=0$ one can obtain the following ``Bianchi identities":
\begin{align}\label{Bianchi}
    \nabla_{[a}\We_{bc]de}-\Co_{d[ab}g_{c]e}+\Co_{e[ab}g_{c]d}=0, \quad \nabla_{[b}\Co^{b}{}_{cd]}=0\,,
\end{align}
where $\Co_{bcd}\equiv - \frac{1}{D-3}\nabla_a W^{a}{}_{bcd}$.  Thanks to these identities, we can express $\nabla_{(a)}\We_{bcde}$ in terms of the coordinates $\nabla_{(a_1\ldots a_k} \We^b{}_{a_{k+1}a_{k+2})c}$ introduced above. Equivalently, the
algebra of functions on $\Ered$
can be identified as the quotient of the algebra generated by ghosts and 
$\nabla_{(a)}\We_{bcde}$ by the relations generated by the $\nabla_{(c)}$-prolongations of the above Bianchi identities.

\section{Boundary system}\label{sec: boundary system}
\subsection{Natural gPDEs on the boundary}\label{subsec: adapted}
Until now we studied gravity on the manifold $X$ with or without boundary and have equivalently reformulated it as a gauge PDE $(\Ered,Q,\cI_{\red})$ subject to the additional condition $\Omega>0$. The important point is that negative powers of $\Omega$ enter neither $Q$ nor the generators of $\cI_{\red}$. This means that one can consistently extend $\Ered$ by taking the range of $\Omega$ to be $\Omega\geq 0$ rather than $\Omega>0$. In other words, from now on we add the boundary at $\Omega=0$ to our total space $\Ered$ and keep the same notation for it. Moreover,  we assume that $X$ is a manifold with a boundary and $\Sigma\equiv \d X$ is its boundary.

In this way we are in the setting of the general discussion of Section~\bref{sec:gpde-bound}. In particular, restricting $\Ered$ to the interior of $X$ and at the same time restricting to the interior of the fibres, i.e. requiring $\Omega>0$, gives precisely the equivalent reformulation of GR on $Int(X)$. At the same time the induced gPDE on the $Q$-corner (i.e. the $Q$-boundary of the restriction of $\Ered$ to $T[1]\Sigma$) is the gPDE that describes the boundary structure.  

Before passing to the gPDE on the boundary that describes the boundary structure it is convenient to describe in more details the pullback of $(\Ered,Q,\cI_\red)$ to the boundary $T[1]\Sigma$. This gPDE plays an important technical role in our analysis and we denote it by $(\Ee,Q,\cI)$ in what follows. Let us stress that this gPDE is defined over the boundary $T[1]\Sigma$ and the range of the fibre coordinate $\Omega$ is $\Omega \geq 0$. 

From now on we only encounter gPDEs defined over the boundary $T[1]\Sigma$ and hence it is legitimate to systematically exclude $T[1]\Sigma$ from the notations. Moreover, unless otherwise specified, manifolds, bundles, etc. are always bundles, usually $Q$-bundles, over $T[1]\Sigma$. For instance, a bundle $\pi:E^\prime \to E$ means that both $E$ and $E^\prime$ are bundles over $T[1]\Sigma$ and $\pi$ respects the bundle structure. In addition, defining a coordinate system on $E$ we often do not explicitly mention coordinates originating from the base space as these are always taken to be $x^\mu,\theta^\mu$, $\mu=0,\ldots D-2$.  Of course, the fibre coordinates are inherited from $\Ered$ and hence the notations are unchanged.

Because our aim is to study the gauge PDE induced by the conformal-like GR on the boundary it is useful to work in the basis adapted to this setup. More precisely, all the relations of the previous section were invariant under the general linear transformations acting on all the coordinates with indices $a,b,\ldots$ as the respective tensors~\footnote{This can be naturally promoted to local $gl(D)$-symmetry but we do not need this now.} Now we want to treat one direction in the underlying linear space as the distinguished one  and introduce the adapted basis where the components of e.g. $\xi^a$ split into $\xi^\Omega$ and $\xi^A$, $A=0,1,\ldots D-2$. For any (coordinate) function $f$ we denote $f^{(N)}\equiv (\nabla_\Omega)^N f$, in particular $f^{(0)}\equiv f$. 
For the future use it is convenient to introduce the decomposition of the ideals of the type described in~\bref{prostoe-3}, which is adapted to the decomposition $\xi^{\Omega},\xi^A$. Let $I^\infty(f,\nabla)$ be an ideal generated by  $\nabla_a\nabla_b \ldots f^{\alpha}$ and such that $(Q-\xi^a\nabla_a) I^0(f) \subset I^0(f)$ then 
 \begin{equation}
I^k (f,\nabla_a)\subset I^{(k)}(f,\nabla_A), \qquad 
I^{(N)}(f,\nabla_A)\equiv \bigcup_{i=0}^{N}I^\infty(f^{(i)},\nabla_A)\,.
\end{equation}
Indeed, $\commut{\nabla_a}{\nabla_b} I^k (f,\nabla_c) \subset I^k (f,\nabla_c)$ and hence any generator of $I^k (f,\nabla_c)$ can be represented as a linear combination of generators of $I^{(k)}(f,\nabla_A)$ by taking all $\nabla_\Omega$ to the right.

Another ingredient we need in what follows is the degree $-1$ vector fields $\nu^{(N)}_{a}$ defined in $\Ee$ as
\begin{equation}
\label{nu-def}
\nu^{(N)}_{a}\equiv ad_{\nabla_\Omega}^N\Big(\frac{\partial}{\partial\xi^{a}}\Big)\equiv \Big[\nabla_{\Omega},\dots,\Big[\nabla_{\Omega},\frac{\partial}{\partial\xi^{a}}\Big]\Big],\quad \nu^{(0)}_{a}=\frac{\partial}{\partial\xi^{a}}
\end{equation}
Note that $\nu^{(N)}_{\Omega}=0$ if $N>0$. It is also useful to define
\begin{equation}
    \mathcal{D}_{a}^{(N)}\equiv \commut{Q}{\nu^{(N)}_{a}}= ad^N_{\nabla_\Omega}(\nabla_a)=[\nabla_{\Omega},\dots,[\nabla_{\Omega},\nabla_a]]\,.
\end{equation}

Vector fields $\mathcal{D}_{a}^{(N)}$ can be used to represent $[\nabla^{N}_{\Omega},\nabla_{A}]$. More precisely, 
\begin{prop}\label{utv-commutator}
\begin{align}
[\nabla^{N}_{\Omega},\nabla_{A}]=\sum_{i=0}^{N-1}C^{i}_{N}\mathcal{D}^{(N-i)}_{A}\nabla^{i}_{\Omega},
\end{align}
where $C^{i}_{N}=\frac{N!}{(N-i)!i!}$ are the binomial coefficients.
\end{prop}
\begin{proof}
    Proved by induction starting from:
    \begin{align}
            [\nabla^{2}_{\Omega},\nabla_{A}]=\nabla_{\Omega} [\nabla_{\Omega},\nabla_{A}] + [\nabla_{\Omega},\nabla_{A}]\nabla_{\Omega}=\mathcal{D}^{(2)}_{A}+2\mathcal{D}^{(1)}_{A}\nabla_{\Omega}.
    \end{align}
    Suppose $[\nabla_{\Omega}^{N},\nabla_{A}]=\sum_{i=0}^{N-1}C^{i}_{N}\mathcal{D}^{(N-i)}_{A}\nabla^{i}_{\Omega}$ for a given $N$. Then
    \begin{multline}
        [\nabla_{\Omega}^{N+1},\nabla_{A}]=\nabla_{\Omega}[\nabla_{\Omega}^{N},\nabla_{A}]+[\nabla_{\Omega},\nabla_{A}]\nabla^N_{\Omega}=
        \\
        =\sum_{i=0}^{N-1}C^{i}_{N}\mathcal{D}^{(N-i+1)}_{A}\nabla_{\Omega}^{i}+\sum_{i=0}^{N-1}C^{i}_{N}\mathcal{D}_{A}^{(N-i)}\nabla_{\Omega}^{i+1}+\mathcal{D}_{A}^{(1)}\nabla_{\Omega}^{N}=\\
        =\mathcal{D}_{A}^{(N+1)}+\sum_{i=1}^{N-1}(C^{i}_{N}+C^{i-1}_{N})\mathcal{D}_{A}^{(N+1-i)}\nabla_{\Omega}^{i+1}+(N+1)\mathcal{D}_{A}^{(1)}\nabla_{\Omega}^{N}=\sum_{i=0}^{N}C^{i}_{N+1}\mathcal{D}^{(N+1-i)}_{A}\nabla^{i}_{\Omega}
    \end{multline}
\end{proof}
In particular, for a general $f \in \cC^{\infty}(\Ee)$ one has:
\begin{align}\label{decomp-nablaomega}
    \nabla_{\Omega}^{N}\nabla_{A}f=\nabla_Af^{(N)}+[\nabla_{\Omega}^{N},\nabla_A]f=\nabla_{A}f^{(N)}+\sum_{i=0}^{N-1}C^{i}_{N}\mathcal{D}^{(N-i)}_{A}f^{(i)}=\sum_{i=0}^{N}C^{i}_{N}\mathcal{D}^{(N-i)}_{A}f^{(i)}\,.
\end{align}
Finally, let us give the recursive formula for $\nu_{A}^{(N)}$:

\begin{multline}\label{nu-counted-bulk}
            \nu^{(N)}_{A}=-(\nabla_\Omega^{N-1}\We^{b}{}_{c\Omega A}+(N-1)\mathcal{P}^{db}_{c\Omega}\nabla^{N-2}_{\Omega}\Co_{d\Omega A})\frac{\partial}{\partial C_{c}{}^{b}}-\\-\nabla_\Omega^{N-1}\Co_{d\Omega A}\frac{\partial}{\partial\lambda_d}+(N-1)\nabla_\Omega^{N-2}\Co_{\Omega\Omega A}\frac{\partial}{\partial \lambda}+\\+\sum_{i=1}^{N-1}C^{i}_{N-1}(\nabla_{\Omega}^{N-i-1}\We^{b}{}_{\Omega\Omega A}+\frac{N-i-1}{i+1}\nabla_{\Omega}^{N-i-2}\mathcal{P}^{db}_{\Omega\Omega}\Co_{{d\Omega A}})\nu^{(i-1)}_b, \quad N\geq1\,,
\end{multline}
where $\mathcal{P}^{d e}_{a b}\equiv (-g^{de}g_{ab}+\delta^{d}_{a}\delta^{e}_{b}+\delta^{d}_{b}\delta^{e}_{a})$.

\subsection{Boundary GR}
The induced gPDE on the boundary is the sub-gPDE of $\Ee\equiv i^* \Ered$, $i:T[1]\Sigma \hookrightarrow T[1]X$ singled out by the following conditions:
\begin{align}\label{OmegaQOmega}
    \Omega=0, \qquad Q\Omega=0\,,
\end{align}
along with the extra requirement $\nabla_a \Omega\neq0$. This is of course the realisation of the Penrose notion of asymptotic boundary within the present approach. Conditions \eqref{OmegaQOmega} single out the $Q$-boundary of the fibre of $\Ee$ while  $\nabla_a \Omega\neq0$ ensures that $\Omega$ can be promoted to an eligible boundary-defining function.\footnote{More precisely, if $\sigma$ is a solution to the off-shell gPDE $(\Ered,Q)$ then this condition indeed ensures that $\sigma^*(\Omega)$ has a nonvanishing gradient at the boundary.}

It is also convenient to equivalently reduce the resulting gPDE by restricting to the subbundle $\Ee_B\subset \Ee$ singled out by \eqref{OmegaQOmega}, $\nabla_a \Omega\neq0$ and
\begin{equation}
\begin{aligned}
  \Omega^{(1)}
    &= g_{\Omega\Omega}\,,
  &\qquad
  C_{\Omega}{}^{\Omega}
    - g^{\Omega\Omega}\xi^{b}\nabla_b\Omega^{(1)}
    + \lambda
    &= 0\,,
  \\
  \nabla_A\Omega
    &= 0\,,
  &
  C_{\Omega}{}^{A}
    - \xi^{b}\nabla_b\nabla^{A}\Omega
    &= 0\,,
  \\
  \Omega^{(2)}
    &= 0\,,
  &
  \lambda_\Omega
    + g^{\Omega\Omega}\xi^a\nabla_a\Omega^{(2)}
    &= 0\,,
  \\
  g_{\Omega A}
    &= 0\,,
  &
  C_A{}^{\Omega}
    + \frac{1}{g_{\Omega\Omega}}
      \xi^{b}\nabla_b\nabla_A\Omega
    &= 0\,.
\end{aligned}
\label{boundcond}
\end{equation}
The constraints listed in the right column are obtained by acting with $Q$ on the ones in the left and then using \eqref{OmegaQOmega}
and the left column. This equivalent reduction is the minor modification of that from~\cite{Grigoriev:2023kkk} to which we refer for further details.  Let us note that $\xi^\Omega$ reduced to $E_B$ vanishes.

It is convenient to identify the resulting gPDE $E_B$ as a $Q$-submanifold in $\Ee$. More precisely, the algebra of functions on $E_B$ can be seen as a quotient of $C^\infty(\Ee)$ by the ideal $\cK_B$ generated by the left hand sides of  \eqref{OmegaQOmega}, \eqref{boundcond}. 
 In other words, if $b$ denotes the embedding $b:E_B \hookrightarrow \Ee$ then $\cK_B=\ker b^*$.

\subsection{Induced Einstein equations}\label{sec:induced-Einstein}
We now recall that what we are interested in is not  $\Ee$ or $E_{B}$ but rather the  sub-gPDE 
$\cE\subset E_B\subset \Ee$ singled out by the restriction of $\cI$ to $E_B$, where $\cI$ originates from the Einstein equations. We denote by 
$\cK \equiv \cI \cup \cK_B\subset \cC^\infty(\Ee)$ the ideal of functions vanishing on $\cE\subset \Ee$. This section is devoted to studying the structure of $\cK$ and identifying a suitable set of its generators.

It is convenient to introduce filtration components of $\cK$ as follows:
\begin{equation}\label{I-filtr}
\cK^{(N)}=I^{(N)}(G_{ab},\nabla_{C}) \cup I^{0}(S)\cup \cK_B\,, 
\end{equation}

where $\cK_B$ denotes the ideal of functions on $\Ee$ vanishing on $E_B \subset \Ee$.
As $(N)$ refers to $(\nabla_{\Omega})^N$, the ideal 
 $\cK^{(0)}$ can be understood as
 $\cK_B$ extended by the leading part of the Einstein equations in the sense of near-boundary expansion.  We have the following:
\begin{lemma}\label{lemma-I0}
\begin{enumerate}
    \item Vector fields $\dl{\xi^A},\dl{C_{A}{}^{B}}$ and $\dl{\lambda_{B}}$ on $\Ee$ preserve $\cK^{(0)}$. 
    \item $Q$ preserve $\cK^{(0)}$.

    \item $\cK^{(0)}$ is generated by $\cK_B$ and the following functions:
    \begin{align}
       \nabla_{(A)}\Omega, \quad  \nabla_{(A_1\ldots A_n B)}\Omega^{(1)}, \quad  \nabla_{(A)}\Omega^{(2)}\,,\quad \nabla_{(A_1}\dots\nabla_{A_n}J^{B}{}_{C)D}\,, \quad 
       g_{\Omega\Omega}-\tilde{\Lambda}\,,
    \end{align}
where $J^{B}{}_{CD}\equiv \We^{B}{}_{C\Omega D}$, $\Omega^{(k)}=(\nabla_\Omega)^k \Omega$, and $n\geq 0$.
\end{enumerate}
\end{lemma}
In the Lemma and in what follows we use the rescaled cosmological constant $\tilde\Lambda$ defined as:\footnote{Note that in the standard AdS/CFT notation, $l^{-2}\equiv -\frac{2\Lambda}{(D-1)(D-2)}$, and thus $\tilde\Lambda\equiv -l^{-2}$.}
\begin{equation}
\tilde{\Lambda}\equiv\frac{2\Lambda}{(D-1)(D-2)}
\end{equation}

\begin{cor}\label{zam-delta}
1) and 2) immediate imply that $\nabla_{A}\equiv \Big[Q,\dfrac{\partial}{\partial\xi^{A}}\Big]$, $\Gamma^A\equiv \Big[Q,\dfrac{\partial}{\partial \lambda_{A}}\Big]$, $\Delta^{A}{}_{B}\equiv \Big[Q,\dfrac{\partial}{\partial C_{A}{}^{B}}\Big]$ preserve $\cK^{(0)}$.
\end{cor}
\begin{proof}

To prove 1) let us note that $\dl{\xi^A},\dl{C_{A}{}^{B}}$ and $\dl{\lambda_{B}}$ are of degree $-1$ it is enough to show that these vector fields preserve  $\cK_B$ modulo $\cK^{(0)}$ because $I^{(0)}(G_{ab},\nabla_{C}) \cup I^{0}(S)$ is generated by degree $0$ functions. In its turn $\cK_B$ is generated by the left-hand sides of \eqref{OmegaQOmega}, \eqref{boundcond}. Acting on them by $\dl{\xi^A}$ one readily finds that the result belongs to $\cK^{(0)}$. For instance,
    \begin{align}
         \frac{\partial}{\partial\xi^{A}}(C_{\Omega}{}^{\Omega}-\xi^{b}\nabla_b \Omega^{(1)}+\lambda)=-\nabla_A \Omega^{(1)}\in \cK^{(0)}
    \end{align}
In a similar way one proves the statement for $\frac{\partial}{\partial C_{A}{}^{B}}$, $\frac{\partial}{\partial \lambda_A}$.

To prove 2) we need to show that $Q$ preserves $\cK^{(0)}$. Because $Q$ preserves $\cK_B$ it is enough to check $Q$-invariance of $I^{(0)}(G_{ab},\nabla_C)\cup I^0(S)$ up to $\cK_B$. First, a straightforward computation shows that $\nabla_A$ preserves $\cK^{(0)}$. Then
\begin{align}
    Q \nabla_{A_1}\dots\nabla_{A_n}G_{bc}=\nabla_{A_1}\dots\nabla_{A_n}(QG_{bc})
\end{align}
In fact $QG_{bc} \in \cK^{(0)}$ because $(Q-\xi^{a}\nabla_a)G_{bc}\subset I^{0}(G_{bc})$ and $\xi^\Omega \Omega^{(1)} \in \cK_B$. Analogously, one has $QS\in \mathcal{K}^{(0)}$.

To prove 3)  let us first show that using the freedom of adding functions from $\cK_B$ the generators of the ideal $I^{(0)}(G_{bc},\nabla_A)$ can be brought to a simple form. Namely, using $g_{\Omega A}\in \cK_B$ and $G_{bc}=\nabla_b\nabla_c\Omega+g_{bc}\rho$, where $\rho=-\frac{1}{D}(\nabla_\Omega\nabla^{\Omega}\Omega+\nabla_A\nabla^A\Omega)$ one gets:
\begin{align}
&\nabla_{(C)}G_{\Omega\Omega}=\nabla_{(C)}\Big(\frac{D-1}{D}\Omega^{(2)}-\frac{g_{\Omega\Omega}}{D}\nabla_A\nabla^A\Omega\Big)+\cK_B,\\
&\nabla_{(C)}G_{A\Omega}=\nabla_{(C)}\nabla_A \Omega^{(1)}+\cK_B ,\\
&\nabla_{(C)}G_{AB}=\nabla_{(C)}\Big(\nabla_A\nabla_B\Omega-\frac{g_{AB}}{D}(\nabla_\Omega\nabla^{\Omega}\Omega+\nabla_D\nabla^D\Omega)\Big)+\cK_B
\end{align}
Using the first equation in the last one and employing \eqref{commutrelations}, from which it follows that
\begin{align}
    [\nabla_C, \nabla_{A}\nabla^{A}]\Omega=[\nabla_C,\nabla_A]\nabla^{A}\Omega=W^{A}{}_{bCA}\nabla^b\Omega-C^{A}{}_{CA}\Omega\in\cK_B
\end{align}
one finds that $\cK^{(0)}$ can be generated by $\cK_B$, $S$ and  
\begin{align}\label{lemma-I0-perepis}
\nabla_{A_1}\dots\nabla_{A_n}\Omega^{(2)}\,, \quad \nabla_{A_1}\dots\nabla_{A_n}\nabla_B\Omega^{(1)}\,, \quad \nabla_{A_1}\dots\nabla_{A_n}\nabla_B\nabla_C\Omega\,.
\end{align}
Then we employ a version of Proposition~\bref{prostoe-3} to show that  generators \eqref{lemma-I0-perepis} can be replaced with
    \begin{align}
\nabla_{(A_1}\dots\nabla_{A_n)}\Omega^{(2)},\qquad  \nabla_{(A_1}\dots\nabla_{A_n)}\nabla_B \Omega^{(1)}, \qquad \nabla_{(A_1}\dots\nabla_{A_n)}\nabla_B\nabla_C\Omega\,.
\end{align}
It turns out that the ``hook" component of the $\nabla_{(A)}\nabla_B \Omega^{(1)}$ does not give new generators as 
\begin{equation}
[\nabla_{A}, \nabla_B]\Omega^{(1)}=-\We^{c}{}_{\Omega A B}\nabla_c\Omega-\Co_{\Omega A B}\Omega \in \cK_B.
\end{equation}
 In a similar way
\begin{align}
   [\nabla_{A}, \nabla_B]\nabla_C\Omega=-\We^{d}{}_{CAB}\nabla_d\Omega-\Co_{CAB}\Omega=-J_{ABC}+\cK_B\,,
\end{align}
giving the following addition generators:
$\nabla_{(A_1}\dots\nabla_{A_{n-1}}J^{B}{}_{A_n)C}$.
\end{proof}

The reason why the ideal $ \cK^{(0)}$ turns out to be extremely useful is the following:
\begin{prop}\label{mnogoidealov}
Let $f_\alpha\in \cC^{\infty}(\Ee)$ be such that  
    \begin{align}
    \label{Qcompf}
        (Q-\xi^a\nabla_a)f_\alpha\in I^0(f_\alpha).
    \end{align}
    Then 
    \begin{align}
        QI^{(N)}(f_\alpha,\nabla_{A})\subset I^{(N)}(f_\alpha,\nabla_{A})\cup  \cK^{(0)}\,,
    \end{align}
recall that  $I^{(N)}(f_\alpha,\nabla_A)\equiv \cup_{i=0}^N I(f^{(i)}_\alpha,\nabla_A)$, $f_\alpha^{(i)}\equiv (\nabla_{\Omega})^{i}f_\alpha$.
\end{prop}
In other words the  ideals $I^{(N)}(f_\alpha,\nabla_A)$ are $Q$-invariant up to $\cK^{(0)}$ or, equivalently, $Q$ is tangent to the zero locus of $I^{(N)}(f_\alpha,\nabla_A) \cup \cK^{(0)}$.

\begin{proof}
    Applying $\nabla_\Omega^N$ to \eqref{Qcompf} gives:
    \begin{align}
    Qf^{(N)}-\xi^\Omega f^{(N+1)} \in I^{(N)}(f_\alpha,\nabla_A).
    \end{align}
    and hence
    \begin{align}
        Qf^{(N)}\in I^{(N)}(f_\alpha,\nabla_A)\cup \cK_B\,.
    \end{align}
Acting on this with $\nabla_{A_1}\dots\nabla_{A_n}$ and using the Corollary \bref{zam-delta}, we obtain
\begin{align}
        Q\nabla_{A_1}\dots \nabla_{A_n} f_\alpha^{(N)} \in I^{(N)}( f_\alpha,\nabla_A)\cup \cK^{(0)}\,.
    \end{align}
\end{proof}

We now take into account the remaining components of the equations of motion, which are encoded in $\cK^{(N)}$, $N>0$. A direct consequence of Proposition \bref{mnogoidealov} is that $Q\cK^{(N)}\subset \cK^{(N)}$, $N\geq0$.  Also we have the following important:
\begin{thm}\label{prop:mnogoidealov}{\rm(Subleading Einstein equations)}

\begin{enumerate}
    \item Vector fields $\nu_A^{(N)}$ defined on $\Ee$ by \eqref{nu-def}  preserve $\cK^{(N+1)}$
    \item The ideal $\cK^{(N)}$ is generated by $\cK^{(0)}$ and the following functions: 
    \begin{gather}
\nabla_{(C)}\Omega^{(i+2)},\qquad \nabla_{(C)}  O_{AB}^{(i-1)}\,, \quad\ i\leq N \,,\\
\nabla_{(C)} O^{(D-3)}_{A} \quad \text{(present for $N\geq D-1$ only!)}\,.
    \end{gather}
\end{enumerate}
where 
\begin{equation}\label{O-constraints}
\begin{aligned}
    & O_A^{(j)}\equiv \nabla_\Omega^{j}(\nabla^{B}T_{BA}),\\
    &   O_{AB}^{(j)}\equiv (D-3-j)T_{AB}^{(j)}-j\tilde{\Lambda}\nabla_\Omega^{j-1}\nabla^{C}J_{CBA}\,.
\end{aligned}
\end{equation}
\end{thm}
Note that $\mathcal{D}^{(N)}_A\equiv [\nu_{A}^{(N)},Q]$ also preserves $\cK^{(N+1)}$ thanks to 1. 
\begin{proof}
As the theorem is mainly technical, the proof is given in Appendix~\bref{AppA}.
\end{proof}

\subsection{Off-shell boundary system}\label{sec:boundary-gPDE}

As a result of the previous sections, we obtain an infinite chain of ideals in $\cC^{\infty}(\Ee)$ associated with the $(i)$-subleading orders of the Einstein equations:
\begin{align}
0 \subset \cK^{(0)}\subset\cdots \subset \cK^{(i)} \subset \cK^{(i+1)} \subset \cdots\subset \cK.
\end{align}
According to Proposition \bref{mnogoidealov} all these ideals  are $Q$-invariant ideals. Consequently, we have a dual chain of $Q$-submanifolds of $\Ee$
\begin{align}
   \cE
    \hookrightarrow\cdots \hookrightarrow \Ee^{(i+1)}\hookrightarrow \Ee^{(i)}\hookrightarrow\cdots \hookrightarrow \Ee^{(0)}\hookrightarrow \Ee\,,
\end{align}
where $\Ee^{(i)}$ is the zero locus of $\cK^{(i)}$ and $\cE$ is the zero locus of the entire $\cK$. In other words, $\cE$ is the gPDE induced on the boundary by the asymptotically AdS gravity. In what follows we  denote as $b_i$ the embedding of $\Ee^{(i)}$ into $\Ee$.

It turns out that only $\Ee^{(i)}$ with $i\leq D-3$ are actually useful for our analysis. As for $\cE$, it is the zero locus of non-linear constraints and we are going to treat it as a sub-gPDE of certain quotient of $\Ee^{(D-3)}$. Moreover, this quotient turns out to be an off-shell gPDE and hence simplifies a lot the analysis of the field theory described by $\cE$. In particular, the explicit form of the equations of motion encoded in $\cE$ can be derived systematically.

As a first step, we provide a more detailed description of $\Ee^{(0)}$. Recall that  
$\Ee^{(0)}$ is a sub-gPDE in $\Ee$ singled out by the leading-order Einstein equations:
\begin{prop}\label{utv-coord-E0}
For $D\geq 5$,  as a coordinate system on $\Ee^{(0)}$ one can take the following functions on $\Ee$
    \begin{align}\label{form-coord-E0}
        \{\xi^{B},\enspace C_{B}{}^{C},\enspace \lambda,\enspace\lambda_{B},\qquad \nabla_{(A)}\Omega^{(N+3)},\enspace  g_{BC},\enspace \nabla_{((A)} \We^{B}{}_{CD)E},\enspace\nabla_{(A)}T_{BC}^{(N)},\quad  |A|\geq0, N\geq0\}
    \end{align}
    pulled back to $\Ee^{(0)}$. Here $b_0:\Ee^{(0)} \hookrightarrow \Ee$ and we use the notations $f^{(N)}\equiv \nabla_{\Omega}^{N}f$, $T_{BC}\equiv W_{\Omega B\Omega C}$. Moreover, $b_0^{*}T^{(N)}_{BC}$, $N\geq0$ are trace-free.
\end{prop}
Note that Lemma \bref{lemma-I0} implies that the vector fields $\dl{\xi^A},\dl{C_{A}{}^{B}}$ and $\dl{\lambda_{B}}$
 are tangent to $\Ee^{(i)} \subset \Ee$ for all $i\geq0$.  In particular, $\nabla_{A}\equiv[\frac{\partial}{\partial\xi^{A}},Q]$ is tangent to $\Ee^{(0)}$ and we keep denoting its restriction to $\Ee^{(0)}$ by $\nabla_{A}$. The coordinates in the degree-$0$ sector are generated by $\nabla_{(A)}$ applied to the basic coordinates and hence the structure of $Q$ is substantially simplified. 
\begin{proof}
The ghost degree one sector is obtained directly as the pullback of the ghost degree one coordinates from $\Ee$ and taking into account that functions $\xi^\Omega,C_{\Omega}{}^{A},C_{A}{}^{\Omega},\lambda_\Omega$ belong to $\cK^{(0)}$ and hence vanish on $\Ee^{(0)}$. Analogously: $b_0^{*}g_{\Omega A}=0, b_0^{*}g_{\Omega\Omega}=\tilde\Lambda$.

As for the $\Omega$-sector, notice that using Proposition~\bref{prostoe-3} we can replace  coordinates $\{\nabla_{(a)}\Omega\}$ on $\Ee$ with $\{\nabla_{(A)}\Omega^{(N)}\, ,N\geq 0 \}$. Furthermore, Lemma \bref{lemma-I0} implies that $\{\nabla_{(A)}\Omega^{(0)}, \nabla_{(A)}\Omega^{(1)}, \nabla_{(A)}\Omega^{(2)}\}$ belong to $\cK^{(0)}$.

In the curvature sector, things are somewhat trickier. To begin with, note that using the symmetries of the Weyl tensor, we can parametrize $\{W_{abcd}\}$ as $\{W_{ABCD},J_{ABC},T_{AB}\}$, where $J_{ABC}\equiv W_{AB\Omega C}$ and $T_{AB}\equiv W_{\Omega A\Omega B}$. 
Next, recall that Bianchi identities \eqref{Bianchi} have the form
\begin{align}\label{Bianchi-E0}
    \nabla_{[a}\We_{bc]de}-\Co_{d[ab}g_{c]e}+\Co_{e[ab}g_{c]d}=0
\end{align}
and that $C_{bcd}\equiv - \frac{1}{D-3}\nabla_a W^{a}{}_{bcd}$. Then we explicitly spilt the indexes into $\Omega,A$ and work up to the ideal $I(g_{\Omega A})\subset \cK_B$ generated by $g_{\Omega A}$. Note that $I(g_{\Omega A})$ is $\nabla_b$-invariant. If three indexes are $\Omega$ then all the terms in \eqref{Bianchi-E0} vanish identically. For one $\Omega$-index $a=\Omega$ we have
\begin{align}\label{proof-coordE0-1}
           \nabla_{\Omega}\We_{BCDE}= 
 \nabla_{B} J_{DEC}-\nabla_{C} J_{DEB}+6\Co_{D[\Omega B}g_{C]E}-6\Co_{E[\Omega B}g_{C]D}\,,
    \end{align}
where in view of $C_{bcd}\equiv - \frac{1}{D-3}\nabla_a \We^{a}{}_{bcd}$
\begin{align}\label{Cotton-T}
\begin{split}
&\Co_{\Omega\Omega B}=\frac{1}{D-3}\Big(\nabla_AT^{A}{}_{B}\Big)+I(g_{\Omega A})\, ,\\
    &\Co_{A \Omega B}=-\frac{1}{D-3}\left(g^{\Omega \Omega} T_{A B}^{(1)}+\nabla_C J^{C}{}_{A B}\right)+I(g_{\Omega A}),\\
       & \Co_{ABC}=-\frac{1}{D-3}\Big(g^{\Omega \Omega}J^{(1)}_{ABC}+\nabla_{D} \We^{D}{}_{ABC}\Big)+I(g_{\Omega A}).
       \end{split}
\end{align}
Note that upon substituting the expression for the Cotton tensor into \eqref{proof-coordE0-1}, the right-hand side of this relation does not involve $\nabla_\Omega W^{A}{}_{BCD}$ and can be interpreted as its expression in terms of the other variables.

Similarly, in the case where two $\Omega$-indexes are present we have
\begin{align}
        \nabla_{\Omega}J_{ABC}=\nabla_{A} T_{BC}-\nabla_{B}T_{AC}+6 \Co_{\Omega[AB}g_{\Omega]C}-6\Co_{C[AB}g_{\Omega] \Omega}\,,
\end{align}
giving 
\begin{align}\label{proof-coordE0-2}
            \frac{D-4}{D-3}\nabla_{\Omega}J_{ABC}=\nabla_{A} T_{BC}-\nabla_{B}T_{AC}+6 \Co_{\Omega[A B}g_{\Omega]C}+\frac{1}{D-3}\nabla_{D}\We^{D}{}_{CAB}+I(g_{\Omega A})\,.
\end{align}
For $D\geq5$ this allows us to express $\nabla_\Omega J_{ABC}$ in terms of the other variables. Note that since $\nabla_a$ preserves the ideal $I(g_{\Omega A})$, we can proceed analogously for all $\nabla_{(A)}\nabla_\Omega^{N}\We^{A}{}_{BCD}$ and $\nabla_{(A)}\nabla_\Omega^{N}J_{ABC}$ with $N\geq1$.

  The remaining part of the Bianchi identities can be employed to express all components in terms of the following symmetrized ones: $\nabla_{((A)}W^{B}{}_{C)DE}$ and $\nabla_{((A)}J^{B}{}_{C)D}$. Finally, according to Lemma~\bref{lemma-I0} we have  $\nabla_{(A_1\dots A_N}J^{B}{}_{C)D}\in \cK^{(0)}$
so that we indeed arrive at the desired coordinates in the curvature sector.

The vanishing of the trace of $b_0^{*}T^{(N)}_{BC}$ directly follows from the tracelessness of the original $W^{b}{}_{cde}$ and the fact that $g_{\Omega A}\in \cK^{(0)}$.
\end{proof}

From now on we restrict our general considerations to the case $D\geq 5$, treating the case $D=4$ separately in Section \bref{sec:odd}. Theorem~\bref{prop:mnogoidealov} implies that $ b_0^{*}\cI$ is generated by  $\nabla_{(C)}b_0^{*}O^{(D-3)}_A$ and 
\begin{align}
    \nabla_{(C)}b_0^{*} \Omega^{(j+3)}, \qquad \nabla_{(C)}\left((D-3-j)b_0^{*}T_{AB}^{(j)}-j\tilde{\Lambda}b_0^{*}\nabla_\Omega^{j-1}\nabla^{C}J_{CBA}\right),\quad j\geq0\,.
\end{align}
Comparing this with the coordinates on $\Ee^{(0)}$ given in Proposition \bref{utv-coord-E0}, we see that these generators (except for the second one with $j = D - 3$) can be solved with respect to $T^{(j)}_{AB}$. Strictly speaking, one should also make sure that $b_0^{*}\nabla_\Omega^{j-1}\nabla^{C}J_{CBA}$ does not depend on $b_0^{*}T^{(j)}_{AB}$. Although this fact is almost evident from a straightforward counting of the number of $\nabla_\Omega$ and $\nabla_A$ acting on the Weyl tensor, the rigorous proof requires a technically involved computation of $J^{(N)}_{ABC}$, which is presented in the Proposition~\bref{App-J-counted} in Appendix. Explicitly solving these ``simple'' equations from $\cK$ we arrive at:
\begin{definition}\label{def-hatE}
    Let us call $\hat \Ee$ the subbundle of $\Ee^{(0)}$ singled out by
    \begin{align}
       & \nabla_{(A)}\Omega^{(N+2)}=0\, , \quad \nabla_{(A)}\Big((D-3-j)T_{BC}^{(j)}-j\tilde{\Lambda}\nabla_\Omega^{j-1}\nabla^{D}J_{DCB}\Big)=0\, \quad j\neq D-3\,, N\geq0
       \end{align}
    pulled back to $\Ee^{(0)}$. The embedding $\hat \Ee \to \Ee$ is denoted by  $\hat b$ and for any function $f$ on $\Ee$ we define $\hat f\equiv \hat b^{*}f$. As a fibre coordinate system on $\hat \Ee$ we use
\begin{align}\label{hat-e-coord}
    \hat \xi^{B},\quad \hat C_{B}{}^{C},\quad \hat \lambda,\quad \hat \lambda_{B},\quad    \hat g_{BC},\quad \nabla_{((A)} \hat \We^{B}{}_{CD)E},\qquad \nabla_{(A)} \cT_{BC},\quad  |A|\geq0\,,
\end{align}
where  $\cT_{AB}\equiv \hat{b}^{*}T_{AB}^{(D-3)}$.
\end{definition}   

Note that from the tracelessness of the original $W^{a}{}_{bcd}$ and the fact that by definition of $\hat \Ee$ we have $\hat T^{(0)}_{AB}=0$, it follows that $\hat W^{A}{}_{BCD}$ is traceless.

There is an important subtlety to be clarified here: strictly speaking, $\hat{\Ee}$ is not a $Q$-submanifold in $\Ee^{(0)}$ since $Q$ preserves only the whole ideal $b_0^{*}\cK$. However, there is a natural $Q$-structure on $\hat \Ee$
and this is precisely the one employed in defining the coordinates \eqref{hat-e-coord}
through $\nabla_A$. It can be constructed as follows: as a first step, we pass from $\Ee^{(0)}$ to its sub-gPDE $\Ee^{(D-3)}$ by restricting to the zero locus of $b_0^{*}\cK^{(D-3)}$, which is a $Q$-subbundle. As the second step, we note that the following coordinate functions
\begin{align}
\label{funct-E-D-3}
    \xi^{B},\quad C_{B}{}^{C},\quad \lambda,\quad \lambda_{B},\quad    g_{BC},\quad \nabla_{((A)} \We^{B}{}_{CD)E},\qquad \nabla_{(A)}T_{BC}^{(D-3)},\quad  |A|\geq0
\end{align}
restricted to $\Ee^{(D-3)}$ form a subset of the coordinates therein.  Moreover, one can verify the subalgebra of functions generated by this subset is closed under $Q$. For all coordinates except $\nabla_{(A)}T_{BC}^{(D-3)}$ this follows directly from the $b_{D-3}$-pullback of the $Q$-action on $\Ee$ given in \eqref{gran-ish1}, \eqref{W-counted}. The case of the $Q$-action on $T_{AB}^{(D-3)}$ is technically more involved and relies on the results proved in the Appendix, see Propositions~\bref{App-QT} and~\bref{App-J-counted}. It follows, $\Ee^{(D-3)}$ is a $Q$-bundle over $\hat \Ee^\prime$ which is identified as a manifold whose algebra of functions is generated by
\eqref{funct-E-D-3} restricted to $\Ee^{(D-3)}$. Moreover, $\hat \Ee$
is a section of $\Ee^{(D-3)} \to \hat \Ee^\prime$ and hence $\hat \Ee$ and $\hat \Ee^\prime$ are diffeomorphic and in this way $\hat \Ee$ acquires the $Q$-structure. 
More geometrically, $\Ee^{(D-3)}$ can be considered as a bundle over $\hat \Ee$ by taking as the projection the composition of the projection to $\hat \Ee^\prime$ and the section map $\hat \Ee^\prime \to \hat \Ee$. Then the $Q$-structure on $\hat \Ee$ is just the projection that on $\Ee^{(D-3)}$. In other words, $Q f=i^*(Q\pi^*f)$ for any $f\in \cC^\infty (\hat \Ee)$, where $i$ denotes the embedding of $\hat \Ee \to \Ee^{(D-3)}$ and $\pi$ be a projection $\Ee^{(D-3)}\to \hat \Ee$. The way it is defined guarantees that the restriction of such $Q$ on $\hat \Ee$ to $\cE\subset \hat \Ee$ coincides with the restriction of $Q$ on $\Ee^{(D-3)}$ to $ \cE\subset \Ee^{(D-3)}$. It follows that implicit gPDE $(\hat \Ee,Q,\hat \cI)$ where $\hat \cI=\hat b^* \cI$ is just a way to define the desired gPDE $(\cE,Q)$. Note that $\hat b^{*}\cI=\hat b^{*}\cK$.

In order to perform concrete computations on $\hat \Ee$ we need an efficient way to transfer the structures from $\Ee$ to $\hat \Ee$. If $\hat \Ee$ were a $Q$-submanifold of $\Ee$ that would have been straightforward. But it is not and a more elaborate technique is in order. Namely, we identify a subalgebra of functions on $\Ee$ such that their restrictions to $\hat \Ee$ is a $Q$-map. We call $f\in  \cC^\infty(\Ee)$  $\hat \Ee$-projectable if its restriction to $\Ee^{(D-3)}$ is constant along the fibres of $\Ee^{(D-3)}$ or, equivalently can be obtained as a pullback from $\hat \Ee$ by the projection $\Ee^{(D-3)} \to \hat \Ee$.\footnote{The idea here is the same as in the definition of observable is the context of gauge systems. We require the restrictions of $f$ to a submanifold (1st class constrained surface in the Hamiltonian approach or stationary surface in the Lagrangian one) to be invariant  under the distribution therein (gauge distribution). This can be made precise by picking any smooth extension of the vertical distribution on $\Ee^{(D-3)}$ to $\Ee$. Then the projectable functions are those annihilated by the distribution modulo functions vanishing on $\Ee^{(D-3)}$, i.e. precisely the observables from the gauge system perspective. 
}
Examples of $\hat \Ee$-projectable functions include \eqref{funct-E-D-3} as well as $T_{AB}^{(N)}$, $J_{ABC}^{(N)}$ for $N\leq D-3$. A simple criterion: the restriction of a function to $\Ee^{(0)}$ does not depend on $T^{(N)}_{AB}$, $\Omega^{(N+2)}$ with $N\geq D-2$.

Given a vector field tangent to $\Ee^{(D-3)}$ and such that its restriction to
$\Ee^{(D-3)}$ is projectable, it automatically preserves the subalgebra of $\hat \Ee$-projectable functions and hence defines a vector field on $\hat \Ee$ in a canonical way. We shall refer to such vector fields on $\Ee$ as $\hat \Ee$-projectable and will denote their restriction to $\hat \Ee$ by the same symbol. Note that if $f\in \cC^\infty(\Ee)$ and $V\in Vect(\Ee)$ are $\hat \Ee$-projectable, then we have $V\hat f=\hat b^{*}Vf$. In particular, $Q$ on $\Ee$ is like that and indeed induces $Q$ on $\hat \Ee$ in this way i.e. $\hat b^*$ restricted to $\hat \Ee$-projectable functions is a $Q$-map. 
Also, if a ghost degree $-1$ vector field $Y$ is tangent to $\Ee^{(D-3)}$, then it is automatically tangent to $\hat \Ee\subset \Ee^{(D-3)}$, since the relations defining  this submanifold are degree $0$. Consequently, the corresponding $Q$-exact vector fields $[Q,Y]$ are $\hat \Ee$-projectable. It then follows from Lemma~\bref{lemma-I0} and Theorem~\bref{prop:mnogoidealov} that the vector fields on $E_{\red}$ 
$\nabla_A$, $\Gamma^B$, $\Delta^{A}{}_{B}$, and $\mathcal{D}_A^{(N)}$, $N\leq D-4$, are $\hat \Ee$-projectable, and we will use the same notation for their restrictions to $\hat \Ee$.

As we have just seen the action of $Q$
on coordinate functions \eqref{funct-E-D-3} pulled back to $\hat \Ee$   can be computed on $\Ee^{(D-3)}$ which in turn can be computed on $\Ee$. Applying this procedure to \eqref{gran-ish1} and \eqref{W-counted} and using relations~\eqref{boundcond} and Lemma~\bref{lemma-I0}, we find that
\begin{equation} 
\label{Q-gran-background}
\begin{aligned}
        Q \hat g_{BC}=&~\hat C_{B}{}^A \hat g_{AC}+\hat C_{C}{}^A \hat g_{BA}+2\hat \lambda \hat g_{BC},\\
         Q\hat\lambda_A=&~\hat C_{A}{}^{B}\hat \lambda_B  +\dfrac{1}{2}\hat \xi^B \hat \xi^C \hat \Co_{ABC},\\
        Q\hat \lambda=&~\hat \xi^A \hat \lambda_A,\quad Q\hat \xi^B=\hat \xi^A \hat C_{A}{}^{B}\\
        Q \hat C_{B}{}^C=&~\hat C_{B}{}^A \hat C_{A}{}^C+\hat \lambda_B\hat \xi^C-\hat \lambda^C \hat \xi_B+\delta_{B}^C \hat \lambda_A \hat \xi^A +\dfrac{1}{2}\hat \xi^A\hat \xi^D \hat \We^{C}{}_{BAD}\,\\
          Q\hat \We^{C}{}_{BAD}=&~\hat \xi^E \nabla_{E}\hat \We^{C}{}_{BAD}-\hat C_{E}{}^{C}\hat \We^{E}{}_{BAD}+\hat C_{B}{}^{E}\hat \We^{C}{}_{EAD}+\\
        ~&~\hat C_{A}{}^{E}\hat \We^{C}{}_{BED}+ \hat C_{D}{}^{E}\hat \We^{C}{}_{BAE}\,.
\end{aligned}
\end{equation}
Note that thanks to $[\nabla_{A}, Q]=0$, the above coordinates along with the $\nabla_{A}$-prolongations of $\hat \We^{B}{}_{CDE}$ are closed under $Q$, i.e. the action of $Q$ does not involve the remaining coordinates $\nabla_{(A)}\cT_{BC}$. Moreover, the underlying $Q$-manifold (whose algebra of functions is generated by the above coordinates) coincides with that describing the conformal geometry (see Proposition \bref{prop:bundle-over-conf}
; of course the lowercase indices are to be  replaced by the uppercase ones). If we denote by $(\Ee^{conf},Q)$ the respective  gPDE we have:
\begin{prop} 
$(\hat \Ee,Q)$ is a $Q$-bundle over $(\Ee^{conf},Q)$.  The fibre coordinates on this bundle are, of course, $\nabla_{(A)} \cT_{BC}$
\end{prop}

On top of having a relatively simple structure, $(\hat \Ee,Q)$ has an additional very attractive feature of being off-shell, as we are going to demonstrate explicitly  in Section~\bref{sec:gravity-sections}.
 From the above analysis it follows that $\cE$ is the zero locus of the ideal $\hat\cI\equiv \hat b^{*}\cI$, generated by $\nabla_{(C)} O^{(D-3)}_{AB}, \nabla_{(C)} O^{(D-3)}_{A}$ pulled back to $\hat \Ee$. In what follows we use the notations:
\begin{equation}
\label{hatOdef}
  \cO_{AB}\equiv -\frac{1}{(D-3)\tilde{\Lambda}}\hat b^* ( O^{(D-3)}_{AB})=\hat b^{*}\nabla_\Omega^{D-4}\nabla^{C}J_{CBA}\,, \qquad   \cO_{A}=\hat b^* ( O^{(D-3)}_{A})=\hat b^{*}\nabla_{\Omega}^{D-3}\nabla^{B}T_{BA}\,.
\end{equation}
 These functions are well-defined as pullbacks of functions from $\Ee$, but it is not entirely obvious what they are in terms of coordinates on $\hat \Ee$. For this, it will be necessary to introduce a certain calculus, which is the subject of the next section. Using  this calculus, in Section~\bref{sec:even} we show that $\cO_{AB}$ corresponds to the Fefferman-Graham obstruction tensor.

 It is obvious that  $\hat \Ee$ the zero locus of $\hat \cI$ coincides with $\cE$. As an implicit gPDE what we are interested in is $(\hat \Ee, Q,\hat \cI)$.

Although their explicit form is not known to us yet, we can already make some general statements:
\begin{prop}\label{prop-div-obstr}
    On $\hat \Ee$ we have
    \begin{align}
      \cO_{A}{}^{A}=0\,,\quad  \nabla^{A}\cO_{AB}=0,\quad \Gamma_A\cO_B=-\tilde\Lambda(D-3) \cO_{AB}
    \end{align}
    where $\Gamma_A\equiv \hat g_{AB}\Gamma^B$, $\Gamma_B =[\frac{\partial}{\partial \hat \lambda_{B}},Q]$.
\end{prop}
\begin{proof}The tracelessness condition for $\cO_{AB}$ follows directly from \eqref{hatOdef} together with the fact that $J^{A}{}_{BC}$ is traceless (which, in turn, follows from $W^{a}{}_{bac}=0$).

    The divergence-free condition is obtained by simply applying $\hat b^{*}$ to the formula \eqref{App-div-free} that appeared in the proof of Theorem~\bref{prop:mnogoidealov}. In fact, the origin of this equation can be traced back to the Noether identities for the Einstein equations.

       The last formula follows directly from the Proposition \bref{prop-gamma-conservation} proven in the appendix.
\end{proof}

The first two conditions, as we will see later, are simply the usual tracelessness and divergence-free conditions for the obstruction tensor. The third one implies that  $\cO_B$ is gauge invariant provided equations $~\cO_{AB}=0$ is imposed.

To anticipate a little, functions $\nabla_{(C)}\cO_{AB}$ are
constant along the fibres of $\hat \Ee \to \Ee^{conf}$ and hence determine functions on the base $\Ee^{conf}$. It follows they define a $Q$-submanifold $\cE^{conf}$ (more precisely, a $Q$-subbundle, if we view $\Ee^{conf}$ as a bundle over $T[1]\Sigma$) of $\Ee^{conf}$. As we are going to see shortly, this subbundle is precisely the obstruction gPDE, e.g. conformal gravity at the level of equations of motion for $D=5$. 

Moreover, functions $\nabla_{(C)} \cO_{A}$ restrict fibre coordinates $\nabla_{(C)}\cT_{AB}$ only. All in all this means that the gPDE induced on the boundary is itself a bundle over $ \cE^{conf}$. More precisely, it is a subbundle of the initial bundle $\hat \Ee \to \Ee^{conf}$ pulled back to $\cE^{conf} \subset \Ee^{conf}$. In the case of $D=5$ this system describes the conserved energy-momentum tensor $\cT_{AB}$, as we discuss in  more details in the next Section.

\subsection{Boundary calculus}
\label{sec:boundary-calculus}
Although we have obtained a certain description of the boundary gPDE as $(\hat \Ee,Q,\hat \cI)$, it is very implicit for the moment. For instance, we do not know the explicit form of the generators  $\nabla_{(C)}  \cO_{A}$ and $\nabla_{(C)}  \cO_{AB}$ of $\hat \cI$ in terms of the coordinates on $\hat \Ee$. It turns out that the action of $Q$ on the coordinates and the constraints themselves can be expressed in terms of the following functions
\begin{equation}
 \qquad \hat J^{(N)}_{ABC} \,, \qquad  \hat T^{(N)}_{AB}\,, \qquad 0\leq N \leq D-4
\end{equation}
and vector fields 
\begin{equation}
\cD^{(N)}_A\,, \qquad 0\leq N \leq D-4
\end{equation}
defined on $\hat \Ee$. Moreover, these can be explicitly found in terms of the coordinates on $\hat \Ee$ via a certain recursive procedure. In what follows we refer to the above system as to \textit{boundary calculus}.

The important technical tool that allows one to represent the objects like $\cO_{A}\equiv\hat b^{*}\nabla_{\Omega}^{D-3}\nabla^AT_{AB}$ in terms of the boundary calculus is the following: 
\begin{prop}\label{commut-nablas}
    Let $f$ be a  $\hat \Ee$-projectable function on $\Ee$ and $N\leq D-4$. Then  
    \begin{align}
        \hat b^{*}\nabla_\Omega^{N}\nabla_Af=\sum_{i=0}^{N}C^{i}_{N}\mathcal{D}_{A}^{(N-i)}\hat f^{(i)}=\nabla_A \hat f^{(N)}+\sum_{i=0}^{N-1}C^{i}_{N}\mathcal{D}_{A}^{(N-i)}\hat f^{(i)}.
    \end{align}
\end{prop}
\begin{proof}
    Follows immediately from Proposition~\bref{utv-commutator} and the fact that $\mathcal{D}_A^{(N)}$ for $N\leq D-4$ are $\hat \Ee$-projectable and hence $\hat b^{*}\cD_{A}^{(N)}=\cD_{A}^{(N)}\hat b^{*}$.
\end{proof}
\noindent 
Note that an analogous formula holds for all $N$ if $\hat \Ee$ is replaced with the on-shell gPDE $\cE$.

With the help of this simple Proposition we can express $\hat b^{*}\nabla_\Omega^{N}\nabla_Af$ in terms of the boundary calculus and the ``subleadings'' $\hat f^{(i)} \equiv \hat b^* (\nabla_\Omega)^i f$. In particular, constraints $\cO_A$ and $\cO_{AB}$ given by \eqref{hatOdef} take the form\footnote{In the second formula we have already taken into account that the term with $i=0$ vanishes thanks to~\eqref{nu-counted-bulk} and $\hat{T}_{AB}^{(0)}=0$.}:
\begin{align}\label{obstr-hatEB}
    \cO_{AB}&=\nabla^{C}\hat J^{(D-4)}_{CAB}+\sum_{i=0}^{D-5}C^{i}_{D-4}\mathcal{D}^{(D-4-i)|C}\hat J^{(i)}_{CAB}\,,\\
    \cO_{A}&=\nabla^{C} \cT_{AC}+\sum_{i=1}^{D-4}C^{i}_{D-3}\mathcal{D}^{(D-3-i)|C}\hat T^{(i)}_{AC}\,.
\end{align}
In what follows we refer to $\cO_{AB}$ and $\cO_{A}$ as the symbols of the equations of motion for the boundary system. This is natural because, as we are going to see in Section~\bref{sec:gravity-sections}, if $\sigma$ is a solution to  gPDE $(\hat \Ee,Q)$ then it is determined by the unconstrained fields parametrizing leading and subleading boundary values. This means that 
$\sigma^*(\cO_{AB})=0$ and $\sigma^*(\cO_{A})=0$ are differential equations on $\sigma$, ensuring that $\sigma$ is a solution to the boundary system , i.e. to the gPDE $(\hat \Ee,Q,\hat\cI)$.

As for the $Q$-structure on $\hat \Ee$, its action in the conformal geometry sector (the base of $\hat \Ee \to \Ee^{conf}$) was already determined and is given by \eqref{Q-gran-background} so the only nontrivial computation is the action of $Q$ on the coordinate function $\cT_{AB}$. We have the following slightly general statement:
\begin{prop}\label{QTAB}
    On $\hat \Ee$ we have 
    \begin{align}
    \label{QT}
         Q\hat T^{(N)}_{BC}=\hat \xi^{A}\nabla_{A}\hat T^{(N)}_{BC}+\hat C_{B}{}^{A}\hat T^{(N)}_{AC}+\hat C_{C}{}^{A}\hat T^{(N)}_{BA}+\hat \lambda_{A}\Gamma^{A}\hat T^{(N)}_{BC}-N\hat \lambda \hat T^{(N)}_{BC},\quad N\leq D-3
    \end{align}
    where
    \begin{align}
    \label{QTact}
        \Gamma_{A}\hat T^{(N)}_{BC}=\tilde{\Lambda}\sum^{N-2}_{i=0}d^{i}_N\mathcal{D}_{A}^{(N-2-i)}\hat T^{(i)}_{BC}+N\tilde{\Lambda}(\hat J^{(N-1)}_{ABC}+\hat J^{(N-1)}_{ACB})\,,\quad N\leq D-3
    \end{align}
    where $d^{i}_N\equiv C^i_N(N-1-i)$.
\end{prop}
The proof is based on the useful formula
\begin{equation}\label{formula Gamma}
       \hat b^{*}[\Gamma^{B},\nabla_{\Omega}^{N}]=\tilde{\Lambda}\sum_{i=0}^{N-2}d^{i}_N \mathcal D^{(N-2-i)|B}\hat b^{*}\nabla^{i}_\Omega+N\hat b^{*}\nabla_\Omega^{N-1}(\tilde\Lambda g^{BC}\Delta^{\Omega}{}_{C}-\Delta^{B}{}_{\Omega}),\qquad N\leq D-2
\end{equation}
and is placed in Appendix \bref{sec:prof-com}.

The above discussion confirms that all we need to know is the boundary calculus on $\hat \Ee$. We now turn to the recursive construction of the boundary calculus: 
\begin{theorem}\label{bound-calculus}
{\rm(Boundary calculus)}

The system $\hat T^{(N)}_{AB}\,, \hat J^{(N)}_{ABC}  
\,, \cD^{(N)}_A$ on $\hat \Ee$ satisfy the following relations
\begin{align}\label{formuli-calculus}
    \begin{split}
            \hat T^{(N)}_{BC}&=\frac{N\tilde{\Lambda}}{D-3-N}( \hat b^{*}\nabla_\Omega^{N-1}\nabla^{A}J_{A BC}), \quad 0\leq N\leq D-4,\\
        \hat J^{(N)}_{ABC}&=\dfrac{N}{N-1}\hat b^{*}\nabla_{\Omega}^{N-1}( \nabla_{A}T_{BC}-\nabla_{B}T_{AC})\,, \quad 2\leq N\leq D-4,
                \end{split}
\end{align}
and 
\begin{align}\label{formuli-D}
\begin{split}
        &\mathcal{D}_{A}^{(N)}f=\Big(-\hat J^{(N-1),D}{}_{CA}\Delta^{C}{}_{D}+\dfrac{1}{\tilde\Lambda N}\hat T^{(N)}_{CA}\Gamma^{C}-\frac{1}{N}\sum_{i=1}^{N-2} d^{i}_N \hat T^{(i),C}{}_{A}\mathcal{D}_{C}^{(N-2-i)}\Big)f\,, \quad 1\leq N\leq D-4\,.
\end{split}
\end{align}
Moreover, 
\begin{equation}\label{J-lower}
\hat J^{(0)}_{ABC}=0\,, \qquad 
\hat J^{(1)}_{ABC}=
\tilde{\Lambda} \dfrac{1}{D-4}\nabla^{D}\hat \We_{DCAB}\equiv -\tilde{\Lambda}\hat \Co_{CAB}\,,\qquad 
\hat T^{(0)}_{AB}=0\,, \qquad 
\mathcal{D}^{(0)}_{A}\equiv\nabla_A\,,\\
\end{equation}
and all the higher order $J^{(\cdot)}_{ABC},T^{(\cdot)}_{BC},\cD^{(\cdot)}_A$ are recursively determined by the above relations where each expression of the form $\hat b^{*}\nabla_\Omega^{N}\nabla_A f$  is evaluated as 
\begin{align}
\label{Omega-A-commut}
    \hat b^{*}\nabla_\Omega^{N}\nabla_A f=\sum_{i=0}^{N}C^{i}_{N}\mathcal{D}_{A}^{(N-i)}\hat f^{(i)}\,,
\end{align}
and the action of $\Gamma_A$ on $\hat T^{(N)}_{BC}$ is given by~\eqref{QTact}.
\end{theorem}
\begin{proof}
The formula for $\hat T^{(N)}_{AB}$ is a direct consequence of the definition of $\hat \Ee$, see Definition~\bref{def-hatE}.
The proof of the formula for $\hat J^{(N)}_{ABC}$ follows from the technically involved Proposition~\bref{App-J-counted} placed in the Appendix. Similarly, in Proposition ~\bref{App-nu-counted} the formula for $\nu_A^{(N)}$ is proven, from which the formula for $\mathcal{D}_A^{(N)}$ directly follows. Finally, to prove the iterativity of the proposed procedure, it is important to note that the action of $\Delta^{A}{}_{B}$ and $\Gamma_C$ on functions in $\cC^\infty(\hat \Ee)$ can be computed from the action of $Q$ on the coordinates given in $\eqref{Q-gran-background}$, $\eqref{QT}$.
\end{proof}
\begin{remark}\label{remark-higherN}
    It follows from the proofs in Appendix~\bref{sec:prof-com} and Appendix ~\bref{sec:prof-BC} that formulas 
~\eqref{QTact}, \eqref{formula Gamma}, \eqref{formuli-calculus} and \eqref{formuli-D} remain valid for higher $N$ as well, but in this case only modulo the ideal $\hat{\cI}$ (which in fact means that they hold on the on-shell gPDE $\cE \subset \hat \Ee$). We do not need this fact in studying gravity and Yang-Mills, but it  becomes important for the analysis of the extended GJMS equations in Section~\bref{sec:GJMS}.
\end{remark}

\begin{cor}
Functions 
\begin{equation}
\hat T^{(N)}_{AB}\,,~~\hat J^{(N)}_{ABC}\,,~~ \cD^{(M)}_A \hat T^{(L)}_{BC}\,,~~ \cD^{(M)}_A \hat J^{(L)}_{BCD}\,,~~ \qquad M+L\leq N
\end{equation}
 are $\cT_{AB}$-independent for $N\leq D-4$ and hence are functions pulled back from $\Ee^{conf}$ by $\hat \Ee \to \Ee^{conf} $. In particular, the submanifold $\cE^{conf} \subset \Ee^{conf}$ singled out by $\nabla_{(C)}\cO_{AB}$ is a sub-gPDE.  
\end{cor}
In this sense the boundary calculus as defined above can be defined solely in terms of the conformal geometry and hence can be considered in the conformal geometry setup. The functions $\hat T_{AB}^{(N)}$ for $N\leq D-4$ are closely related to what Graham in \cite{graham2009extended} referred to as extended obstruction tensors.
\begin{cor}
The gPDE $\cE$ induced on the boundary (i.e. a sub-gPDE of $\hat \Ee$ singled out by $\hat \cI$) is a $Q$-bundle over the gPDE $\cE^{conf}$. 
\end{cor}
Indeed, the remaining constraints $\nabla_{(C)} \cO_{A}$ do not restrict the base and hence define a subbundle of $\hat \Ee$ pulled back to $\cE^{conf}$. We can conclude that the boundary system can be seen as the conformal geometry subject to some conformal invariant equations arising from $\cO_{AB}$ together with the ``matter-like'' field $\cT_{AB}$ subject to equations arising from $\cO_{A}$.

Another simple consequence of the above theorem is:
\begin{prop}\label{utv-nechet}
For any odd integer $i$ such that  $1\leq i\leq D-4$ we have: 
\begin{equation}
\hat T^{(i)}_{BC}=0\,, \qquad \cD^{(i)}_{A}=0\,, \qquad \hat J^{(i-1)}_{ABC}=0\,.
\end{equation}
\end{prop}
The proof is relegated to the Appendix~\bref{sec:proof-TDJ}

\section{Boundary structure of gravity}
\label{sec:boundary-GR}
After having developed boundary calculus for asymptotically AdS GR let us now have a close look at the gauge theory induced on the asymptotic boundary. The case of even- and odd-dimensional boundary are quite different and we treat them separately.

\subsection{Odd-dimensional boundary}\label{sec:odd}
Now we discuss the case where the boundary dimension $d\equiv D-1$ is odd and assume $d\geq 3$. 
Let us first discuss the case $d=3$, which we have excluded since Proposition~\bref{form-coord-E0}.
\begin{prop}
    For $d=3$ the boundary gPDE $\hat{\Ee}$  can be described in terms of fibre coordinates 
    \begin{align}
        \hat \xi^{B},\quad \hat C_{B}{}^{C},\quad \hat \lambda,\quad \hat \lambda_{B},\quad    \hat g_{BC},\quad \nabla_{((A)} \hat J^{(1)|B}{}_{C)D},\qquad \nabla_{(A)} \cT_{BC},\quad  |A|\geq0
    \end{align}
    where  $\cT_{BC}\equiv \hat b^{*}T^{(1)}_{BC}$, the action of $Q$ is
\begin{equation} \label{Q-d3}
\begin{aligned}
        Q \hat g_{BC}=&~\hat C_{B}{}^A \hat g_{AC}+\hat C_{C}{}^A \hat g_{BA}+2\hat \lambda \hat g_{BC},\\
         Q\hat\lambda_A=&~\hat C_{A}{}^{B}\hat \lambda_B  +\dfrac{1}{2}\hat \xi^B \hat \xi^C \hat J^{(1)}{}_{BCA},\\
        Q\hat \lambda=&~\hat \xi^A \hat \lambda_A,\quad Q\hat \xi^B=\hat \xi^A \hat C_{A}{}^{B}\\
        Q \hat C_{B}{}^C=&~\hat C_{B}{}^A \hat C_{A}{}^C+\hat \lambda_B\hat \xi^C-\hat \lambda^C \hat \xi_B+\delta_{B}^C \hat \lambda_A \hat \xi^A \,\\
        QJ^{(1)}_{BCD}=&\hat\xi^{A}\nabla_{A}\hat J^{(1)}_{BCD}+\hat C_B{}^{A}\hat J^{(1)}_{ACD}+\hat C_{C}{}^{A}\hat J^{(1)}_{BAD}+\hat C_{D}{}^{A}\hat J^{(1)}_{BCA},\\
         Q\cT_{BC}=&\hat \xi^{A}\nabla_{A}\cT_{BC}+\hat C_{B}{}^{A}\cT_{AC}+\hat C_{C}{}^{A}\cT_{BA}-\hat \lambda \cT_{BC}
\end{aligned}
\end{equation}
and we have $\cO_{AB}=0$, $\cO_B=\nabla^A\cT_{AB}.$
\end{prop}
\begin{proof}
    The main difference here is that in this dimension the Bianchi identities, as follows from \eqref{proof-coordE0-2}, do not fix $J^{(1)}_{BCD}$.  
However, by substituting two indices $\Omega$ into $\nabla_{[a}C^{b}{}_{cd]}=0$, we observe that $\nabla_{(A)}J^{(N)}_{ABC}$, $N\geq 2$, are still determined by the Bianchi identities.  
Furthermore, from the algebraic symmetries of the Weyl tensor we have that in dimension $d=3$ $b_0^{*}W^{B}{}_{CDE}=0$.  
Therefore, using $C_{ABC}=-\nabla_d W^{d}{}_{ABC}$, we find that 
\begin{align}
    \hat C_{ABC}=-\tilde\Lambda\, \hat J^{(1)}_{BCA}
\end{align}
which leads to the formula for $Q\hat\lambda_{B}$. The computation for $Q\cT_{AB}$, $QJ^{(1)}_{BCD}$, and $\cO_{AB}$ is straightforward.
\end{proof}

 Note that again all coordinates except for $\nabla_{(A)}\cT_{AB}$ can be interpreted as coordinates on the base of the $Q$-bundle $\hat \Ee\to \Ee^{conf}$, where the action of $Q$ on $\Ee^{conf}$ coincides with the formulas given for the BRST complex of three-dimensional conformal geometry in \cite{Boulanger:2004eh}, with $\hat J^{(1)}_{ABC}$ playing the role of the three-dimensional Cotton tensor.

 For a general odd-dimensional boundary, we have:
\begin{prop}\label{J-nechet}
    For $d$ odd the boundary gPDE $\hat \Ee$ introduced in Section~\bref{sec:boundary-gPDE} is such that:
    \begin{enumerate}
        \item ${{\cO}}_{AB}=0$
        \item $\Gamma_{A}\cT_{BC}=0$
        \item $\cO_A=\nabla^{B}\cT_{BA}$
    \end{enumerate}
\end{prop}
\begin{proof}
Of course, the case $d=3$, as has been shown, satisfies this, and therefore we will prove it for $d>3$.

First of all, Theorem~\bref{bound-calculus} and Proposition~\bref{utv-nechet} imply:
\begin{equation}
    \frac{D-5}{D-4}\hat J^{(D-4)}_{ABC}=\hat b^{*}\nabla_{\Omega}^{D-5}(\nabla_{A}T_{BC}-\nabla_{B}T_{AC})=
    \hat b^{*}([\nabla^{D-5}_{\Omega},\nabla_{A}]T_{BC}-[\nabla^{D-5}_{\Omega},\nabla_{B}]T_{AC})\,.
\end{equation}
Because $D>5$ by the assumption, Proposition \bref{utv-nechet} together with \eqref{Omega-A-commut} implies $\hat J^{(D-4)}_{ABC}=0$. Using this together with, again, Proposition~\bref{utv-nechet}, we obtain 
\begin{align}
    {\cO}_{BC}=\nabla^{A}\hat J^{(D-4)}_{A BC}+\hat b^{*}[\nabla_{\Omega}^{D-4},\nabla^{A}]J_{A BC}=0\,.
\end{align}
Similarly, we have
\begin{align}
    \begin{split}
               \Gamma_{A}\cT_{BC}&=\tilde{\Lambda}\sum^{D-5}_{i=0}d^{i}_{D-3}\mathcal{D}_{A}^{(D-5-i)}\hat T^{(i)}_{BC}+(D-3)\tilde{\Lambda}(\hat J^{(D-4)}_{ABC}+\hat J^{(D-4)}_{ACB})=0,\\
                    \cO_{A}&=\nabla^{C} \cT_{AC}+\sum_{i=1}^{D-4}C^{i}_{D-3}\mathcal{D}^{(D-3-i)|C}\hat T^{(i)}_{AC}=\nabla^{C} \cT_{AC}.
    \end{split}
\end{align}
\end{proof}

Because $\cO_{AB}$ vanishes identically the gPDE induced on the boundary simply describes conformal geometry with an additional ``energy-momentum tensor'' $\cT_{AB}$ subject to the equations of motion emerging from the constraint  $\cO_A\equiv \nabla^{B}\cT_{BA}$ and its $\nabla_C$ prolongations. More precisely, the induced boundary gPDE
is a subbundle of $\hat \Ee \to E^{conf}$ singled out by $\cO_A$ and its prolongations.

Moreover, the induced boundary gPDE is a vector bundle over $\Ee^{conf}$ because $Q\cT_{AB}$ is linear in $\nabla_{(C)}\cT_{AB}$
as is easily seen from Proposition~\bref{QTAB}.
In particular, a zero section of this vector bundle is the conformal geometry gPDE $\Ee^{conf}$ itself. In particular, 
\begin{prop}\label{prop:odd-vanish}
For $d$ odd, the following functions and vector fields on $\hat \Ee$
\begin{equation}
\hat T^{(2i+1)}_{A B}\,, \qquad \hat J^{(2i)}_{A B C}\,, \qquad \cD^{(2i+1)}_{A}\,, \quad i \geq 0\,,
\end{equation}
vanish at the zero section of $\hat \Ee \to \Ee^{conf}$, i.e. at $\nabla_{(C)}\cT_{AB}=0$.
\end{prop}
The composition of the zero section with $\hat b$ realizes $\Ee^{conf}$ as a $Q$-submanifold in $\Ee$.
\subsection{Even-dimensional boundary}\label{sec:even}
This case is more interesting. Before describing the structure of the boundary gPDE in this case, let us give a couple of low-dimensional examples:
\begin{example}\label{example-5D}
   In the case of $D=5$ one has 
   \begin{equation}
   \cO_{AB}=\tilde{\Lambda}\hat B_{AB}\,,
   \qquad 
   \cO_A=\nabla^{B}\cT_{BA}\,,
   \end{equation}
  where $\hat B_{BC}\equiv \nabla^{A}\hat \Co_{BCA}$ is the ``symbol of  Bach tensor''. Namely, solutions of $\Ee^{conf}$ seen as a gPDE are parametrised by configurations of
   a metric on $X$ along with the choice of local frame. When pulled back by such a solution, $\hat B_{BC}$ becomes a Bach tensor of the metric. More detailed discussion of the explicit form of the underlying equations is given in Section~\bref{sec:gravity-sections}. Furthermore, 
   \begin{align}
        Q\cT_{BC}=\hat \xi^{A}\nabla_{A}\cT_{BC}+\hat C_{B}{}^{A}\cT_{AC}+\hat C_{C}{}^{A}\cT_{BA}-2\tilde{\Lambda}^2\hat \lambda^{A}(\hat \Co_{CAB}+\hat \Co_{BAC})-2\hat \lambda \cT_{BC}
    \end{align}
which defines gauge transformations of the field $\cT_{AB}$. The above assertions are checked by direct computations:
    \begin{align}
    \begin{split}
                \cO_{C}&=\nabla^{A}\cT_{AC}+\hat b^*[\nabla_{\Omega}^2,\nabla^{A}]T_{AC}=\nabla^{A}\cT_{AC},\\
\cO_{BC}&=\nabla^{A}\hat J^{(1)}_{ABC}+\hat b^{*}[\nabla_\Omega,\nabla^{A}]J_{ABC}=\nabla^{A}\hat J^{(1)}_{ABC}=-\tilde{\Lambda}\nabla^{A}\hat \Co_{CAB}=\tilde{\Lambda} \hat B_{BC}\,,
    \end{split}
    \end{align}
    where we made use of $\hat T^{(0)}_{AC}=0$, $\mathcal{D}^{(1)}_A=0$, and $\hat J^{(1)}_{ABC}=-\tilde{\Lambda}\hat C_{CAB}$, see Eqs.~\eqref{J-lower} and \eqref{utv-nechet}. The only nontrivial part of the $Q\cT_{AB}$ computation is:
    \begin{align}
        \Gamma_{A}\cT_{BC}= \tilde{\Lambda} \nabla_{A}\hat T^{(0)}_{BC}+2 \tilde{\Lambda}(\hat J^{(1)}_{ABC}+\hat J^{(1)}_{ACB})=-2\tilde{\Lambda}^2(\hat \Co_{CAB}+\hat \Co_{BAC})\,.
    \end{align}
In other words, for $D=5$ the boundary gPDE describes Bach-flat geometry (conformal gravity), together with the ``subleading'' (that corresponds to the boundary energy-momentum tensor in the AdS/CFT context) $ \cT_{BC}$  satisfying the conservation condition. Note, however, that $\nabla_{(A)} \cT_{BC}=0$ is in general not a sub-gPDE of the boundary gPDE, i.e. $Q$ is not tangent to it, — as follows from the fact that $\Gamma_{A} \cT_{BC}$ is not zero in general (though it may vanish if an additional condition on the background is imposed). Of course boundary gPDE is a bundle over $\hat \Ee^{conf}$ and hence 
$\hat \Ee^{conf}$ can be recovered as the quotient of the boundary gPDE by the fibres.
\end{example}

\begin{prop}
\label{example-7D}
    For $D=7$ one has:
    \begin{align}
        \begin{split}
             \cO_{BC}&=\frac{3}{2}\tilde{\Lambda}^{2}\Big(\nabla^{A}\nabla_{A}\hat B_{BC}+2\hat \We_{DCAB}\hat B^{AD}+2 \hat \Co_{C}{}^{AD}\hat \Co_{BAD}-4\hat \Co^{D}{}_{B}{}^{A}\hat \Co_{ACD}\Big),\\
            \cO_B&=\nabla^{B}\cT_{BC},\\
             \Gamma_{A}\cT_{BC}&=6\tilde{\Lambda}^3\Big(3\nabla_A\hat B_{BC}-\nabla_{B}\hat B_{AC}-\nabla_{C}\hat B_{AB}\Big).
        \end{split}
    \end{align}
\end{prop}
\begin{proof}
    From Proposition~\bref{utv-nechet} we have $\mathcal{D}^{(1)}_{A}=\mathcal{D}^{(3)}_{A}=0$, $\hat J^{(0)}_{ABC}=\hat J^{(2)}_{ABC}=\hat T^{(0)}_{BC}=\hat T^{(1)}_{BC}=\hat T^{(3)}_{BC}=0$, and that $\hat J^{(1)}_{ABC}=-\tilde{\Lambda} \hat \Co_{CAB}$. Also, analogously to the previous example, using \eqref{formuli-calculus}, we obtain $ \hat T_{AB}^{(2)}=\tilde{\Lambda}^{2}\hat B_{AB}.$

    It follows from \eqref{obstr-hatEB} and the formulas above that:
\begin{align}\label{ex:7D-1}
    \cO_{BC}=\nabla^{A}\hat J_{ABC}^{(3)}+3\mathcal{D}^{(2)|A}\hat J^{(1)}_{ABC}\,.
\end{align}
Using Theorem~\bref{bound-calculus} we obtain
\begin{align}\label{ex:7D-2}
\begin{split}
        \nabla^{A}\hat J^{(3)}_{ABC}=&\frac{3}{2}\nabla^{A}\hat{b}^{*}\nabla_\Omega^{2}(\nabla_{A}T_{BC}-\nabla_{B}T_{AC})=\frac{3}{2}(\nabla^{A}\nabla_A \hat T^{(2)}_{BC}-[\nabla^{A},\nabla_{B}]\hat T^{(2)}_{AC})=\\&=\frac{3}{2}\tilde{\Lambda}^{2}\Big(\nabla_{A}\nabla^{A}\hat B_{BC}+\hat \We_{DCAB}\hat B^{AD}-2\tilde{\Lambda}^{2}\hat \Co^{DA}{}_{B}(\hat \Co_{CDA}+\hat \Co_{ADC})\Big)
        \end{split}
\end{align}
where, in the last term, we used Proposition~\bref{QTAB}, from which $\Gamma_{D}\hat T_{AC}^{(2)}=-2\tilde{\Lambda}^{2}(\hat \Co_{CDA}+\hat \Co_{ADC})$.
\begin{align}
            \mathcal{D}^{(2)|A}\hat J^{(1)}_{ABC}&=-\hat J^{(1)|D}{}_{E}{}^{A}\Delta^{E}{}_{D}\hat J^{(1)}_{ABC}+\frac{1}{2\tilde\Lambda}\hat T^{(2)|A}{}_{D}\Gamma^{D}\hat J^{(1)}_{ABC}=\\&=\tilde{\Lambda}^{2}(-\hat \Co^{AD}{}_{B}\hat \Co_{CAD}-\hat \Co^{AD}{}_{C}\hat \Co_{DAB}+\frac{1}{2}\hat B^{AD}\hat \We_{DCAB})
\end{align}
where we have also made use of $\Gamma^{A}\hat \Co_{BCD}=-\hat \We^{A}{}_{BCD}.$ Substituting these results \eqref{ex:7D-1}, together with the Bianchi identity for the Cotton tensor, yields
\begin{align}
    \cO_{BC}=\frac{3}{2}\tilde{\Lambda}^{2}\Big(\nabla^{A}\nabla_{A}\hat B_{BC}+2\hat \We_{DCAB}\hat B^{AD}+2 \hat \Co_{C}{}^{AD}\hat \Co_{BAD}-4\hat \Co^{D}{}_{B}{}^{A}\hat \Co_{ACD}\Big).
\end{align}
This is precisely the formula for the Fefferman-Graham obstruction tensor in six dimensions, up to corrections by the Schouten tensor which, as we shall see in Section~\bref{sec:gravity-sections}, arise on the solution space of the gPDE.
Next, using \eqref{obstr-hatEB},
\begin{align}
    \cO_A=\nabla^B\cT_{BA}+C^{2}_{4}\mathcal{D}^{(2)|B}\hat T^{(2)}_{BA}
\end{align}
    
    \begin{align}
        \mathcal{D}^{(2)|B}\hat T^{(2)}_{BA}=-\hat J^{(1)|D}{}_{A}{}^{B}\hat T^{(2)}_{BD}+\frac{1}{2\tilde{\Lambda}}\hat T^{(2)}_{D}{}^{B}\Gamma^{D}\hat T^{(2)}_{BA}=\tilde{\Lambda}^3\hat \Co^{BD}{}_{A}\hat B_{BD}-\tilde{\Lambda}^3 \hat B_{D}{}^{B}(\hat \Co_{A}{}^{D}{}_{B}+\hat \Co_{B}{}^{D}{}_{A})=0.
    \end{align}

    Finally, Proposition~\bref{QTAB} yields
    \begin{align}
        \Gamma_{A}\cT_{BC}=6\tilde{\Lambda}\nabla_{A}\hat T^{(2)}_{BC}+4\tilde{\Lambda}(\hat J^{(3)}_{ABC}+\hat J^{(3)}_{ACB})=6\tilde{\Lambda}^3(3\nabla_A\hat B_{BC}-\nabla_B\hat B_{AC}-\nabla_{C}\hat B_{AB}),
    \end{align}
    where we used $\hat J^{(3)}_{ABC}=\dfrac{3}{2}b^{*}\nabla_\Omega^2(\nabla_{A}T_{BC}-\nabla_{B}T_{AC})=\frac{3}{2}\tilde{\Lambda}^2(\nabla_{A}\hat B_{BC}-\nabla_B\hat B_{AC})$.
\end{proof}

The general situation is described by:
\begin{prop}\label{obstr-even} For $d$ even, $\hat \Ee$ describes conformal geometry with an additional field $\cT_{AB}$ subject to the equations of the following structure:
\begin{equation}
    \begin{aligned}
        &(\nabla_{A}\nabla^A)^{\frac{D-5}{2}}\hat B_{BC}+\dots=0,\\
        &\nabla^{A}\cT_{A B}+\sum_{i=0}^{D-4}C^{i}_{D-3}\mathcal{D}_{A}^{(D-3-i)}\hat T^{(i)|A}{}_{B}=0\,.
    \end{aligned}
    \end{equation}
  \end{prop}
The first equation is $\cT_{AB}$-independent and determines a Fefferman-Graham obstruction of the boundary conformal structure to extend to a formal solution to the bulk Einstein equations near the boundary. The second equation is a generalised conservation condition satisfied by the subleading $\cT_{AB}$. 
\begin{proof}
Only the first formula needs to be proved. Since the conformal-geometry sector is generated by the action of $\nabla_A$, it admits a natural filtration which, roughly speaking, counts the number of derivatives (this interpretation will be precisely valid on the solution space of the gPDE). In the present proof we will monitor only the leading order with respect to this filtration and everything below will be denoted by $\dots$. The case $D=5$ was already analyzed and is excluded. For $D>5$ from Theorem~\bref{bound-calculus} we have  (here $2\leq N\leq D-4$):
\begin{align}
            \nabla^{A}\hat J^{(N)}_{ABC}=\frac{N}{N-1}\nabla^{A}\hat b^{*}\nabla_{\Omega}^{N-1}(\nabla_{A}T_{BC}-\nabla_B T_{AC})=\frac{N}{N-1}(\nabla^{A}\nabla_A\hat{T}^{(N-1)}_{BC}-\nabla^{A}\nabla_{B}\hat{T}^{(N-1)}_{AC})+\dots.
\end{align}
We can rewrite 
\begin{align}
    \nabla^{A}\nabla_{B}\hat{T}^{(N-1)}_{AC}=\nabla_{B}\nabla^{A}\hat T_{AC}^{(N-1)}+[\nabla^{A},\nabla_B]\hat T_{AC}^{(N-1)}=\nabla_{B}\nabla^{A}\hat T_{AC}^{(N-1)}+\dots\,.
\end{align}
But it follows from the proof of Theorem~\bref{prop:mnogoidealov} that the divergence $\nabla^{A}\hat T_{AB}^{(N-1)}=-\sum_{i=1}^{N-1}C^{i}_{N-2}\mathcal{D}^{(N-1-i)|A}\hat T_{AB}^{(i)}$, and hence it itself belongs to lower-order corrections. Consequently,
\begin{align}
    \nabla^{A}\hat J^{(N)}_{ABC}=\frac{N}{N-1}\nabla^{A}\nabla_A\hat{T}^{(N-1)}_{BC}+\dots\,.
\end{align}
Using the formula for $\hat T^{(N)}_{AB}$ from Theorem~\bref{bound-calculus}, we conclude that (here $\propto$ denotes equality up to a nonzero coefficient)
\begin{align}
    \hat T^{(N)}_{BC}\propto \hat b^{*}\nabla_{\Omega}^{N-1}\nabla^{A}J_{ABC}=\nabla^{A}\hat J^{(N-1)}_{ABC}+\dots\propto\nabla^{A}\nabla_{A}\hat T^{(N-2)}_{BC}+\dots\,.
\end{align}
It follows that

\begin{align}
   \cO_{BC}=\nabla^{A}\hat J^{(D-4)}_{ABC}+\dots\propto \nabla^{A}\nabla_A \hat T^{(D-5)}_{BC}+\dots\propto (\nabla^{A}\nabla_{A})^2\hat T_{BC}^{(D-7)}+\dots\,.
\end{align}
This procedure continues until (depending on the dimension) it terminates at $\hat T^{(2)}_{AB} \propto \hat B_{AB}$.
    \end{proof}

Finally, similarly to Proposition ~\bref{prop:odd-vanish}, we have that:
\begin{prop}\label{prop:even-vanish}
For $d$ even, the following functions and vector field on $\hat \Ee$
\begin{equation}
\hat T^{(2i+1)}_{AB}\,, \qquad \hat J^{(2i)}_{A BC}\,, \qquad \cD^{(2i+1)}_{A}\,, \quad i \geq 0\,.
\end{equation}
vanish at zero locus of $\hat \cI$\footnote{Note that for $2i+1\leq D-4$, these functions and vector fields vanish exactly, as follows from Proposition~\bref{utv-nechet}.}. 
\end{prop}

\subsection{Solution space of the boundary theory}\label{sec:gravity-sections}

The on-shell boundary GR is described by the gPDE $(\cE,Q)$ obtained in the previous Section and it encodes the equations of motion and gauge transformation of the underlying system in the usual way, i.e. $\sigma$ is a solution if $\dx\circ \sigma^*=\sigma^*\circ Q$, see Section \bref{sec:prelim}. However, this condition typically contains  an infinite number of equations on the components of $\sigma$ and in this form is not very useful in practice. It is much preferable to work in terms of $\hat \Ee$, which is almost as simple as a jet-bundle. More precisely, if we eliminate the metric $\hat g_{AB}$
and the symmetric part of $\hat C_B{}^C$ then  solutions of
$\hat \Ee$ are 1:1 with normal Cartan connections describing the conformal geometry along with the unconstrained conformal field $\cT_{AB}$ defined on this background. In the field-theory terminology we are talking about the off-shell system. At the same time, keeping the metric alive turns out to be very convenient and corresponds to describing the same data in terms of not necessary orthonormal frame and, in particular, in the coordinate frame.

To be more precise let us concentrate on the sector of conformal geometry described by $\Ee^{conf}$. The action of $Q$ is determined by~\eqref{Q-gran-background} which we again list here\footnote{In this section, we restrict ourselves to the case $D\geq5$; the case $D=4$ can be carried out in exactly the same way and uses the action of $Q$ given in \eqref{Q-d3}.}:
\begin{equation} 
\label{Q-gran-background-copy}
\begin{aligned}
        Q \hat g_{BC}=&~\hat C_{B}{}^A \hat g_{AC}+\hat C_{C}{}^A \hat g_{BA}+2\hat \lambda \hat g_{BC},\\
         Q\hat\lambda_A=&~\hat C_{A}{}^{B}\hat \lambda_B  +\dfrac{1}{2}\hat \xi^B \hat \xi^C \hat \Co_{ABC},\\
        Q\hat \lambda=&~\hat \xi^A \hat \lambda_A,\quad Q\hat \xi^B=\hat \xi^A \hat C_{A}{}^{B}\\
        Q \hat C_{B}{}^C=&~\hat C_{B}{}^A \hat C_{A}{}^C+\hat \lambda_B\hat \xi^C-\hat \lambda^C \hat \xi_B+\delta_{B}^C \hat \lambda_A \hat \xi^A +\half \hat \xi^A\hat \xi^D \hat \We^{C}{}_{BAD}\,\\
          Q\hat \We^{C}{}_{BAD}=&~\hat \xi^E \nabla_{E}\hat \We^{C}{}_{BAD}-\hat C_{E}{}^{C}\hat \We^{E}{}_{BAD}+\hat C_{B}{}^{E}\hat \We^{C}{}_{EAD}+\\
        ~&~\hat C_{A}{}^{E}\hat \We^{C}{}_{BED}+ \hat C_{D}{}^{E}\hat \We^{C}{}_{BAE}\,,
\end{aligned}
\end{equation}
and $\commut{Q}{\nabla_A}=0$. We introduce a specal notation for the frame field parametrizing configurations for $\xi^A$: $\sigma^{*}\xi^{A}\equiv e^{A}\equiv e^{A}{}_{\mu}\theta^{\mu}$.

Analysing the gauge transformations with parameters associated to ghosts $C^{A}{}_B$ one finds that 
any configuration of $e^A_\mu(x)$ can be locally gauge fixed to $e^{A}{}_{\mu}=\delta^{A}_{\mu}$, where $\mu$ is a base index. Moreover, this completely eliminates the gauge freedom associated to $\bar C_A{}^B$. Of course, this gauge choice simply corresponds to passing to the coordinate frame. In the same way analysing the gauge freedom associated to $\hat \lambda_A$ ghosts, one finds that the field $\sigma^*(\hat \lambda)$ can be set to zero. In what follows we actively use this gauge and call it ``metric-like''. In this gauge $g_{AB}(x)\equiv\sigma^*(\hat g_{AB})$ becomes a tensor field on $\Sigma$ naturally interpreted as metric. Similarly, slightly abusing notation, we will drop the hat when passing to the solution space and $\mathbf{T}_{AB}(x)\equiv \sigma^{*}\cT_{AB}$.

Note that an alternative interpretation of $\Ee^{conf}$ is based on Cartan geometry while the resulting formulation of the underlying gauge-field theory is usually referred to as frame-like one in the literature. This is constructed by setting $g_{AB}(x)=\eta_{AB}$, where $\eta_{AB}$ is the constant Minkowski metric, by exploiting the gauge freedom encoded in the symmetric part of $\bar C_A{}^B$. In this gauge the spectrum of ghosts precisely corresponds to the conformal algebra $o(D-1,2)$ and the respective $1$-form is to be identified with the Cartan connection encoding the conformal structure, i.e. the $o(D-1,2)$-valued 1-form whose curvature is restricted by~\eqref{Q-gran-background-copy}. In addition one typically uses an additional gauge where $\lambda=0$. In particular, in this case $\hat \We^C{}_{BAD}$ is identified with the Weyl tensor in the frame-like formulation of conformal geometry.

Working in the ``metric-like'' gauge $e^{A}{}_{\mu}=\delta^{A}_{\mu}$ (i.e. $\sigma^{*}\hat \xi^{A}=\theta^{A}$) and $\lambda=0$, we introduce a notation for ghost degree $1$ coordinates on $\hat{\Ee}$ $C^{I}$: $\sigma^{*} \hat C^{I}\equiv C^{I}{}_{|B}\theta^{B}$. The equation $\dx \sigma^{*}\hat \lambda=\sigma^{*}Q\hat \lambda$ implies:
\begin{align}
    \lambda_{[A|B]}=0.
\end{align}
In a similar way, $\dx \sigma^{*}\hat \xi^{A}=\sigma^{*}Q\hat \xi^{A}$ implies
\begin{align}\label{solspace:C}
    C_{[A}{}^{B}{}_{|C]}=0.
\end{align}
Furthermore, $\dx \sigma^{*}\hat g_{AB}=\sigma^{*}Q\hat g_{AB}$ can be rewritten as
\begin{align}
    \partial_C g_{AB}=C_{A}{}^{D}{}_{|C}g_{DB}+C_{B}{}^{D}{}_{|C}g_{AD}\,,
\end{align}
together with \eqref{solspace:C}, giving
\begin{align}\label{Gamma}
    C_{A}{}^{B}{}_{|C}=\Gamma_{AC}^{B}{}[g],
\end{align}
where $\Gamma_{AC}^{B}{}[g]$ are the Christoffel symbols constructed from the metric $g_{AB}(x)\equiv \sigma^*\hat g_{AB}$.

Considering next equations $\dx \sigma^{*}\hat C_{B}{}^{C}=\sigma^{*}Q\hat C_{B}{}^{C}$ one finds
\begin{align}
\partial_{[A}\Gamma^{C}_{D]B}=\Gamma^{E}_{B[A}\Gamma^{C}_{D]E}+\lambda_{B|[A}\delta^{C}_{D]}-\lambda^{C}{}_{|[A}g_{D]B}+\frac{1}{2}\sigma^{*}\hat \We^{C}{}_{BAD}.
\end{align}
where we made use of \eqref{Gamma}. Introducing 
the Riemann curvature of the Levi-Civita connection $\Gamma$ the above equations can be rewritten as
\begin{align}
    R^{C}{}_{BAD}[g]=\sigma^{*}\hat \We^{C}{}_{BAD}+2\lambda_{B|[A}\delta^{C}_{D]}-2\lambda^{C}{}_{|[A}g_{D]B}.
\end{align}
This expression is the standard decomposition of the Riemann tensor in terms of the Schouten and Weyl tensors and implies that $W^{B}{}_{CAD}[g]\equiv\sigma^{*}\hat W^{B}{}_{CAD}$ is indeed the Weyl tensor of the metric $g_{AB}(x)$, while $\lambda_{A|B}$ is the Schouten tensor (up to a numerical factor). Indeed, taking a trace gives \
\begin{align}
    R_{BD}[g]=(2-d)\lambda_{B|D}-\lambda^{A}{}_{|A}g_{BD}
\end{align}
which is the standard rewriting of the Ricci tensor in terms of the Schouten tensor and means
\begin{align}
    \lambda_{A|B}=-P_{AB}[g],
\end{align}
where $P_{AB}[g]$ is the Schouten tensor built from the metric. Then equation $\dx \sigma^{*}\hat \lambda_{A}=\sigma^{*}Q\hat \lambda_{A}$ gives 
\begin{align}\label{sections-cotton}
    2\nabla_{[A}^{g}P_{B]C}=-\sigma^{*}\hat \Co_{C AB}\,,
\end{align}
where $\nabla_{\mu}^{g}$ is the usual Levi-Civita covariant derivative. This is the standard  expression of the Cotton tensor $C_{CAB}[g]\equiv \sigma^{*}\hat C_{CAB}$ through the Schouten tensor.

Finally, by considering the equations involving coordinates $\nabla_{(C)} \hat \We^D{}_{ABC}$ one finds that they simply determine the respective fields in terms of the derivatives of the metric.  In other words, $\Ee^{conf}$ describes the unconstrained conformal geometry on the boundary.

\begin{remark}
The above discussion is a particular explicit realisation of the general feature of gPDE approach that the metric-like and frame-like formulations of a given theory can be unified and, moreover, the frame-like formulation can be systematically derived starting from the metric-like. This is achieved by considering the jet-bundle BV formulation of the metric-like system as a gauge PDE and taking its maximal equivalent reduction (often called minimal model), giving the frame-like formulation of the system. In the case of conformal geometry this path is explained in~\cite{Dneprov:2022jyn,Grigoriev:2023kkk}. The general procedure can be found in~\cite{Barnich:2010sw,Grigoriev:2010ic,Grigoriev:2012xg}, see also~\cite{Basile:2022nou}. In the linearised setup the construction was already presented in~\cite{Barnich:2004cr}. Having said this, let us stress, however, that what we highlight now is that the metric-like formulation is also naturally encoded in the slight extension of the minimal model extended by $\hat g_{AB}$ and the symmetric part of the ghosts $\hat C_A{}^B$.
\end{remark} 

Now we get back to the entire $\hat \Ee$ and observe that the analogous considerations also apply to the sector of $\nabla_{(C)}\cT_{AB}$. Because in this sector the ghost variables are absent, there are no additional gauge symmetries while the equations of motion $\dx \sigma^* \cT_{AB}=\sigma^* Q\cT_{AB}$ simply express $\sigma^*(\nabla_{(C)}\cT_{AB})$ in terms of space-time derivatives of $\mathbf{T}_{AB}\equiv \sigma^*(\cT_{AB})$ and the fields of $\Ee^{conf}$.
To summarize, $\hat \Ee$ seen as a gPDE over $T[1]\Sigma$ is off-shell.  More formally, it is equivalent to a non-negatively graded (i.e. no antifields) jet-bundle BV system.

Any vector field on the total space $\hat \Ee$ of our gPDE gives rise to a vector field on the space of sections. Namely, if $V$ is a vector field and $\sigma$ a section then the variation of $\sigma$ under the prolongation of $V$ is given by 
\begin{equation}
\delta \sigma^*=\sigma^* \circ V
\end{equation}
and is interpreted as a value of $V^{prol}$ at the point $\sigma$ of the space of sections. More formally, $\sigma^* \circ V$
is a vector field along $\sigma$. Let us see what  $V^{prol}$ looks like in terms of coordinates on the space of sections, i.e. fields $\phi^i(x,\theta)\equiv \sigma^* \phi^i$, where $\phi^i$ denote all the fibre coordinates of $\hat \Ee$. We have, $\sigma^* V \phi^i=\sigma^* V^i(\phi,x,\theta)=V^i(\sigma^*(\phi),x,\theta)$, where we introduced components $V^i(\phi)\equiv V\phi^i$. 

Let us specialize this construction by taking $\nabla_A$ as $V$.  For an arbitrary  degree-$0$ function $f$ on $\hat \Ee$ we have:
\begin{align}
    Qf=\xi^{A}\nabla_{A}f+C_{A}{}^{B}\Delta^{A}{}_{B}f+\lambda_{A}\Gamma^{A}f+\lambda \Delta f\,.
\end{align}
Because $\commut{Q}{\nabla_A}=0$, vector fields $\nabla_A$ are symmetries and hence $\nabla_A^{prol}$ preserve the solution space,  i.e. they are tangent to the subspace of solutions. If $\sigma$ is a solution, i.e. $\dx \sigma^{*}f=\sigma^{*}Qf$ we have
\begin{multline}
\sigma^{*}(\xi^A\nabla_{A}f)=\sigma^*(QF- 
\hat C_{A}{}^{B}\Delta^{B}{}_{A}f-\hat \lambda_{A}\Gamma^{A}f-\hat \lambda \Delta f)=
\\
=
\dx \sigma^*f-\sigma^*(\hat C_{A}{}^{B}\Delta^{A}{}_{B}f+\hat \lambda_{A}\Gamma^{A}f+\hat \lambda \Delta f)=\dx \sigma^*f-\sigma^*((Q-\hat \xi^A\nabla_{A})f)
\end{multline}
If in addition, the metric-like gauge is imposed one gets:
\begin{equation}
\label{nablaA}
\sigma^*(\nabla_{A}f) =
\dl{x^A}\sigma^*f- \Gamma_{CA}^{B}\sigma^*(\Delta^{C}{}_{B}f)+P_{AB}\sigma^*(\Gamma^{B}f)\,.
\end{equation}
Note that in this form $\sigma^*(\nabla_A f)$ is a useful object even if  $\sigma$ is not a solution.

The first two terms in~\eqref{nablaA} can be regarded as the Levi-Civita covariant derivative $\nabla^g_A$:
\begin{equation}
\label{nabla-full}
\nabla^g_A \sigma^{*}f\equiv 
\dl{x_A}\sigma^*f-\Gamma_{AB}^C \sigma^*(\Delta_{C}{}^{B}f)\,,
\end{equation}
Note that  $\nabla^g_A$ defined this way is, strictly speaking, a vector field along $\sigma$. However, if we use degree-$0$ fibre coordinates $\phi^i_0$ such that $\Delta_C{}^B \phi_0^i=(\Delta_C{}^B)^i_j\phi_0^j$ and the coefficients are field-independent one can reinterpret $\nabla^g_A$ as a usual covariant derivative acting on tensor fields in a given linear representation. As for the
total covariant derivative~\eqref{nabla-full} it can be directly related to the covariant derivative introduced in~\cite{wunsch1986conformally}
and to the Weyl covariant derivative from \cite{Boulanger:2004zf}. At the same time, it can be related to the tractor connection~\cite{Bailey:1994}.

For further applications, it is also useful to be able to rewrite expressions of the form $\sigma^{*}\nabla_{A}\nabla^{A}f$ as Levi-Civita derivatives with correction terms. The following formula holds:
\begin{multline}
\label{nabla-laplacian}
    \sigma^{*}\nabla_{A}\nabla^{A}f
    =\nabla_A^g\nabla^{g|A}\sigma^{*}f+(\partial_AP)\sigma^{*}\Gamma^Af
    \\
    +2P^{AD}\nabla^{g}_{A}\sigma^{*}\Gamma_D f+P^{AB}P_{A}{}^{C}\sigma^{*}\Gamma_B\Gamma_Cf+P\sigma^{*}(-\Delta^{A}{}_{A}+\hat\Delta) f\,,
\end{multline}
where $\hat{\Delta}\equiv [Q,\frac{\partial}{\partial\hat\lambda}]$ and $P\equiv P_A{}^{A}$. To prove this formula, it is sufficient to use equation~\eqref{nablaA} twice, the standard Riemannian geometry identity $\nabla^{g}_{A}P^{A}{}_{B}=\nabla^{g}_{B}P=\partial_B P$, and to commute $\nabla_A$ and $\Gamma_B$ in the same way as in~\eqref{commutrelations}.

Let us mention that the interpretation of degree-$0$ coordinates on the minimal model of the BRST complex as generalised tensor fields and degree-$1$ coordinates as connections was already in~\cite{Brandt:1996mh} in the general setting and was specialized to conformal geometry in~\cite{Boulanger:2004eh,Boulanger:2004zf}. Note however that in the above references the interpretation was done at the algebraic level while here the covariant derivatives act on the sections of the underlying fibre bundle over the spacetime.

We now use the above technique to show that for $D=5$ the function $\hat B_{CA}\equiv\nabla^{B}\hat \Co_{CAB}$ on $\hat \Ee$ indeed corresponds to the Bach tensor on the solution space. Using $\Gamma^{D}\hat \Co_{CAB}=-\hat \We^{D}{}_{CAB}$ and formula \eqref{nablaA} we obtain
\begin{align}
\sigma^{*}\hat B_{CA}=\sigma^{*}\nabla^{B}\hat \Co_{CAB}=\nabla^{g|B}\sigma^{*}\hat \Co_{CAB}+P^{D B}\sigma^{*}\Gamma_{D}\hat \Co_{CAB}=\nabla^{g|B}C_{CAB}-P^{DB}\We_{D CA B}
\end{align}
which is the standard formula for the Bach tensor and we can denote $\sigma^{*}\hat B_{BC}=B_{BC}[g]$. For $D>5$, the same formula provides a natural generalization of the Bach tensor to higher dimensions, which, however, is no longer a conformally invariant tensor, since $\Gamma_{A}\hat B_{BC}=-(D-5)(\hat C_{CAB}+\hat C_{BAC}). $\footnote{Although this formula can be obtained by a direct computation, one can notice that there is a simpler way to derive it using the boundary calculus. Indeed, from equation~\eqref{formuli-calculus} we have $(D-5)\hat T^{(2)}_{BC}=2\tilde{\Lambda}^2\hat{B}_{BC}$, and from ~\eqref{QTact} it follows that $\Gamma_{A}\hat{T}_{BC}=-2\tilde{\Lambda}^2(\hat C_{CAB}+\hat C_{BAC})$.} As we have seen in Example~\bref{example-5D} for $D=5$ $\nabla_{(C)}\hat B_{AB}$ generate the ideal $\hat \cI$ of the boundary gPDE and hence 
$\nabla_{(C)}\hat B_{AB}$ vanish on solutions. The other way around, if $\sigma$ is a solution to $\hat \Ee$ such that $\sigma^*(\hat B_{AB})=0$ then $\sigma^*(\nabla_{(C)}\hat B_{AB})=0$. The ``conservation equation" $\sigma^{*}\cO_B=0$ in this dimension takes the form
\begin{align}
    \sigma^{*}\nabla_{A}\cT^{A}{}_{B}=\nabla^{g}_A\mathbf{T}^{A}{}_{B}+2\tilde{\Lambda}^{2}P^{AC}C_{ABC}=0
\end{align}

For $D=7$, as follows from Example~\bref{example-7D}, the ``obstruction" condition on solutions $\sigma^{*}\cO_{BC}=0$ is
\begin{align}
\begin{split}
        0=\sigma^{*}\Big(\nabla^{A}\nabla_{A}\hat B_{BC}+2\We_{DCAB}B^{AD}+2  \Co_{C}{}^{AD} \Co_{BAD}-4 \Co^{D}{}_{B}{}^{A} \Co_{ACD}\Big).
\end{split}
\end{align}
Using formula \eqref{nabla-laplacian} together with expressions for $\Gamma_A B_{BC}$, $\Gamma_A C_{BCD}$ given above, it is straightforward to rewrite this formula as
\begin{align}
\begin{split}
        \nabla^{g|A}\nabla^g_{A}B_{BC}+4(\partial_AP)C_{(BC)}{}^{A}+8P^{AD}\nabla^{g}_{A}C_{(BC)D}+4P^{AD}P_{A}{}^{E}W_{ECDB}-4PB_{BC}+\\+2W_{DCAB}B^{AD}+2  \Co_{C}{}^{AD} \Co_{BAD}-4 \Co^{D}{}_{B}{}^{A} \Co_{ACD}=0
        \end{split}
\end{align}
This is precisely the formula for the obstruction tensor in this dimension from~\cite{Fefferman:2007rka}. The condition $\sigma^{*}\cO_{B}=0$ takes the form
\begin{align}
    \nabla^{g}_A\mathbf{T}^{A}{}_{B}+6\tilde{\Lambda}^3P^{AD}\sigma^{*}(2\nabla_{D}\hat B_{BA}-\nabla_{B}\hat B_{DA})=0
\end{align}
where $\sigma^{*}\nabla_{A}\hat B_{BC}$ can again be rewritten using ~\eqref{nablaA}.


\section{Boundary structure of (gauge)fields}\label{sec:matter}
In the previous section we have described the boundary structure of the asymptotically AdS gravity. Now we extend the framework to cover (gauge) fields coupled to the gravity background. For simplicity, we call these fields ``matter fields'' having in mind that they can have their own gauge symmetries so the case of e.g. Yang-Mills theory on the gravity background is covered. 

\subsection{Gauge fields over boundary GR}
\label{sec:matter-general}
Let us recall that our approach to boundary structure of gravity is to equivalently reformulate it as a conformal-like gravity in the interior and then extend it to the boundary, requiring fields to be well-defined there and, finally, to set the boundary defining field $\Omega$ to vanish in a way compatible with gauge symmetry.

We now extend the approach to matter fields defined on gravity background and start with the combined system consisting of gravity and matter.  More precisely,  we limit ourselves to the case where matter fields do not source gravity so that Einstein equations remain intact. At the same time the equations of motion and gauge transformations of matter fields do depend on the background. In the gPDE language this means that we are dealing with a gPDE over background:
\begin{equation}
(\Ecl^{\phi}, Q,\cI_{cl}^\phi)\;\overset{\pi^\phi}\longrightarrow\; (\Ecl,Q,\cI_{cl})\;\overset{\pi}\longrightarrow\; (T[1]X, \dx)\,,
\end{equation}
where $(\Ecl,Q,\cI_{cl})$ denotes the gPDE for conformal-like GR introduced in Section \bref{sec:conformal-like}. As always in this work, to keep conventions concise we use $Q$ to denote both the $Q$-structure on the total space and its projection to the background gPDE. The same applies to restrictions of $Q$ to various submanifolds we encounter. Note $(\Ecl^{\phi}, Q,\cI_{cl}^\phi)$ is naturally a bundle over $T[1]X$ with the projection being $\pi \circ \pi^\phi$ and can also be considered as an implicit gPDE. 

Let us recall that the boundary gPDE $\Ee_B$ for gravity, which  was constructed in the preceding Sections, is obtained in the following steps: i) restricting to the sub-gPDE $(\Ered,Q)$ by equivalent reduction ii) pulling-back  $(\Ered,Q)$ to the boundary $T[1]\Sigma$ and setting  $\Omega=0, Q\Omega=0$ followed by an additional equivalent reduction. All in all this gave an embedding (of $Q$-manifolds) 
\begin{equation}
    (\Ee_B,Q) \hookrightarrow (\Ecl,Q)\,.
\end{equation}
The ideal $\cI_B$ on $\Ee_B$ defining the on-shell system is obtained as a pullback of the initial $\cI_{cl}$, giving a gPDE $(\Ee_B,Q,\cI_B)$ that describes the boundary structure of asymptotically AdS GR. We refer to the  local gauge theory described by $(\Ee_B,Q,\cI_B)$ as to the boundary GR. Finally, by solving part of the equations from $\cI_B$ one arrives at the equivalent formulation of the boundary GR
as an implicit  gPDE $(\hat\Ee,Q,\hat \cI)$, whose advantage is that  $(\hat\Ee,Q)$ is off-shell and its structure can be described explicitly in terms of the coordinates on $\hat \Ee$ using the boundary calculus constructed in Section~\bref{sec:boundary-calculus}.

The embedding of gPDE $(\Ee_{B},Q,\cI_B)$ for boundary GR  into $(\Ecl,Q,\cI_{cl})$ induces the pullback bundle:
\begin{align}
\Ee^{\phi}_{B} \;\longrightarrow \Ee_{B}
\end{align}
which we interpret as the boundary system for the matter field $\phi$ on the background of boundary GR. Again, just like $(\Ee_B,Q)$, gPDE $(E^\phi_B,Q)$ comes equipped with the ideal $\cI^\phi_B$ representing what remains of the equations of motion after the restriction to the boundary and the equivalent reductions. More precisely, $\cI^\phi_B$ is generated by $\cI_B$ pulled back to $\Ee^{\phi}_{B}$ and the prolongation of the equations of motion for matter fields. All in all the gPDE describing the boundary structure is $(\Ee^\phi_B,Q,\cI^\phi_B)$. Of course, in order to study what it describes a substantial reformulation of this system is in order.

In these considerations there is an important potential problem to be fixed. In order for the reduction from $\Ecl$ to $\Ered$ and then to $\Ee_B$ to work we had to rescale the metric $g$ by $\Omega^2$ and then set $\Omega=0$. This could result in negative powers of $\Omega$ appearing in the equations of motion for matter fields after the pull-back to $\Ee_B$. The same applies to the $Q$-structure in the sector of matter fields. Moreover, we should be able to prescribe the asymptotic behaviour of the matter fields.

The way to cure this is similar to what we did in deriving $(\Ee_B,Q,\cI_B)$. Namely, we reformulate the initial equations of motion in such a way that they remain equivalent for $\Omega>0$ while their restrictions  to $\Omega=0$ become well-defined. 
More precisely, let $\tilde \phi,\tilde c$ be a collective notation for matter fields and corresponding ghost variables (if any) on the GR background described by metric $\tilde{g}$. The equations of motion for $\tilde\phi$ are assumed to be diffeomorphism-invariant and are denoted by $\tilde P[\tilde \phi;\tilde{g}]=0$. What we actually need is $(\Ecl^\phi,Q,\cI_{cl}^\phi)$ which is defined to be a lift of this matter system from GR background to the background of conformal-like reformulation of GR, i.e. the extension of GR by the Stueckelberg pair $\{\Omega,\lambda\}$. The  lift is constructed  by declaring $\tilde\phi$ to be inert under the Weyl transformations encoded in $Q$. To control the boundary behaviour of the matter we repeat the trick employed for gravity and switch to the coordinate system $\phi\equiv \Omega^{w}\tilde \phi$, $g\equiv \Omega^2\tilde{g}$, where $w\in \mathbb{R}$. This $w$ parametrizes the behaviour of our matter fields near the boundary and for the moment we leave it generic.

Let us now turn to our equations of motion $\tilde P[\tilde\phi, \tilde g]$ written in terms of the new coordinates and try to express $\tilde P[\phi,g]$ as a polynomial in $\Omega$ times a singular prefactor:
\begin{align}
\label{P-exp-Omega}
    \tilde P[\phi;g,\Omega]=\Omega^{-w_P}\sum_{k=0}\Omega^{k} \tilde P_k(\phi,D\phi;g,Dg,D_{a\ldots}\Omega)\,, \quad w_P\geq 0
\end{align}
where $P_0$ is nonvanishing and $\tilde P_k$ may depend on derivatives of $\Omega$, but not on $\Omega$ itself. As we will see shortly the above decomposition exists for most of the standard choices of matter. For the general considerations, we assume from now on that it exists and that only a finite number of $P_k$ are nonvanishing.

Given the decomposition \eqref{P-exp-Omega}, we consider the rescaled constraints $P[\phi;g]\equiv \Omega^{w_P}\tilde P[\phi,g]$ which are equivalent to $\tilde P$ when $\Omega>0$ and are well-defined on the boundary.  However, this does not imply that the entire system is well-defined at $\Omega=0$ because the gauge transformations for matter fields may still be ill-defined there (in gPDE terms, this means that the action of $Q$ on $\phi$ contains negative powers of $\Omega$). Typically, this can be fixed by adjusting weight $w$. Moreover, in the case of YM theory, the requirement that $Q$ is well-defined on the boundary, uniquely determines $w$.  That gauge symmetry fixes the conformal weight of gauge fields is well-known. In what follows we assume that $w$ is chosen in such a way that both the rescaled equations and  $Q$ are well defined. All in all
this gives a  gPDE $(\Ee^\phi_B,Q, \cI^\phi_B)$ on the boundary, which is a gPDE over the background gPDE $(\Ee_B,Q,\cI_B)$. This describes the boundary structure of the asymptotically-AdS gravity along with the matter fields defined over it. Note that the boundary structure might generally depend on the choice of $w$ but usually $w$ is essentially fixed by the consistency of the above procedure.

In the rest of this Section we show how the technical tools introduced to study the boundary structure of gravity extend to the case where matter fields are also present.  Not to complicate the exposition, we restrict ourselves to the case where matter fields have closed gauge algebra and no ghosts for ghosts are present. In our approach this means that we can restrict ourselves to the case where the fibre of $\Ecl^\phi$ has coordinates of degree $0$ and $1$ only, which we denote by $D_{(a)}\phi^i$ and $c^\alpha$, respectively.

The matter equations of motion generate the $Q$-invariant ideal $I(P[\phi,g],D_{a})$. Then, the equivalent reduction $r:\Ered \hookrightarrow \Ecl$ in the gravity sector, described in Section \bref{sec:pre-min}, induces a pullback of the bundle $\Ecl^{\phi}\rightarrow \Ecl$ to $\Ered$, which we denote by $\Ee^\phi_{\red}$. It inherits the gPDE structure  $(\Ee^\phi_{\red},Q,\cI^\phi_{\red})$. By some abuse of notations we also use $r$ to denote the embedding $r:\Ee_{\red}^\phi \to \Ecl^\phi$, for instance: $\cI^\phi_{\red}=r^*\cI^\phi_{cl}$. The following statement is an immediate generalisation of Proposition~\bref{prostoe2} and Corollary~\bref{cor-Omega}.
\begin{prop}
If $(Q-\xi^{a}D_a)\phi\in I^0(\phi)$ then as fibre coordinates on $\Ee^{\phi}_{\red}\rightarrow \Ee_{\red}$ one can take
	\begin{align}
		\{\nabla_{(a)} (r^{*}\phi),\nabla_{(a)}(r^{*} c^{\alpha})\}\,,
	\end{align}
where $\nabla_a\equiv [Q,\frac{\partial}{\partial \xi^{a}}]$.\footnote{Recall, that we use the same notations for the ghosts and their restrictions to $\Ered$. The same applies to $\Ee^\phi_{\red}$.} If $(Q-\xi^{a}D_a)P[\phi,g]\in I^0(P[\phi,g])$ the pullback of the ideal $\cI_\phi$ to $\Ee^\phi_{\red}$ takes the form
	\begin{align}
		r^{*} \cI_\phi= I^{\infty}(r^{*}P[\phi, g], \nabla_a).
	\end{align}
\end{prop}
\noindent

In fact, most of the statements proved for $\Ee_{\red}$ remain true for $\Ee^\phi_{\red}$. One subtlety is that
in addition to restricting to $\Ered$ it is extremely convenient to perform the additional equivalent reduction in the sector of fibre coordinates in order to end up with a finite number of degree $1$ coordinates. This does not bring extra complications and we assume that $\Ee^\phi_{\red}$ is, if necessary, the reduced one, and we keep denoting the remaining fibre coordinates of degree $1$ by $c^\alpha$. Another subtlety is that the fibre ghost variables $c^\alpha$ are to be considered on the equal footing with all the ghost degree $1$ coordinates on $\Ered$. In particular, it is convenient to consider vector field $[Q,\dl{C^I}]$ where now $C^I$ denote all the ghost variables including $c^\alpha$ and these vector fields are defined on $(\Ee^\phi_{\red},Q)$ rather than its base $(\Ered,Q)$. Note that $[Q,\dl{C^I}]$ are projectable with respect to $\Ee^\phi_{\red} \to \Ered$ because $\Ee^\phi_{\red} \to \Ered$ is a $Q$-bundle and $\dl{C^I}$ is projectable by construction. This justifies using the same notations for some of these vector fields. In particular:
\begin{equation}
 \Gamma^a\equiv \Big[\frac{\partial}{\partial \lambda_a}, Q\Big]\,,
 \qquad \mathcal{D}^{(N)}_A\equiv ad_{\nabla_\Omega}^{N}(\nabla_A)\,.
 \end{equation}

Repeating the analysis of Sections ~\bref{subsec: adapted} and ~\bref{sec: boundary system} for the case of $\Ee^\phi_{\red}$,  denoting by $\Ee^\phi$ the pullback of $\Ee^\phi_{\red}$ to $T[1]\d X \subset T[1]X$, and dropping $r^{*}$ we arrive at:
\begin{prop}
\label{prop:ideal-b}
    If $(Q-\xi^{a}\nabla_a)\phi\in I^0(\phi)$, the fibre coordinates of $\Ee^{\phi}\rightarrow \Ee$ can be chosen as 
    \begin{align}
        \{\nabla_{(A)}  \phi^{(N)}, c^{\alpha}|N\geq0\}
    \end{align}
where, as before, we use the conventions: $ f^{(N)}\equiv \nabla_\Omega^{N}f$. The pullback of the ideal $\cI^\phi_{\red}$ to $\Ee^\phi$ is generated by $\cI$ and 
\begin{equation}
 \nabla_{(A)} P^{(N)}\,, \quad N\geq 0\,. 
    \end{equation}      
\end{prop}

The next step in the analysis of the boundary structure involves 
restricting $\Ee^\phi$ to $\hat \Ee \subset \Ee$, where we have an explicit description of the boundary GR and solving a part of the matter equations of motion accounted in $\cI^\phi$. This leads to the final gPDE $(\hat \Ee^\phi,Q,\hat \cI^\phi)$ which describes the boundary structure of gravity + matter fields in the asymptotically AdS setup. The explicit structure of $(\hat \Ee^\phi,Q,\hat \cI^\phi)$ substantially depends on what type of matter fields is being considered and is discussed in detail for a few  concrete systems in the remaining Sections.

\subsection{Scalar field}\label{sec:GJMS}
In this section we illustrate our approach using the example of scalar field subject to the Klein-Gordon (KG) equation. As a fibre of $\Ecl^\varphi$ we take a fibre of a jet-bundle for a scalar field with coordinates $D_{(a)}\tilde \varphi$ so that the KG equation reads as:
\begin{equation}
\tilde P[\tilde \varphi,\tilde g]=\frac{1}{\sqrt{|\tilde g|}}D_a(\sqrt{|\tilde g|}\tilde g^{ab} D_b\tilde\varphi)+m^2 \tilde\varphi\,.
\end{equation}
According to the general procedure explained in the previous Section we can now switch to the fibre coordinates $D_{(a)}\varphi$, where $\varphi=\Omega^{w}\tilde\varphi$. In the new coordinate system we have
\begin{align}
    Q\varphi=\xi^{a}\nabla_a\varphi+w \lambda \varphi\,,
\end{align}
where we employed $\nabla_a$ instead of $D_a$ (note that in the usual coordinates on $\Ecl^\varphi$ we have $\nabla_a\varphi=[\dl{\xi^a},Q]\varphi=D_a\varphi $). It is straightforward to verify that because it is written in terms of $\nabla_a$, this formula remains valid after we pull-back our bundle to $\Ered\subset \Ecl$.

To proceed further we pullback $\Ecl^\varphi$ to $\Ered\subset \Ecl$
and then pullback the resulting gPDE (as a bundle over $T[1]X$) to the boundary $T[1]\Sigma$. This results in the gPDE $(\Ee^\varphi,Q)$ over $(\Ee,Q)$. Then the Klein-Gordon equation written in terms of $g,\varphi$ takes the form $\tilde P=\Omega^{-w}P$, where 
\begin{equation}\label{phi-P}
\begin{gathered}
P[\varphi,g]\equiv \Omega^{2}\nabla_{a}\nabla^{a}\varphi+\Omega(c_{w,1}\nabla_a \varphi \nabla^{a}\Omega- w (\nabla_a\nabla^a\Omega)\varphi)+(w(d+w)\nabla_a\Omega \nabla^{a}\Omega +m^2)\varphi,\\
   c_{w,N}\equiv -2w-d+N
\end{gathered}
\end{equation}
and $d\equiv D-1$ denotes the boundary dimension. Note that $P$ is polynomial in $\Omega$, and hence the assumption is fulfilled. This form of KG equation is a version of what is known in the literature as a Weyl-covariant formulation; see, e.g. \cite{Gover:2008pt,Gover:2008sw}. The prolongation of $P$, i.e. functions $\nabla_{(a)}P$, generate the ideal which together with the pull-back of $\cI$ from the base $\Ee$ form a total ideal $\cI^\varphi$ on $\Ee^\varphi$. In other words we have reformulated our system as a gPDE $(\Ee^\varphi,Q,\cI^{\varphi})$
over background gPDE $(\Ee,Q,\cI)$. 

Just like in the analysis of the gravity sector, we next split the index as $\{a\}=\{\Omega,A\}$ and switch to the adapted fibre coordinates $\{\nabla_{(A)}\varphi^{(N)},\,N\geq0\}$ instead of $\{\nabla_{(a)}\varphi\}$ by making use of $(Q-\xi^{a}\nabla_a)\varphi\in I^{0}(\varphi)$. As usual, we also use notations $f^{(N)}\equiv \nabla_\Omega^{N}f$.

At the end of the day we are interested in $\Ee^\varphi$ pulled back to the on-shell background gPDE $(\cE,Q)$. However, it is not convenient to work with $\cE$ and we treat it implicitly.  More precisely, we identify functions on $\Ee^\varphi$ pulled back to $\cE \subset \Ee$ as equivalence classes modulo the ideal $\cK$ which is generated by the pullbacks to total space $\Ee^\varphi$ of functions on the base $\Ee$ vanishing on $\cE$.  Of course, this is a slight abuse of conventions but it does not lead to confusions and we do the same with its sub-ideals $\cK^{(N)}$. In practice, we analyse the prolongation of KG equations leaving the terms from $\cK$ implicit. In particular, in this way we can easily evaluate $Q\varphi^{(N)}$ as in Proposition ~\bref{App-QT}, giving  
    \begin{align}
    \label{prop: KG-Q}
        &Q \varphi^{(N)}=\xi^{A}\nabla_A \varphi^{(N)}+\tilde{\Lambda}\lambda^{A}\sum_{i=0}^{N-2}d^{i}_N \mathcal{D}_{A}^{(N-2-i)} \varphi^{(i)}+(w-N)\lambda \varphi^{(N)}+\cK^{(0)}\,.
    \end{align}

We now turn to the analysis of the 
ideal of the KG equation, which is generated by $\{\nabla_{(a)}P\}$. Just like in the case of GR it is convenient to use $\{\nabla_{(A)}P^{(N)},\,N\geq0\}$ as generators of the ideal. In particular, one has:
\begin{equation}
P^{(0)}=(\tilde{\Lambda}w(d+w)+m^{2})\varphi^{(0)}+\cK^{(0)}\,,
\end{equation}
where, as before, we do not explicitly spell out terms from $\cK$. This generator is trivial if the mass $m$ and Weyl weight $w$ satisfy the familiar relation:
\begin{equation}
\label{m-w-rel}
m^{2} =-\tilde{\Lambda}w(d+w) \,.
\end{equation}
Although it is possible to give a consistent interpretation of the boundary system even if \eqref{m-w-rel} does not hold, we refrain from doing this and in what follows assume~\eqref{m-w-rel}. Acting with $(\nabla_\Omega)^{N}$ on \eqref{phi-P} and using $\nabla_{\Omega}^{N}\nabla_a\nabla_b\Omega\in \cK^{(N)}$ we obtain the explicit form of the subleading equations:
\begin{align}\label{PN}
\begin{split}
        &\frac{1}{\tilde\Lambda N} P^{(N)}=c_{w,N}\varphi^{(N)}+(N-1)\tilde{\Lambda} \nabla_\Omega^{N-2} \nabla_{A}\nabla^A\varphi+\cK^{(N-1)}, \quad N\geq 1.
\end{split}
\end{align}

It is clear  that the structure of the system described by the zero locus of $\nabla_{(A)}P^{(N)}$ critically depends on whether $c_{w,N}\equiv -2w+1-D+N$ vanishes for some $N=N^*$ or not. If it does not, all the fibre coordinates $\nabla_{(A)} \varphi^{(N)}$ can be eliminated using $\nabla_{(A)}P^{(N)}$ and hence the only remaining independent fibre coordinates on the zero locus of the ideal are $\nabla_{(A)}\varphi^{(0)}$. It follows, the resulting system  describes the unconstrained scalar field of Weyl weight $w$, defined over the background of $\cE$. A more careful analysis shows that the same result holds even for a more general background $\hat E$. 
Moreover, because $Q\varphi^{(0)}$
depends on $\varphi^{(0)}$ and ghosts only, $\varphi^{(0)}$ is indeed a conformal primary.

We now turn to the more interesting case 
where $c_{w,N^*}=0$ for some positive integer $N^*$. 
The analysis is then analogous to the case of gravity. In particular, it turns out very useful to introduce an analogue of the off-shell boundary gPDE $\hat \Ee$. This is done in two steps. First, we restrict our bundle to $\hat\Ee \subset \Ee$ because it is $\hat \Ee$
where we have a full control over the gravity background. Second, we pass to the subbundle  $\hat E^\varphi$ of $E^\varphi|_{\hat \Ee}$, defined as the zero locus of  the following  equations:
        \begin{align}
            \varphi^{(N)}=-\frac{N-1}{c_{w,N}}\tilde{\Lambda}\nabla_\Omega^{N-2}\nabla_A\nabla^{A}\varphi, \quad N\geq 1,\,N\neq N^{*}\,.
        \end{align}
Strictly speaking, the above equations are defined on $\Ee^\varphi$ and it is their pullback to $E^\varphi|_{\hat \Ee}$
that defines a subbundle $\hat E^\varphi \subset E^\varphi|_{\hat \Ee}$. 
In practice, for the moment we work with $\hat \Ee^{\varphi}$ as a submanifold in $\Ee^\varphi$ and we denote by $\hat b_{\varphi}$ the corresponding embedding $\hat E_B^{\varphi}\hookrightarrow E^{\varphi}_{\red}$. 
For any $f\in \cC^{\infty}(\Ee^\varphi)$  we use the notation  $\hat f\equiv \hat b_{\varphi}^*f$.

Similarly to the gravity case discussed in Section~\bref{sec:boundary-gPDE}, we encounter that $Q$ is not tangent to $\hat \Ee^\varphi$ seen as submanifold in $\Ee^\varphi$. However, as before, we can naturally define a $Q$-structure on $\hat E^\varphi$ but this time in a slightly weaker sense. Namely, this $Q$-structure is well defined modulo the ideal $\hat \cI \equiv \hat b^* \cI$ that defines the on-shell system on $\hat \Ee$. This does not complicate the general analysis as at the end of the day we are going to study our matter field coupled to the on-shell background and hence terms from $\hat \cI$ vanish on such configurations. To equip $\hat \Ee^\varphi$ with a $Q$-structure we restrict to the subbundle $\Ee^{\varphi,N^*} \subset \Ee^{\varphi}$, singled out by $\nabla_{(A)}P^{(N)}=0$ with $N<N^*$. The result is a $Q$-subbundle thanks to a straightforward extension of Proposition \bref{mnogoidealov}. Furthermore,  for $N=N^*$ equation \eqref{prop: KG-Q} implies that $Q \varphi^{(N^*)}$ does not depend on $\varphi^{(k)}$ with $k>N^*$ so that the subalgebra of functions in  $\nabla_{(A)}\varphi,\nabla_{(A)}\varphi^{(N^{*})}$ and coordinates from $\hat\Ee$ is closed under $Q$. This means that $\Ee^{\varphi,N^*}$ is a $Q$-bundle over another $Q$-manifold that we denote by $\hat \Ee^{\varphi\prime}$. Finally, the subbundle $\hat \Ee^\varphi$ (this can be obtained by imposing  $\nabla_{(A)}P^{(N)}=0$ with $N>N^*$) is a section of $\Ee^{\varphi,N^*} \to \hat \Ee^{\varphi\prime}$ and hence
the natural $Q$-structure from $\hat \Ee^{\varphi\prime}$ can be pushed forward to $\hat \Ee^{\varphi}$. Let us stress that in all these steps $Q$ is assumed to be defined modulo $\hat \cI$. Nevertheless, for $N^*$ sufficiently small one can check that $\cI$ does not enter the relations and hence $\hat \Ee^{\varphi}$ can be equipped with the genuine $Q$-structure.

Repeating once again the logic employed in the analysis of gravity, we introduce a subalgebra of $\hat \Ee^\varphi$-projectable functions in $\cC^\infty(\Ee^\varphi)$, which consists of functions whose restrictions to $\Ee^{\varphi,N^*}$ are constant along the fibres (i.e. can be obtained as pullbacks of functions from $\hat \Ee^{\varphi\prime}$). In particular, for such a function we have $Q \hat b_\varphi^{*} f=\hat b_\varphi^{*}Qf$ modulo $\hat\cI$. It is straightforward to see that $\varphi^{(N)}$ are $\hat \Ee^\varphi$-projectable for $N\leq N^{*}$. It is convenient to introduce the following notations:
\begin{align}\label{psi-notation}
    \psi\equiv (\hat b_\varphi)^{*}\varphi^{(N^{*})},\qquad \mathcal {P}\equiv   (\hat b_\varphi)^{*}\nabla_\Omega^{N^{*}-2}\nabla_A\nabla^{A}\varphi\,,
\end{align}
where $\cP=0$ for $N^{*}=1$. The fibre coordinates on $\hat \Ee^\varphi$ are then taken to be $\{\nabla_{(A)}\hat\varphi, \nabla_{(A)}\psi\}$ and the ideal $\hat \cI^{\varphi}\equiv \hat b_\varphi^{*}\cI^{\varphi}$ is now generated by $\nabla_{(A)}\cP$ and $\hat\cI$. Recall that $\hat \Ee^\varphi$ is a bundle over $\hat \Ee$. The resulting implicit gPDE $(\hat \Ee^\varphi,Q,\hat \cI^{\varphi})$ describes the boundary structure of the KG field on the background of on-shell boundary gPDE for GR. Of course, in the analysis of KG field one can disregard the gravity gPDE by pulling back $\hat \Ee^\varphi$ to a given on-shell configuration of the boundary GR, i.e. a particular solution of $(\hat \Ee,Q,\hat \cI)$.

By representing $\hat b_\varphi^{*}(\nabla_\Omega)^N\nabla_A\nabla^A \varphi$ in terms of $\cD_A^{(N-k)}\hat\varphi^{(k)}$ one observes  that $\cP$ can be expressed (modulo $\hat\cI$) in terms of the collection of functions $\hat\varphi^{(k)}$, $1\leq N<N^{*}$ together with functions $\hat J^{(N)}_{ABC}$, $ \hat T^{(N)}_{AB}$, which enter through vector fields  $\cD_A^{(N)}$.  This constitutes the extension of the boundary calculus introduced in Section~\bref{sec:boundary-calculus} to the case where on top of gravity we have a KG field coupled to it. Note that $\hat J^{(N)}_{ABC}$, $ \hat T^{(N)}_{AB}$ are now seen as function on $\hat \Ee^\varphi$, i.e. the pullbacks of the corresponding functions from $\hat \Ee$ by the canonical projection $\hat \Ee^\varphi \to \hat\Ee$. As for the vector fields $\cD_A^{(N)}$, they are now defined on $\hat \Ee^\varphi$ from the very start as the restriction of $(ad_{\nabla_\Omega})^N\dl{\xi^A}$ from $\Ee^{\varphi}$ to $\hat \Ee^\varphi$ and are generally defined modulo $\hat \cI$. However, all their properties remain true as the expression of  $\cD_A^{(N)}$ in terms of $Q$ and ghosts is unchanged.

Let us summarize the scalar field extension of the boundary calculus, introduced in Theorem~\bref{bound-calculus}:
\begin{prop}\label{calculus-scalar}
Functions $\hat\varphi^{(N)} \in \cC^\infty(\hat \Ee^\varphi)$ given by
    \begin{align}\label{formula-phiN}
            \hat\varphi^{(N)}=-\frac{N-1}{c_{w,N}}\tilde{\Lambda} (\hat b_\varphi)^{*}\nabla_\Omega^{N-2}\nabla_A\nabla^{A}\varphi, \quad 1\leq N<N^{*}
        \end{align}
can be recursively determined (modulo $\hat \cI$) in terms of $\hat \varphi\equiv\hat \varphi^{(0)}$  and $T^{(N)}_{AB}$, $J_{ABC}^{(N)}$ by
\begin{align} \label{1-kg}
        \hat b_\varphi^{*}\nabla_\Omega^{N}\nabla_Af=\sum_{i=0}^{N}C^{i}_{N}\mathcal{D}_{A}^{(N-i)}\hat f^{(i)}=\nabla_A \hat f^{(N)}+\sum_{i=0}^{N-1}C^{i}_{N}\mathcal{D}_{A}^{(N-i)}\hat f^{(i)}
\end{align}
and 
\begin{align}\label{phi-D}
        &\mathcal{D}_{A}^{(N)}f=\Big(-\hat J^{(N-1),D}{}_{CA}\Delta^{C}{}_{D}+\dfrac{1}{\tilde\Lambda N}\hat T^{(N)}_{CA}\Gamma^{C}-\frac{1}{N}\sum_{i=1}^{N-2} d^{i}_N \hat T^{(i),C}{}_{A}\mathcal{D}_{C}^{(N-2-i)}\Big)f\,, \quad  N\geq1\,.
    \end{align}
\end{prop}
\begin{proof} 
Equation \eqref{1-kg} is just a version of Proposition~\bref{commut-nablas} while \eqref{phi-D} is just the direct extension to $\hat E^{\varphi}_B$
of the relations \eqref{formuli-D} from Theorem~\bref{bound-calculus}. Indeed, because we consider the background to be on-shell, Remark~\bref{remark-higherN} implies 
that  there is no restriction on $N$.
\end{proof}
As an illustration of the procedure, let us compute $\hat b_\varphi^{*}\nabla_{\Omega}^{N}\nabla_{A}\nabla^{A}\varphi$:
\begin{align}\label{Laplacian-exp}
        \hat b_\varphi^{*}\nabla_{\Omega}^{N}\nabla_{A}\nabla^{A}\varphi=\sum_{i=0}^{N}C^{i}_{N}\mathcal{D}_{A}^{(N-i)} (\hat b_\varphi)^{*}\nabla_{\Omega}^{i}\nabla^A\varphi=\sum_{i=0}^{N}C^{i}_{N}\mathcal{D}_{A}^{(N-i)}\sum_{j=0}^{i}C^{j}_i\mathcal{D}^{(i-j)|A}\hat\varphi^{(j)}\,.
    \end{align}
    Here and in the rest of this Section all equalities are assumed to hold modulo $\hat \cI$. Note that for odd $N\leq D-4$ we have $\cD^{(N)}_{A}=0$, thanks to the trivial generalization of Proposition~\bref{utv-nechet}.

The action of $Q$ on the fibre coordinates of $\hat \Ee^\varphi$ is also expressed in terms of the above boundary calculus. Indeed, thanks to \eqref{prop: KG-Q} we have:
    \begin{align}
    \begin{split}
    &Q\hat \varphi=\hat \xi^{A}\nabla_{A}\hat \varphi+w\hat \lambda\hat \varphi\,,\\
        &Q \psi=\hat \xi^{A}\nabla_A \psi+\tilde{\Lambda}\hat \lambda^{A}\sum_{i=0}^{N^{*}-2}d^{i}_{N^*} \mathcal{D}_{A}^{(N^{*}-2-i)} \hat \varphi^{(i)}+(w-N^{*})\hat \lambda \psi+\hat \cI\,.
            \end{split}
    \end{align}

Before passing to the examples let us note that for $N^*$ sufficiently small the ideal $\hat \cI$ does not actually enter the boundary calculus. One can already see this in formula \eqref{PN}, where  the terms from $\cK$ in the right-hand side vanish after pulling back to $\hat E_B$ for $N^{*}$  sufficiently small. In this case $\hat\Ee^\varphi$ can be considered a genuine gauge PDE which is off-shell both in the sector of GR and in the matter sector.

Having a boundary calculus at our disposal we now look at a number of particular cases. 
\begin{prop}\label{prop:GJMS-trivial}
    Let  $N^{*}=2w+d$ be an odd positive integer. If $d$ is odd, assume in addition that the background gPDE is also subject to $\nabla_{(A)}\cT_{BC}=0$. Then $\cP$ vanishes identically, i.e. the ideal is trivial, and  the boundary system describes just two unconstrained scalar fields of weights $w$ and $w-N^{*}=-w-d$. Moreover, $\Gamma_{D}\psi=0$, i.e. $\psi$ is also a conformal primary.
\end{prop}
\begin{proof}
By induction, \eqref{formula-phiN} and Propositions \bref{prop:even-vanish}, \bref{prop:odd-vanish}, one can prove that $\hat\varphi^{(N)}=0$ for all odd $N\leq N^{*}~-~2$.
\end{proof}
In the standard conventions used in the context of AdS/CFT, $\varphi$
and $\psi$ are known as leading and sub-leading boundary values. Their conformal weights are $\Delta_-=-w$ and $\Delta_+=d-\Delta_-=w+d$, see e.g.~\cite{Skenderis:2002wp} and references therein.

Before moving on to the case of generic even $N^{*}$ which we  henceforth parametrize as $N^{*}\equiv 2\ell$, let us discuss a couple of examples:
\begin{example}
Conformal Laplacian: $\ell=1$, or equivalently, $-w=\frac{d}{2}-1$. In this case: 
\begin{align}
    \cP=\nabla_{A}\nabla^{A}\hat\varphi\,.
\end{align}
Note that $-w=\frac{d}{2}-1$ indeed corresponds to the correct Weyl weight of the usual conformally coupled scalar field.  Note that the above formula is expressed in terms of covariant derivatives $\nabla_A$ rather than the usual Levi-Civita ones and therefore does not explicitly contain the familiar scalar curvature correction. The correction becomes explicit in the actual equations of motion $\sigma^*(\cP)=0$, where $\sigma$ is a lift of a configuration $\varphi(x)$ to a solution of $(\hat E^\phi_B,Q)$, see Section \bref{sec:gravity-sections} for more details. 
We also observe that 
\begin{align}
    \Gamma_{A}\psi=\tilde\Lambda \nabla_A\hat\varphi
\end{align}
so that, as in the gravitational case, we cannot in general set the subleading to zero. Such a condition is not compatible with the background gauge transformation, i.e. it is not conformally invariant.
\end{example}
\begin{example}
Fradkin--Tseytlin--Paneitz operator \cite{Fradkin:1981iu,Paneitz:1983}:  $d\geq4$, $l=2$. Using 
\eqref{Laplacian-exp} together with \eqref{formula-phiN} we can easily get $\hat\varphi^{(1)}=0$, $\hat\varphi^{(2)}=\frac{1}{2}\nabla_{A}\nabla^{A}\hat\varphi$. Hence, using \eqref{phi-D} we get 
\begin{align}
    \cP=\sum_{i=0}^{2}C^{i}_{2}\mathcal{D}_{A}^{(2-i)}\sum_{j=0}^{i}C^{j}_i\mathcal{D}_{A}^{(i-j)}\hat\varphi^{(j)}=\mathcal{D}^{(2)|A}\nabla^{A}\hat\varphi+\nabla_A\mathcal{D}^{(2)|A}\hat\varphi+\nabla_{A}\nabla^{A}\hat\varphi^{(2)}=\frac{\tilde\Lambda}{2}(\nabla_{A}\nabla^{A})^2\hat\varphi
\end{align}
and
\begin{align}
    \Gamma_{A}\psi=3\tilde{\Lambda}\nabla_{A}(\nabla_{B}\nabla^{B})\hat\varphi\,.
\end{align}
\end{example}
\begin{example}
6th order extended GJMS-like operator: $l=3$. For $d\geq 4$ direct computations give:
    \begin{align}\label{gGJMS}
        \frac{8}{3\tilde{\Lambda}^2}\cP=(\nabla_{A}\nabla^{A})^3\hat\varphi+\frac{8}{\tilde{\Lambda}^2}\hat T^{(2)}_{AB}\nabla^{A}\nabla^B\hat\varphi\,.
    \end{align}
An interesting new feature shows up here: for $d=4$, $\hat T^{(2)}_{AB}$ entering $\cP$ is not a local function in the conformal geometry (gravitational leading) but is a gravitational subleading $\hat T^{(2)}_{AB}\equiv\cT_{AB}$.
In other words, in this case $\varphi$ is coupled to both the conformal geometry and the gravity subleading $\cT_{AB}$ which is itself a conformal field coupled to the conformal geometry background.  Unlike the previous examples, for $d=4$ equation \eqref{gGJMS} is conformally-invariant only if the background is Bach-flat. This example gives an elegant resolution of the puzzle of what is a natural generalisation of the 6th order GJMS operator in 4 dimensions that is well-defined on conformal gravity background.

For $d>4$, Theorem~\bref{bound-calculus} gives $\hat{T}^{(2)}_{AB}=\frac{2\tilde{\Lambda}^2}{D-5}\hat B_{AB}$ so that 
\begin{align}
\frac{8}{3\tilde{\Lambda}^2}\cP =(\nabla_{A}\nabla^{A})^3\hat\varphi+\frac{16}{D-5}\hat B_{AB}\nabla^{A}\nabla^B\hat\varphi
    \end{align}
does not depend on the gravity subleading. This describes the 6th order GJMS operator \cite{wunsch1986conformally,GJMS}.  Up to a different normalization of the Bach tensor, the above form coincides with that from~\cite{wunsch2000:some}.
\end{example}
\begin{prop}
    For  $N^{*}=2\ell$, i.e. $-w=\frac{d}{2}-\ell$, gPDE $\hat E_B^{\varphi}$ describes conformal scalar field $\hat \varphi$ with an additional unconstrained subleading field $\psi$. The leading $\hat\varphi$ is subject to the equation of the following form  
    \begin{align}
        \cP\sim (\nabla_{A}\nabla^{A})^l\hat\varphi+\dots
    \end{align}
    which does not involve the subleading $\psi$. Moreover, for $l\leq\frac{d}{2}$ or,  in other words, $w\leq 0$, this equation does not involve the subleading gravitational field $\cT_{AB}$ and coincides with order-$2\ell$ GJMS operator.
\end{prop}
Note that if one is interested in global solutions to the bulk equations  rather than in the near-boundary
analysis, it turns out that the subleading is not unconstrained but is actually uniquely determined as a functional of the leading if one requires solutions to behave well in the interior.  The respective functional is called the Dirichlet-to-Neumann map. By using it to express the subleading gravitational data in terms of the leading data and substituting the result into the extended GJMS operator, one may obtain a generally nonlocal operator that nevertheless depends only on the leading boundary data for gravity.

\begin{proof}
The proof that the operator has this form is completely analogous to the proof of Proposition~\bref{obstr-even} and uses equations~\eqref{formula-phiN}, \eqref{psi-notation}, \eqref{Laplacian-exp}. Tracking the leading order in $\nabla_{(A)}\hat\varphi$, we obtain
\begin{align}
    \cP=(\hat b^\varphi)^{*}\nabla_\Omega^{2l-2}\nabla_A\nabla^{A}\varphi=\nabla_A\nabla^{A}\hat\varphi^{(2l-2)}+\dots\propto(\nabla_A\nabla^{A})^2\hat\varphi^{(2l-4)}+\dots
\end{align}
and so on until it terminates at $\hat{\varphi}^{(0)}$.

The proof that this equation does not involve the subleading gravitational field for $l\leq \frac{d}{2}$ is slightly more subtle. To begin with, note that the subleading gravitational field can enter this equation only through the vector fields $\cD_{A}^{(N)}$, $N\geq D-3$. Since
\begin{align}
    \cP=\sum_{i=0}^{2l-2}C^{i}_{2l-2}\mathcal{D}_{A}^{(2l-2-i)}\sum_{j=0}^{i}C^{j}_i\mathcal{D}^{(i-j)|A}\hat\varphi^{(j)}
\end{align}
we see that the subleading $\cT_{AB}$ definitely does not enter this equation when $(2l-2)<D-3$ or equivalently $l<\frac{d}{2}$\footnote{Strictly speaking, we should also verify that the subleading term does not enter the equations for $\hat{\varphi}^{(N)}$, $N\leq 2l-2$, but since the equations for these functions have the same structure, only of lower order, this is obviously satisfied for them as well.}. It remains to verify that this still holds for $l=\frac{d}{2}$. The only terms that can produce $\cT_{AB}$ are
\begin{align}
    \cD_A^{(D-3)}\nabla^{A}\hat \varphi+\nabla_A \cD^{(D-3)|A}\hat\varphi.
\end{align}
Using \eqref{phi-D} it is clear that the gravity subleading can enter only through the following combination:
\begin{align}
    \cT_{AB}\Gamma^{B}\nabla^{A}\hat \varphi
\end{align}
but this vanishes since $\Gamma^{B}\nabla^{A}\hat \varphi\propto g^{AB}\hat \varphi$, as can be easily verified by a direct computation, and $\cT_{AB}$ is traceless. 
\end{proof}
Note that for an odd-dimensional boundary one can consistently set gravitational subleading to zero by imposing the additional condition $\nabla_{(C)}\cT_{AB}=0$, giving a conformal field $\varphi$ defined on a generic conformal geometry background for all $l \geq 1$. Indeed, for $d$-odd 
the gravitational leading is unconstrained and describes a generic conformal geometry.
At the same time, for $d$ even, we cannot set the gravitational subleading to vanish while the leading is subject to the obstruction equation. However, a more careful analysis shows that for $l \leq d/2$, these equations do not involve  $\cT_{AB}$ and are conformally-invariant on any conformal geometry background.

Finally, let us mention that we have omitted a number of special cases. Namely, the case not covered by Proposition \bref{prop:GJMS-trivial} is where both $N^*$ and $d$ are odd while the gravitational subleading is not set to zero. Without trying to be maximally general let us give an example of such systems:
\begin{prop}
    Let the dimension of the boundary be odd and $N^{*}=d+2$ (i.e., $w=1$). Then
    \begin{align}\label{Tphi}
        \cP=-\frac{D-2}{2}\cT^{AB}\nabla_{A}\nabla_B\hat\varphi
    \end{align}
\end{prop}
\begin{proof}
Similarly to Proposition~\bref{prop:GJMS-trivial}, one shows that $\hat\varphi^{(i)}=0$ for odd $i<d+2$. Also, analogously to Proposition~\bref{utv-nechet} we have $\cD_{A}^{(i)}=0$ for odd $i<d-2$. Hence,
\begin{align}
\begin{split}
        \cP&=\sum_{i=0}^{D-1}C^{i}_{D-1}\mathcal{D}_{A}^{(D-1-i)}\sum_{j=0}^{i}C^{j}_i\mathcal{D}^{(i-j)|A}\hat\varphi^{(j)}=\cD^{(D-1)}_A\nabla^{A}\hat\varphi=\\&=-\frac{1}{D-1}\sum_{i=0}^{D-3}d^{i}_{D-1}\hat T^{(i)|A}{}_{C}\cD^{(D-3-i)|C}\nabla_A\hat\varphi=-\frac{D-2}{2}\cT^{AC}\nabla_{A}\nabla_C\hat\varphi.
        \end{split}
\end{align}
\end{proof}
Constraint~\eqref{Tphi} defines a nonlinear conformally invariant equation involving the conformal fields $\cT_{AB}$ and $\hat \varphi$ defined on the generic conformal geometry background. A possible interpretation of this system is to ask this equation to hold for any configuration of $\cT_{AB}$, i.e. to require trace-free part of $\nabla_A\nabla_B \hat\varphi$ to vanish. This condition can be recognised as the almost Einstein equation on $\hat\varphi$. Note that $w=1$ is precisely the weight required for the conformal invariance of the almost Einstein equation. One can also speculate that this is the Killing condition for the depth-$2$ conformal spin-$2$ gauge field~\cite{Deser:1983mm}, see also ~\cite{Bekaert:2013zya,Barnich:2015tma}, if one identifies $\hat\varphi$ as the corresponding gauge parameter.


\subsection{Yang-Mills theory}
Now we turn to the example  where the matter field defined on the gravity background possesses gauge invariance and its equations of motion are generally nonlinear. The Yang-Mills theory can be defined on a general pseudo-Riemannian background (no constraints on background metric) and the fibre of $E^{\YM}$ is coordinatized by the components $A_b^I$ of the Yang-Mills field, ghosts $\cC^I$ and all their space-time derivatives. We often use the gauge Lie algebra valued $A_b, \cC$ as shorthand notations. The BRST differential is defined in a standard way:
\begin{equation}
QA_b=\xi^cD_cA_b+(D_b\xi^c)A_c+D_b\cC+\commut{A_b}{\cC}\,, \qquad
Q\cC=\xi^cD_c\cC-\half\commut{\cC}{\cC}
\end{equation}
and encode the usual transformation law of the Yang-Mills field and ghost under diffeomorphisms and Yang-Mills gauge transformations.

The Yang-Mills equations  
\begin{equation}
\label{YM-eom-usual}
\tilde Y^b=\frac{1}{\sqrt{|\tilde g|}}(D_a+\commut{A_a}{\cdot})(\sqrt{|\tilde g|}\,\, \tilde g^{ac}\tilde g^{bd} F_{cd})\,, \qquad F_{cd}\equiv D_c A_d- D_d A_c+\commut{A_c}{A_d}\,,
\end{equation}
can be rewritten in terms of $g_{ab}\equiv \Omega^{2} \tilde g_{ab}$ and then it turns out that $Y^b\equiv \Omega^{D-4}\tilde Y^b$ seen as a local function in $A,g,\Omega$ is regular at $\Omega=0$. The jet-prolongations of $Y^b$ together with  $\cI_{cl}$ generate the ideal $\cI^\YM_{cl}$  defining the gPDE that describes YM coupled to the conformal-like GR. Note that we took $A,\cC$ to be inert under the Weyl transformations and hence the uplift of YM system from a generic Riemannian background to conformal-like GR is straightforward. 

According to the general prescription described in Section~\bref{sec:matter-general} the next step is to pull back $E^\YM \to \Ecl$ to $\Ered  \subset \Ecl$. However, it turns out to be very convenient to first perform an equivalent reduction of $E^\YM \to \Ecl$ that only affects the fibre, i.e. the YM sector. This reduction is well known in the literature, see e.g.~\cite{Brandt:1996mh}, in the case of YM on the gravity background. But because the lift to conformal-like reformulation of gravity is trivial in our case ($A$, $\cC$ are inert under Weyl), the reduction is identical and amounts to the restriction to the surface singled out by:
\begin{equation}
\label{YM-redsurf}
D_{(a_1}D_{a_2}\ldots D_{a_{k-1}}A_{a_k)}=0\,, \qquad Q\left(D_{(a_1}D_{a_2}\ldots D_{a_{k-1}}A_{a_k)}\right)=0.
\end{equation}
As the remaining fibre coordinates one can take
$\{\nabla_{((a)}F_{b)c},\mathcal{C}\}$ pulled back to the surface \eqref{YM-redsurf}, where $\nabla_a\equiv \commut{Q}{\dl{\xi^a}}$.

The next step is to pullback the reduced bundle to $\Ered\subset \Ecl$ and then pullback the resulting gPDE (as a bundle over $T[1]X$) to the boundary $T[1]\Sigma$ giving a gauge PDE $(E^\YM,Q,\cI^\YM)$ over $(\Ee,Q,\cI)$, where $\cI^\YM$ is generated by the ideal $\cI$ of the gravity sector and $\nabla_a$-prolongations of 
\begin{align}\label{YM-eq-bulk}
    Y_b=(4-D)(\nabla^{a}\Omega) F_{ab}+\Omega \nabla^{a}F_{ab}\,,
\end{align}
where we keep using the same notations for the fibre coordinates and generators $Y_b$.  Closely related form of the YM equations was employed in \cite{Gover:2023rch} within a different framework and called conformally-compact YM equation. It is easy to see that for $\Omega=const$ \eqref{YM-eq-bulk} reduces to the ordinary Yang-Mills equation.

The action of $Q$ on the fibre coordinates of $E^{YM}$, namely $\{\nabla_{((a)}F_{b)c}, \cC\}$, is determined by
\begin{align}\label{YM-Q}
\begin{split}
    Q\cC&=-\frac{1}{2}[\mathcal{C},\cC] +\frac{1}{2}\xi^{a}\xi^{b}F_{ab},\\
    QF_{ab}&=\xi^{c}\nabla_{c}F_{ab}+C_{a}{}^{c}F_{cb}+C_{b}{}^{c}F_{ac}-[\cC,F_{ab}]\,.
    \end{split}
\end{align}
Recall that the action of $Q$ on any function of the form $\nabla_a\nabla_b\ldots F_{cd}$ (e.g. on coordinate functions) is determined by the above relations and $\commut{Q}{\nabla_a}=0$. If one disregards gravity sector by e.g. picking a fixed background, \eqref{YM-Q} define a minimal model of the gPDE formulation of YM theory. These relations are well-known in the context of local BRST cohomology~\cite{Stora:1984,Barnich:1995ap,Brandt:1996mh}.

 The ghost degree $1$ coordinates on $E^\YM$ are split into the ghosts of the conformal-like gravity sector $C^M=\{\xi^a,\lambda,C_a{}^b,\lambda_a\}$ and the YM ghosts $\cC^I$, $\cC=\cC^I t_I$, where $t_I$ are basis elements of the YM gauge algebra $\algg$. Introducing $R_I\equiv[\frac{\partial}{\partial \cC^{I}},Q]$, one finds: 
\begin{align}\label{YM-commut1}
    [R_I,R_J]=-U^{K}_{IJ}R_K, \quad \Big[R_{I},\Big[\frac{\partial}{\partial{C}^{M}},Q\Big]\Big]=0\,,
\end{align}
where $U^{K}_{IJ}$ are the structure constants of $\algg$, i.e. $\commut{t_I}{t_J}=U_{IJ}^K t_K$. In the gravity sector, all the formulas \eqref{commutrelations} for the commutators  remain intact save for
\begin{align}\label{YM-commut2}
    [\nabla_a,\nabla_b]=-\We^{d}{}_{cab}\Delta^{c}{}_{d}-\Co_{dab}\Gamma^{d}-F_{ab}^{I}R_{I}-(Q\We^{d}{}_{cab})\frac{\partial}{\partial C_{c}{}^{d}}-(Q\Co_{dab})\frac{\partial}{\partial\lambda_d}-(QF^{I}_{ab})\frac{\partial}{\partial \cC^{I}}.
\end{align}
This happens because the underlying Lie algebra is the direct sum of the algebra in the gravity sector and the YM gauge algebra.

The next step is to split the space-time index as $\{a\}=\{\Omega,A\}$. Then the YM Bianchi identities $\nabla_{[a}F_{bc]}=0$
encoded in $Q^2\cC=0$ give:
\begin{equation}
\label{YMBianchi}
\nabla_{\Omega}F_{BC}=\nabla_{B}F_{\Omega C}-\nabla_{C}F_{\Omega B}\,.
\end{equation}
These can be used to choose the following fibre coordinates: 
\begin{align}
\{\nabla_{(C)}J^{(N)}_{A},\nabla_{((A)}F_{B)C},\cC\}, \quad \text{where} \quad J^{(N)}_A=\nabla_\Omega^{N}F_{\Omega A}\,.
\end{align}
Then we can rewrite the $\nabla_a$-prolongations of the Yang-Mills equations $Y_b$ as $\nabla_A$-prolongations of $Y^{(N)}_b$ for $N \geq 0$. Analogously to the scalar field case, $Y^{(N)}_b$ can be computed modulo the ideals $\mathcal{K}^{(N)}$:
\begin{prop}\label{YM-inducedeq}
    The total ideal $\cI^\YM$ on $\Ee^\YM$ is generated by 
    the gravity ideal $\cI$ along with functions
    \begin{align}
    \label{YM-ideal-B}
    \begin{split}
    Y^{(0)}_B&=(4-D)J^{(0)}_B+\cK_B \, ,\\
        Y^{(N)}_{B}&=(4-D+N)J^{(N)}_B+N\tilde{\Lambda}\nabla_{\Omega}^{N-1}\nabla^{A}F_{AB}+\cK^{(N-1)},\quad N\geq1\,,\\
        Y^{(D-3)}_\Omega&=(D-3)\tilde\Lambda \nabla_\Omega^{D-4}\nabla^{A}J^{(0)}_A+\cK^{(D-4)}\,,
            \end{split}
    \end{align}
    and their $\nabla_A$-prolongations.
    \end{prop}
    \begin{proof}
$Y^{(N)}_b$ are easily computed by acting with $\nabla_\Omega^N$ on \eqref{YM-eq-bulk} and taking into account that $\nabla_\Omega^N \nabla_a \nabla_b \Omega \in \mathcal{K}^{(N)}$; $\nabla_\Omega \Omega-\tilde\Lambda, \nabla_A\Omega\in\cK_B$. The only nontrivial point is to show that  $Y^{(N)}_\Omega$, for $N\neq D-3$ can be expressed as a linear combination of $\nabla^{A}Y^{(i)}_{A}$ with $i\leq N$, and therefore they can be omitted from the set of generators. To see this, notice that
\begin{align}\label{YM-divergence}
    \nabla^{b}Y_b=(3-D)\nabla^{a}\Omega\nabla^{b}F_{ab}+\Omega\nabla^{b}\nabla^a F_{ab}=(3-D)\nabla^{a}\Omega\nabla^{b}F_{ab}\,,
\end{align}
where we made use of $\nabla^{a}\nabla^{b}F_{ab}=\frac{1}{2}[\nabla^{a},\nabla^{b}]F_{ab}=0$. Applying $\nabla_\Omega^{N-1}$ to \eqref{YM-divergence} and using $\nabla^a Y_a = \tilde\Lambda^{-1} Y^{(1)} + \nabla^A Y_A+\cK^{(0)}$ together with $Y^{(N)}_{\Omega}=N\tilde\Lambda\nabla_\Omega^{N-1}\nabla^AJ_A^{(0)}+\cK^{(N-1)}$ we obtain
\begin{align}\label{YM-Noether}
    \frac{(3-D+N)}{N\tilde \Lambda}Y_\Omega^{(N)}=\nabla_\Omega^{N-1}\nabla^A Y_A+\cK^{(N-1)}\in \bigcup_{i=0}^{N-1}I(Y_B^{(i)},\nabla_A)\cup\cK^{(N)}.
\end{align}
The inclusion follows from the decomposition of $\nabla_\Omega^{N-1} \nabla^A Y_A$ as a sum $\mathcal{D}^{(N-1-i)|A} Y^{(i)}_A$ and from the fact that $\mathcal{D}_A^{(i)} = [Q, \nu_A^{(i)}]$, where $\nu_A^{(i)}$ has degree $-1$ and preserves $\mathcal{K}^{(i+1)}$ by Theorem \bref{prop:mnogoidealov}.
\end{proof}
Note that equation \eqref{YM-Noether} follows from equation \eqref{YM-divergence} which in turn is the Weyl-covariant extension of the usual Noether identities for the Yang-Mills equation. The analogous phenomenon occurs in the gravity sector, see Appendix~\bref{AppA}.

Similarly to the analysis of gravity and scalar field, almost all the equations in~\eqref{YM-ideal-B} can be solved with respect to a part of the coordinates. More specifically, we introduce the Yang-Mills extension $\hat \Ee^{\YM}$ of the off-shell boundary gPDE $\hat \Ee$ from Section~\bref{sec:boundary-gPDE}. 
\begin{definition}
Bundle $\hat \Ee^\YM \to \hat \Ee$ is the submanifold of $\Ee^{\YM}$ singled out by the ideal generated by $\hat \cK = \ker \hat b^*$ 
        together with 
        \begin{align}\label{YM-J}
            J^{(N)}_B=\frac{N}{(D-4-N)}\tilde{\Lambda}\nabla_{\Omega}^{N-1}\nabla^{A}F_{AB},\quad N\neq D-4.
\end{align}
The embedding $\hat \Ee^\YM \hookrightarrow \Ee^\YM$ is denoted by  $\hat b_\YM$ and for any function $f$ on $\Ee^\YM$ we define $\hat f\equiv \hat b_\YM^{*}f$.
\end{definition}
Because $\hat \Ee^\YM$ can be obtained as subbundle of $\Ee^\YM$ pulled back to $\hat \Ee \subset \Ee$ it is naturally a bundle over $\hat \Ee$ with projection $\hat \pi^\YM$.
Note that, when restricted to functions that are constant along the fibres of $\pi^\YM$, $\hat b_{\YM}^*$ can be identified with the $\hat b^*: \cC^\infty(\Ee) \to \cC^\infty(\hat\Ee)$ in the gravity sector.

Just like in the case of gravity and scalar field, $\hat{\Ee}^{\YM}$ is not a $Q$-submanifold of $\Ee^{\YM}$, but we can still define a natural $Q$-structure on $\hat{\Ee}^{\YM}$. To construct it, we consider the submanifold singled out by $\cK^{(D-3)}$ together with the equations \eqref{YM-J} for  $0\leq N\leq D-5$. This is a $Q$-submanifold which is also a bundle over $\Ee^{(D-3)}$ in a natural way. We denote it by $\Ee^{(D-5)|\YM}\to \Ee^{(D-3)}$. Then we observe that the subalgebra of functions on $E^{YM|(D-5)}$ that is generated by the following coordinate functions on $E^{\YM}$:
 \begin{align}
\label{funct-E-D-3-YM}
    \xi^{B}, C_{B}{}^{C}, \lambda, \lambda_{B},\quad    g_{BC}, \nabla_{((A)} \We^{B}{}_{CD)E}, \nabla_{(A)}T_{BC}^{(D-3)}\,,\quad \nabla_{((A)}F_{B)C},\, \mathcal{C},\, \nabla_{(A)}J_B^{(D-4)},\quad |A|\geq0
\end{align}
 pulled back to $\Ee^{\YM|(D-5)}$ is closed under $Q$. In the gravity sector, this observation is entirely analogous to the one in Section~\bref{sec:boundary-gPDE}.  For $\nabla_{(A)}J_B^{(D-4)}$ it
 it is a direct generalisation of Proposition~\bref{App-QT} while for the remaining generators the statement is immediate. It follows $\Ee^{\YM|(D-5)}$ can be considered a $Q$-bundle $\Ee^{\YM|(D-5)}\to \hat E^{\YM\,\prime}$. Indeed, one simply takes functions in \eqref{funct-E-D-3} as functions pulled back from $E^{\YM\prime}$. Moreover, $\hat E^{\YM}$ can be understood as a section of this bundle and the desired  $Q$-structure on $\hat E^{\YM}$ is simply obtained as a pushforward of that on $\hat E^{\YM\prime}$.

Again following the logic of Section~\bref{sec:boundary-gPDE} one can describe functions on $\hat \Ee^\YM$ as a subquotient of functions on $\Ee^\YM$ as follows: considering the subalgebra of $\hat \Ee^\YM$-projectable functions in $\cC^\infty(\Ee^\YM)$, consisting of functions whose restriction to $E^{\YM|(D-5)}$ is constant along the fibres of $E^{\YM|(D-5)}\to E^{\YM\prime}$. The quotient of $\hat \Ee^\YM$-projectable functions by those vanishing on $E^{\YM|(D-5)}$
is naturally identified with $\cC^\infty(\hat\Ee^\YM)$. Moreover, for 
$\hat \Ee^\YM$-projectable functions we have  $Q\hat b_\YM^*f=\hat b^{*}_\YM Qf$. In particular, $\hat \Ee^\YM$-projectable functions  include all $\hat E$-projectable functions on $\Ee$, e.g.  $T^{(N)}_{AB}\, ,\,N\leq D-3$, pulled back to $\Ee^\YM$. It is straightforward to see that $J^{(N)}_A$ are also $\hat E^\YM$-projectable if $N\leq D-4$.

As a coordinate system on the fibres of $\hat E^\YM_{B}$ we take
\begin{equation}
\hat \cC=\hat b^*_\YM \cC\,, 
\qquad 
\nabla_{((C)} \hat F_{A)B}= \nabla_{((C)} \hat b^*_\YM F_{A)B}\,,
\qquad
\nabla_{(C)}\cJ_B\equiv \nabla_{(C)} \hat b^*_\YM J^{(D-4)}_B\,,
\end{equation}
The ideal $\hat \cI^{\YM} \equiv \hat b_\YM^{*}\cI^\YM$ is generated by the pullback of the ideal $\hat \cI$ from $\hat E$ and the $\nabla_A$-prolongations of functions
\begin{equation}
\mathcal{Y}_B\equiv \hat b^{*}_\YM\nabla_\Omega^{D-5}\nabla^{A}F_{AB}\,,
\qquad 
\mathcal{Y}\equiv \hat b^{*}_\YM\nabla_{\Omega}^{D-4}\nabla^{A}J_A. 
\end{equation}
Note that, unlike the scalar field case, ideal $\hat \cI^{\YM}$ is always strictly $Q$-invariant (recall that $\hat \cI^\varphi$ is generally $Q$-invariant modulo $\hat\cI$ only).  In particular, the on-shell boundary YM system is well defined on an arbitrary, not necessarily obstruction-flat, background. This is a direct consequence of the fact that  $\cK^{(N)}$ appearing in formulas \eqref{YM-ideal-B} for $N \leq D-4$ vanishes upon pullback to $\hat E$. 

The action of $Q$ on some of the coordinates is immediately obtained by restricting \eqref{YM-Q} to $\hat E^\YM_B$:
    \begin{align}
        \begin{split}
    Q\hat \cC&=-\frac{1}{2}[\hat \cC ,\hat \cC]+\frac{1}{2}\hat \xi^{A}\hat \xi^{B}\hat F_{AB},\\
    Q\hat F_{AB}&=\hat \xi^{C}\nabla_{C}\hat F_{AB}+\hat C_{A}{}^{C}\hat F_{CB}+\hat C_{B}{}^{C}\hat F_{AC}-[\hat \cC,\hat F_{AB}]\,.
        \end{split}
    \end{align}
Notice that the action of $Q$ on $\hat\cC$ and  $\nabla_{((C)} \hat F_{A)B}$ does not affect the sector $\cJ_B$, and therefore $\hat E^\YM_B$ itself is a $Q$-bundle with the fibre coordinates $\nabla_{(A)}\cJ_B$. This is identical  to the situation we had in the case of pure gravity and scalar field. Now $\cJ_A$ is interpreted as a YM subleading.

Just like in the case of gravity and scalar field, gauge PDE $(\hat \Ee^\YM,Q)$ considered without restricting to the zero locus of $\hat \cI^\YM$ is a good starting point to discuss the explicit space-time realisation of equations of motion because if we disregard the gravity part, i.e. restrict to a particular solution of $(\hat \Ee,Q)$, the system does not encode any nontrivial differential equations. More precisely,
\begin{prop}\label{prop:off-shell}
Let $\sigma_G:T[1]\Sigma \to \hat \Ee$ be a solution to $(\hat \Ee,Q)$ and  $A_B(x)$ and $\mathbf J_B(x)$ are coefficients of generic 1-form on $\Sigma$ with values in the gauge algebra. Then there exists a unique  solution $\sigma:T[1]\Sigma \to\hat \Ee^\YM$ such that
$\hat \pi^\YM~\circ~\sigma =\sigma_G$ and 
\begin{equation}
\sigma^*(\cC)=A_B(x)\theta^{B}\,, \qquad \sigma^*(\cJ_B)=\mathbf{J}_B(x)\,.
\end{equation}
\end{prop}
In other words, the restriction of $(\hat \Ee^\YM,Q)$ to $\sigma_G$ is an off-shell system.\footnote{The definition of an  off-shell gauge PDE can be made precise as follows: it is a gPDE that is equivalent to a jet-bundle of a non-negatively graded fibre bundle, with $Q$-structure being $Q=\dh+\gamma$, where $\gamma$ is an evolutionary homological vector field and $\dh$ is the horizontal differential seen as a vector field.} We see that solutions are 1:1 with unconstrained configurations of fields $\sigma^*(\cC)$ and $\sigma^*(\cJ_B)$. In particular, the equations of motion in terms of $\sigma^*(\cC)$ and $\sigma^*(\cJ_B)$ can be explicitly written as
\begin{equation}
\sigma^*(\mathcal{Y}_B)=0\,, \qquad 
\sigma^*(\mathcal{Y})=0\,.
\end{equation}
More detailed description of the solution space $\hat \Ee^{\YM}$ is given in Appendix~\bref{App-YM-sol}. Let us emphasize that in the above  Proposition $\sigma_G$ is a solution of the off-shell boundary GR. In other words, we do not impose $\sigma_G^*\hat \cI = 0$.

To proceed further, we need the action of $Q$ on  functions $\hat J_A^{(N)}$ and coordinate functions $\cJ_B$ in particular. The following is a straightforward generalisation of Proposition \bref{QTAB}:
\begin{prop}\label{QJ}
    On $\hat \Ee^\YM$ one has:
    \begin{equation}
            Q\hat J^{(N)}_{A}=\hat \xi^{B}\nabla_{B}\hat J^{(N)}_{A}+\hat C_{A}{}^{B}\hat J^{(N)}_{B}+\hat \lambda_{B}\Gamma^{B}\hat J^{(N)}_{A}-[\hat \cC,\hat J^{(N)}_A]-(N+1)\hat \lambda \hat J^{(N)}_{A},\quad N\leq D-4\,,
    \end{equation}
            where
    \begin{equation}
    \label{YM-Gamma}
            \Gamma_{C}\hat J_{B}^{(N)}=\tilde{\Lambda}\sum_{i=0}^{N-2}d^{i}_{N}\mathcal{D}_{C}^{(N-2-i)}\hat J_{B}^{(i)}+N\tilde{\Lambda} \hat F^{(N-1)}_{CB}\,.
    \end{equation}
\end{prop}
We are now ready to formulate the Yang-Mills theory extension of the boundary calculus, introduced in Theorem~\bref{bound-calculus}:
\begin{theorem}\label{YMbound-calculus}
{\rm (Yang-Mills boundary calculus)}
The system $\hat J^{(N)}_{A}\,, \hat F^{(N)}_{AB}  
\,, \cD^{(N)}_A$ on $\hat \Ee^\YM$ satisfies the following relations
\begin{align}\label{YMformuli-calculus}
    \begin{split}
        \hat F^{(N)}_{AB}&=\hat b_\YM^{*}\nabla_{\Omega}^{N-1}( \nabla_{A}J_{B}-\nabla_{B}J_{A})\,, \quad 1\leq N\leq D-5\,, \\
        \hat J^{(N)}_{B}&=\frac{N\tilde{\Lambda}}{D-4-N}\hat b_\YM^{*}(\nabla_\Omega^{N-1}\nabla^{A}F_{AB}), \quad 0\leq N\leq D-5, 
                \end{split}
\end{align}
and for any degree zero function $f$ and $1\leq N\leq D-5$ one has:
\begin{align}\label{YMformuli-D}
\begin{split}
        &\mathcal{D}_{A}^{(N)}f=\Big(-\hat J^{(N-1),D}{}_{CA}\Delta^{C}{}_{D}+\dfrac{1}{\tilde\Lambda N}\hat T^{(N)}_{CA}\Gamma^{C}-J^{(N-1)|I}{}_{A}R_{I}-\frac{1}{N}\sum_{i=1}^{N-2} d^{i}_N \hat T^{(i),C}{}_{A}\mathcal{D}_{C}^{(N-2-i)}\Big)f\,.\\
\end{split}
\end{align}
Moreover, 
\begin{equation}\label{YMJ-lower}
\hat F^{(0)}_{AB}=\hat{F}_{AB}\,, \qquad 
\hat J^{(0)}_{A}=0\,, \qquad 
\mathcal{D}^{(0)}_{A}\equiv\nabla_A\,,
\end{equation}
and all the higher order $F^{(\cdot)}_{AB},J^{(\cdot)}_{BC},\cD^{(\cdot)}_A$ are recursively determined by the above relations and the gravity boundary calculus Theorem~\bref{bound-calculus}, where each expression of the form $\hat b^{*}\nabla_\Omega^{N}\nabla_A f$  is evaluated as 
\begin{align}
\label{YMOmega-A-commut}
    \hat b^{*}_\YM\nabla_\Omega^{N}\nabla_A f=\sum_{i=0}^{N}C^{i}_{N}\mathcal{D}_{A}^{(N-i)}\hat f^{(i)}\,,
\end{align}
and the action of $\Gamma_A$ on $\hat J^{(N)}_A$ is given by Proposition~\bref{QJ}.
\end{theorem}
\begin{proof}
The first line in \eqref{YMformuli-calculus} can be proven simply by using the Bianchi identity \eqref{YMBianchi} in the form $\nabla_{\Omega}F_{AB}=\nabla_{A}J_{B}-\nabla_{B}J_A$ and acting on it with $\hat{b}^{*}_\YM\nabla_\Omega^{N-1}$. The second line follows from Proposition $\bref{YM-inducedeq}$ and the definition of $\hat{\Ee}^\YM$. To prove \eqref{YMformuli-D}, note that when acting on functions of ghost degree $0$, as follows from \eqref{YM-commut2}, the formula for $[\nabla_{\Omega},\nabla_{A}]$ differs only by a term of the form $J^{I}_B R_I$. Using the fact that $[\nabla_{A},R_{I}]=0$, it is straightforward to see that the formula for $\cD^{(N)}_{A}$ in the gravity sector in Theorem \bref{bound-calculus} differs only by a term of the form $\hat J^{(N-1)|I}_A R_I$.\end{proof}
It is easy to see that, as in the case of pure gravity, the generators of the ideal $\hat \cI^{\YM}$ can be expressed in terms of the boundary calculus:
\begin{align}\label{YM-obstr}
    \begin{split}
        \mathcal{Y}_B&= \nabla^{A}\hat F^{(D-5)}_{AB}+\sum_{i=0}^{D-6}C^{i}_{D-5}\cD^{(D-5-i)|A}\hat{F}_{AB}^{(i)},\\
        \mathcal{Y}&= \nabla^{A}\cJ_A+\sum_{i=0}^{D-5}C_{D-4}^{i}\cD^{(D-4-i)|A}\hat J^{(i)}_{A}.
    \end{split}
\end{align}
The functions $\cY_A$ do not involve the subleading sector and correspond to a certain Weyl-invariant equation imposed on the leading sector while $\cY$ is the modified conservation law for subleading $\cJ_A$. One may view $\cY_A$ and $\cY$ as direct analogues of the obstruction tensor $\cO_{AB}$ and the modified conservation law $\cO_A$ in the sector of gravity. In particular, an analogue of Proposition~\bref{prop-div-obstr} is the following:
\begin{prop}
    On $\hat \Ee^\YM$ we have:
    \begin{align}
        \nabla^{A}\cY_A=0,\qquad \Gamma^A\cY=-\tilde\Lambda(D-4) \cY^A\,,
    \end{align}
    where $\Gamma^A\equiv [\frac{\partial}{\partial \hat \lambda_{A}},Q]$.
\end{prop}
   \begin{proof}
       Again, the first equation is a distant consequence of the Noether identities for the Yang-Mills equations, which can be obtained by applying $\hat b_{\YM}^{*}$ to \eqref{YM-Noether} with $N=D-3$. The second equation states that the Weyl invariance of the conservation equation for the Yang-Mills subleading  $\cJ_A$ is equivalent to imposing the obstruction equation $\cY_B=0$ 
       as can be proved by a straightforward calculation using Lemma~\bref{Lemma-Gamma}. 
   \end{proof} 
Finally, the explicit form of the conformal YM equation and the  conformal conservation equation can be found in terms of the above boundary calculus by using Proposition~\bref{prop:off-shell} and repeating the analysis of Section~\bref{sec:gravity-sections} in the YM sector. In the rest of this section we discuss the structure of the boundary YM system in more detail. 

A direct analogue of Proposition \bref{utv-nechet} is the following:
\begin{prop}\label{YMutv-nechet}
For any odd integer $i$ such that  $1\leq i\leq D-4$ we have: 
\begin{equation}
\hat J^{(i-1)}_{A}=0\,, \qquad \cD^{(i)}_{A}=0\,, \qquad \hat F^{(i)}_{AB}=0\,.
\end{equation}
\end{prop}
The case of an odd-dimensional boundary is again relatively simple (see Proposition \bref{J-nechet}, the proof is analogous).
\begin{prop}
    For $d$ odd we have:
    \begin{align}
        \mathcal{Y}_B=0,\quad \mathcal{Y}= \nabla^{A}\cJ_A,\quad \Gamma_{B}\cJ_A=0.
    \end{align}
\end{prop}
\noindent It follows that if we fix a background solution $\sigma_G$ in the boundary GR sector, the system describes unconstrained field $A_B(x)$ and field $\mathbf{J}_B(x)$ (current) subject to the conservation condition that follows from $\sigma^{*}\cY=0$:
\begin{align}
    \nabla_A^{g}\mathbf{J}^A(x)=0\,,
\end{align}
 where $\nabla_A^g$  is the Levi-Civita covariant derivative, twisted by the Yang-Mills field $A_B(x)$. Weyl transformations for $\mathbf{J}_A$ can be read-off from $Q\cJ$. Assuming that $\sigma_G$ is taken in the metric gauge one finds:
  \begin{align}
      \delta_{\bar\lambda} \mathbf{J}_A=-(D-3)\bar \lambda(x)\mathbf{J}_A\,,
  \end{align} 
where $\bar\lambda(x)$ is the Weyl transformation parameter. In other words 
$\mathbf{J}_A$ has Weyl weight $-(D-3)$.

As before, we begin the analysis of the even-dimensional boundary case by first considering a few examples in lower dimensions:
\begin{example}\label{YM-example-5D}
   In the case of $d=4$ one has 
   \begin{equation}
    \mathcal{Y}_B=\nabla^{A}\hat F_{AB},\quad \mathcal{Y}=\nabla^{A}\cJ_A, \quad \Gamma_B\cJ_{A}=\tilde{\Lambda}\hat F_{BA}.
   \end{equation}
   Choosing a background solution $\sigma_G$ to $\hat E\to T[1]\Sigma$ in the metric-like gauge (see Section~\bref{sec:gravity-sections}), the solution space for $\hat E^\YM|_{\sigma_G}$ is described by two $\algg$-valued fields, $A_B(x)$ and $\mathbf{J}_B(x)$, subject to equations
\begin{align}
    \nabla^g_AF^{A}{}_{B}=0,\qquad \nabla_A^{g}\mathbf{J}^{A}=0\,,
\end{align}
   where $F_{AB}$ is the usual Yang-Mills curvature for $A_B(x)$.
   See Appendix~\bref{App-YM-sol} for details.
       Note that the leading part of the  boundary system is the usual 4d YM. This is compatible with the fact that the Yang-Mills equation is conformally invariant in dimension 4.

       Under Weyl transformations, $\mathbf{J}_A$ transforms as
       \begin{align}
           \delta_{\bar\lambda}\mathbf{J}_A=\tilde\Lambda (\partial^B\bar \lambda) F_{BA}-2\bar \lambda\mathbf{J}_A 
       \end{align}
   \end{example}
   \begin{example}\label{YM-example-7D}
       In the case of $d=6$ one has
       \begin{align}\label{6deqs}
       \begin{split}
           \mathcal{Y}_B&=\frac{\tilde{\Lambda}}{2}\Big((\nabla_{C}\nabla^{C})\nabla^{A}\hat F_{AB}+4\hat \Co^{DA}{}_{B}\hat F_{DA}-2[\hat{F}_{AB},\nabla_{C}\hat F^{CA}]\Big),\\
           \mathcal{Y}&=\nabla^{A}\cJ_A,\quad \Gamma_{C}\cJ_{A}=3\tilde{\Lambda}^2\Big(\nabla_{C}\nabla^{B}\hat F_{BA}-\frac{1}{2}\nabla_{A}\nabla^{B}\hat F_{BC}\Big).
                  \end{split}
       \end{align}
       Analogously to the previous example, the equations satisfied by $A_B(x)$ and $\mathbf{J}_A(x)$ read as (see Appendix~\bref{App-YM-sol} for details):
\begin{multline}\label{YM-d6-eq}
            \Box^g \nabla^{g}_AF^{A}{}_{B}+2\partial_CP F^{C}{}_{B}+4P^{AC}\nabla^g_AF_{CB}-3P\nabla^{g}_AF^{A}{}_{B}+4C^{DA}{}_{B}F_{DA}-2[F_{AB},\nabla_C^gF^{CA}]=0
\end{multline}
       \begin{align}\label{YM-d6-cons}
\nabla_{A}^g\mathbf{J}^{A}+\frac{3}{2}\tilde\Lambda^{2}P^{AC}(\nabla^g_C\nabla^{g|B}F_{BA})=0\,,
       \end{align}
       where $P_{AB}(x)$ is a Schouten tensor and $C_{ABC}(x)$ is a Cotton tensor for the background conformal geometry described by metric $g_{AB}(x)$.
       Following \cite{Gover:2023rch}, where \eqref{YM-d6-eq} was originally found, we call it the higher conformal Yang-Mills equation in dimension $6$.
       The subleading field $\mathbf{J}_A$, in addition to satisfying the modified conservation equation \eqref{YM-d6-cons}, also transforms nontrivially under Weyl rescalings:
       \begin{align}
           \delta_{\bar \lambda}\mathbf{J}_A=3\tilde\Lambda^2\partial^C\bar\lambda\Big(\nabla_C^g\nabla^g_{B}F^{B}{}_A+2P^{D}{}_{C}F_{DA}-\frac{1}{2}\nabla_A^g\nabla^g_{B}F^{B}{}_C-P^{D}{}_{A}F_{DC}\Big)-4\bar\lambda\mathbf{J}_A
       \end{align}
       Note that the inhomogeneous terms originate from $\Gamma_C \cJ_A$ term in $Q\cJ_A$.
   \end{example}
Let us give some details on how~\eqref{6deqs} is derived. Using Proposition~\bref{YMutv-nechet} together with Theorem~\bref{YMbound-calculus} we obtain:
       \begin{align}
       \hat{J}^{(1)}_{B}=\frac{\tilde\Lambda}{2}\nabla^{A}\hat F_{AB},\quad \hat{F}^{(2)}_{AB}=\nabla_{A}\hat J^{(1)}_B-\nabla_{B}\hat J^{(1)}_{A}.
       \end{align}
       Then from \eqref{YM-obstr} and \eqref{YMformuli-D}:
       \begin{align}
                   \mathcal{Y}_B= \nabla^{A}\hat F^{(2)}_{AB}+\cD^{(2)|A}\hat{F}_{AB}=\nabla^{A}(\nabla_{A}\hat J^{(1)}_B-\nabla_{B}\hat J^{(1)}_{A})+\tilde{\Lambda}\hat \Co^{DA}{}_{B}\hat F_{DA}-[\hat F_{AB},\hat J^{(1)|A}]\,,
       \end{align}
   where we have also made use of the fact that $\hat J^{(1)}_{ABC}=-\tilde{\Lambda}\Co_{CAB}$ thanks to~\eqref{J-lower}. Using 
   \begin{align}
       \nabla^{A}\nabla_B\hat{J}^{(1)}_A=[\nabla^{A},\nabla_B]\hat{J}^{(1)}_{A}=-\hat \Co_{D}{}^{A}{}_{B}\Gamma^{D}\hat{J}^{(1)}_{A}-\hat F^{J|A}{}_{B}R_{J}\hat{J}^{(1)}_{A}=-\tilde{\Lambda}\hat \Co^{DAB}\hat F_{DA}+[\hat F^{A}{}_{B},\hat{J}^{(1)}_{A}]
   \end{align}
   and substituting the previously computed $\hat J^{(1)}_{A}$ one finds the explicit expression for $\mathcal{Y}_B$ from~\eqref{6deqs}. The formula for $\mathcal{Y}$ is obtained from \eqref{YM-obstr} using $\cD^{(2)|A}\hat{J}^{(1)}_{A}=0$.   
   Finally, using the above
   in~\eqref{YM-Gamma} one finds the explicit expression for $\Gamma_{C}\cJ_{A}$.
   
\begin{example}\label{YM-example-9D}
    In the case $d=8$ we have
    \begin{align}\label{YM-example-9D-formula}
\begin{split}
            \mathcal{Y}_B=&(\nabla_{A}\nabla^{A})\hat J^{(3)}_B+\frac{3\tilde{\Lambda^2}}{2}(\hat B^{AC}\nabla_{C}\hat F_{AB}-\hat F_{CA}\nabla^{A}\hat B_{B}{}^{C}+\frac{1}{2}\hat B^{A}{}_{B}\nabla^{C}\hat F_{CA})+\\&+6\tilde\Lambda \hat C^{AC}{}_{B}(2\nabla_{A}\hat J_C^{(1)}-\nabla_{C}\hat J^{(1)}_{A})+2[\hat J^{(3)|A},\hat F_{AB}]+6[\hat J^{(1)|A},2\nabla_A\hat J_B^{(1)}-\nabla_B\hat J_A^{(1)}]\,,\\
            \mathcal{Y}=&\nabla^A\cJ_A+\frac{15}{4}\tilde{\Lambda}^{2}\hat B^{AC}\nabla_{C}\hat J^{(1)}_A-5[\hat J^{(3)|A},\hat J^{(1)}_A].
\end{split}
\end{align}
where $\hat J^{(1)}_A$ and $\hat J^{(3)}_A$ are computed using Theorem~\bref{YMbound-calculus} and are given by:
\begin{align}\label{YM-example-9D-J1}
    \begin{split}
        \hat J^{(1)}_B&=\frac{\tilde\Lambda}{4}\nabla^{A}\hat F_{AB},\\
        \hat J^{(3)}_B&=\frac{3\tilde\Lambda^{2}}{8}\Big( (\nabla_{A}\nabla^{A})\nabla^{C}\hat F_{CB}+8\hat \Co^{DA}{}_{B}\hat F_{DA}-2[\hat F_{AB},\nabla^{C}\hat F_C{}^{A}]\Big).
    \end{split}
\end{align}
 The details of the derivation of the explicit form of  $\mathcal{Y}_B$ and $\mathcal{Y}$ can be found in Appendix~\bref{App-YM9}.

 Repeating the analysis of the previous examples one finds that $\sigma^{*}\mathcal{Y}_B=0$ is a higher conformal Yang-Mills equation in dimension $8$, which for a fixed metric-like background $\sigma_G$ can be written as:
 \begin{multline}\label{8d-YMeq}
     \Box^{g}j^{(3)}_B+2 (\partial^DP+2P^{AD}\nabla_A^{g})(2\nabla_D^g j^{(1)}_B-\nabla_B^{g}j ^{(1)}_D+8P_{D}{}^{C}F_{CB}-4P_{B}{}^{C}F_{CD})+ \\+2(-5P_{AD}P^{AD}j^{(1)}_B+8 P_{AE}P^{AD}\nabla^g_{D}F^{E}{}_{B}+2P_{AD}P^{A}{}_{B}j^{(1)|D})-5Pj^{(3)}_B+\\+4( B^{AC}\nabla^g_{C} F_{AB}- F^{CA}(\nabla^{g}_A B_{BC}-4P^{E}{}_{A}C_{BEC}-4P^{E}{}_{A}C_{CEB})+\frac{1}{2} B^{A}{}_{B}\nabla^{g}_{C} F^{C}{}_{A})+\\+4 C^{AC}{}_{B}(2\nabla^g_{A}j_C^{(1)}-\nabla^g_{C}j^{(1)}_{A}+8P_{A}{}^{D}F_{DC}-4P_{C}{}^{D}F_{DA})+\\+2[j^{(3)|A}, F_{AB}] +[ j^{(1)|A},2\nabla^g_A j_B^{(1)}-\nabla^g_Bj_A^{(1)}+8P_{A}{}^{C}F_{CB}-4P_{B}{}^{C}F_{CA}]=0\,,
 \end{multline}
 where $B_{AB}(x)$ is a Bach tensor and
 \begin{align}
    \begin{split}
        &j^{(1)}_B=\nabla^g_AF^{A}{}_{B},\\
        &j^{(3)}_B=\Box^g\nabla^{g}_AF^{A}{}_{B}+4\partial_CP F^{C}{}_{B}+8P^{AC}\nabla^g_AF_{CB}-3P\nabla^{g}_AF^{A}{}_{B}+8C^{DA}{}_{B}F_{DA}-2[F_{AB},\nabla_C^gF^{CA}]\,.
    \end{split}
\end{align}
To the best of our knowledge, this operator has not been found before.

From $\sigma^{*}\cY=0$ we have the following conservation equation:
\begin{multline}\label{8d-conservation}
                \sigma^{*}\mathcal{Y}=\nabla^g_A\mathbf{J}^A+\frac{5\tilde\Lambda^{3}}{4}P^{CB}(C_{C}{}^{D}{}_{B}j^{(1)}_D+B_{CD}F^{D}{}_{B}+3\nabla_{C}^{g}j_{B}^{(3)}+6P_{C}{}^{E}\nabla_{E}^gj^{(1)}_{B})+\\+\frac{15\tilde{\Lambda}^{3}}{16} B^{AC}(\nabla^g_{C} j^{(1)}_A)-\frac{15\tilde\Lambda^3}{32}[j^{(3)|A},j^{(1)}_A].
\end{multline}
For the derivation of equations \eqref{8d-YMeq} and \eqref{8d-conservation}, see Appendix \bref{App-YM-sol}.
\end{example}
The general situation is described by:
\begin{prop}\label{YM-obstr-even} For $d$ even, the boundary system describes higher conformal YM  theory together with an additional field $\cJ_{A}$. The symbols of the equations of motion for these fields have the following structure:
    \begin{align}
    \begin{split}
                \cY_B &\equiv (\nabla_{A}\nabla^A)^{\frac{d-4}{2}} \nabla^C\hat F_{CB}+\dots,\\
 \cY&\equiv \nabla^{A}\cJ_A+\sum_{i=0}^{d-4}C_{d-3}^{i}\cD^{(d-3-i)|A}\hat J^{(i)}_{A}\,.
     \end{split}
\end{align}
Moreover, $\cY_B$ does not involve $\cT_{AB}$, $\cJ_A$ and their $\nabla_A$-derivatives. $\sigma^*(\cY_B)=0$ is the (higher-dimensional) conformal  YM equation on $A_B(x)\theta^B\equiv \sigma^*(\cC)$. The second equation gives  a generalised conservation condition satisfied by the subleading $\mathbf{J}_{A}(x)\equiv \sigma^*(\cJ_A)$. 
  \end{prop}

The proof is a straightforward generalisation of that of Proposition~\bref{obstr-even}.
\begin{remark} Suppose that we are given  an Einstein metric $g_{AB}(x)$ and a  YM field $A_B(x)$ on $\Sigma$, $\dim \Sigma=4,6,8$ such that YM equations are satisfied: $\nabla^g_A F^{AB} = 0$. Then $A,g$ also satisfy the higher conformal Yang-Mills equation, as can be seen from formulas \eqref{YM-d6-eq},\eqref{8d-YMeq}.  This statement is a YM analogue of the known fact that higher conformal gravities contain an Einstein sector \cite{graham2005ambient}, see also~\cite{Boulanger:2025oli}. These should extend to higher even dimensions.
Notice also that conservation equations \eqref{YM-d6-cons} and \eqref{8d-conservation} reduce simply to $\nabla_A^g \mathbf{J}^A = 0$ over an Einstein-YM background.  
\end{remark}

\section{Conclusions}

In this work we gave the  explicit description of the boundary structure of AAdS gravity and (gauge) fields defined on its background. More specifically, we limited ourselves to the gravity itself, scalar field,  and YM theory. Although these do not exhaust the list of known AdS gauge fields, the extension to totally antisymmetric fields should be rather straightforward while general higher spin gauge fields, are well-defined  on flat AdS space only (among gravity backgrounds)  and in this case the boundary structure is known in the literature~\cite{Bekaert:2012vt,Bekaert:2013zya,Chekmenev:2015kzf}.

In studying boundary structure we employ the gauge PDE approach which we extend to the case of asymptotic boundaries by introducing a notion of $Q$-boundary and giving a systematic gPDE reformulation of the Penrose notion of AAdS space, improving the earlier construction from~\cite{Grigoriev:2023kkk}. The main result of the paper is the construction of an algorithmic procedure which allows one to explicitly obtain the structure of the boundary gauge theory, including the equations of motion and gauge symmetries for boundary fields. This procedure is called the boundary calculus and is summarized in Theorems~\bref{bound-calculus}, \bref{YMbound-calculus} and Proposition~\bref{calculus-scalar}. To demonstrate the efficiency of the approach, we explicitly derived a Weyl-invariant generalization of the Yang--Mills equation in eight dimensions, which was not known before.

The approach of this work is analogous to Fefferman-Graham ambient metric construction except that it 
does not employ the ambient space and moreover the near-boundary expansion takes place in the fiber of the coorresponding gPDE rather than in spacetime. 
While FG approach usually focuses on the leading boundary values, we concentrate on the entire boundary structure which is described by a gauge theory involving both the leading and the subleading fields. We demonstrate that the boundary theory has a triangular structure in the sense that the leading equations of motion and gauge generators do not involve subleading fields while the subleading fields are defined on the background of the leading ones. In the gPDE terms this means that the boundary gPDE is itself a $Q$-bundle whose base describes the leading fields while the fibre corresponds to the subleading ones. In particular, in the case of gravity and Yang-Mills theory, the conservation equations for the subleading fields are well-defined if and only if the leading fields satisfy the obstruction equations. The interplay between the gravitational subleading data and the fields propagating on its background is also noteworthy. In particular, we have shown that Weyl-invariant operators with principal symbol $\Box^n$ on the background of the boundary gPDE for gravity exist for all $n>0$ and in all dimensions. However, when the boundary dimension is even, operators of sufficiently high order necessarily involve the gravitational subleading field, giving an interesting resolution of the puzzle of what is a natural generalisation of the higher order GJMS operators on generic conformal geometry background.

The analysis of this work is limited to field theories defined at the level of equations of motion. As for Lagrangians, in the gPDE framework these are encoded in the compatible presymplectic sructures defined on the corresponding gPDE, see~\cite{Grigoriev:2022zlq} for further details and  \cite{Alkalaev:2013hta,Grigoriev:2016wmk,Grigoriev:2020xec,Dneprov:2022jyn} for earlier relevant contributions. Gauge PDE equipped with a presymplectic structure defines a version of BV-BRST extended covariant phase phase space formalism. For free (gauge) fields on manifolds with asymptotic boundaries a version of this construction has been recently put forward in~\cite{Dneprov:2026muy}. Let us also mention an alternative approach to BRST-extended covariant phase space~\cite{Baulieu:2024oql} and earlier systematic study~\cite{Barnich:2001jy} of local gauge theories and their asymptotic symmetries and charges.

\section*{Acknowledgments}
\label{sec:Acknowledgements}
We wish to thank X.~Bekaert, I.~Dneprov, A.~Tseytlin and especially N.~Boulanger and R.~Gover for fruitful discussions. Useful exchanges with G. Barnich, M.~Ba\~nados, J.~Herfray are  acknowledged. M.G. also wishes to thank I.~Krasil'shchik, Th.~Popelensky, A.~Verbovetsky, and J.~Slov\'ak. M.M. is grateful to V.~Krivorol and A.~Selemenchuk for useful discussions. Part of this work was done when the authors participated in the thematic program ``From Asymptotic Symmetries to Flat Holography: Theoretical Aspects and Observable Consequences'' at the Galileo Galilei Institute for Theoretical Physics (Florence, Italy) and in the thematic program ``Carrollian physics and holography'' at the Erwin Schr\"odinger Institute for Mathematics and Physics (Vienna, Austria).

\section*{}

\appendix
\setlength{\itemsep}{1pt}
\small
\section{Proof of Theorem~\bref{prop:mnogoidealov}}\label{AppA}

The ideal $\cK^{(N)}$ is generated by $\cK^{(N-1)}$ together with
\begin{align}\label{th-proof-gen}
	\begin{split}
		&\nabla_{(C)}\Big(\Omega^{(N+2)}-\frac{1}{D}g_{\Omega\Omega}\nabla_\Omega^{N}\nabla^{a}\nabla_a\Omega\Big),\\
		&\nabla_{(C)}\nabla_\Omega^{N}\nabla_{A}\Omega^{(1)} +\cK_B,\\
		&\nabla_{(C)}\nabla_\Omega^{N}\Big(\nabla_A\nabla_B\Omega-\frac{1}{D}g_{AB}\nabla_c\nabla^{c}\Omega\Big).
	\end{split}
\end{align}	
Using 
\begin{align}
	[\nabla_{a},\nabla_{b}]\nabla_c\Omega=-\We^{d}{}_{cab}\nabla_d\Omega-\Co_{cab}\Omega
\end{align}
together with the tracelessness of the Weyl and Cotton tensors one can show that
\begin{align}
	\nabla_\Omega\nabla^{a}\nabla_a\Omega=\nabla^{a}\nabla_\Omega\nabla_a\Omega-\We^{ca}{}_{\Omega a}\nabla_c\Omega-\Co^{a}{}_{\Omega a}\Omega=g^{\Omega\Omega}\Omega^{(3)}+\nabla^{A}\nabla_A\Omega^{(1)}+\cK_B.
\end{align}
Consequently, the generators in \eqref{th-proof-gen} for $N\geq1$ can be rewritten as follows:
\begin{align}\label{th-proof-0}
	\nabla_{(C)}\Omega^{(N+2)},\quad \nabla_{(C)}\nabla_\Omega^{N}\nabla_{A}\Omega^{(1)},\quad \nabla_{(C)}\nabla_\Omega^{N}\nabla_A\nabla_B\Omega.
\end{align}
The first step in the proof is formulated as a lemma.
\begin{lemma}\label{th-lemma}
	
 a) The generators of the ideal $\cK^{(N)}$ 
 \begin{align}
\nabla_\Omega^{N}\nabla_{A}\Omega^{(1)},\quad \nabla_\Omega^{N}\nabla_A\nabla_B\Omega	
 \end{align}

  can be rewritten as 
\begin{align}
\begin{split}
	&(N-1)\nabla_\Omega^{N-2}\nabla^{B}T_{BA}+\cK^{(N-1)},\\
		&(D-3-N+1)T^{(N-1)}_{AB}-(N-1)\tilde{\Lambda}\nabla_\Omega^{N-2}\nabla^{C}J_{CBA}+\cK^{(N-1)}.
\end{split}
\end{align}

b) The vector fields $\nu_A^{(N-1)}$ preserve $\cK^{(N)}$.

\end{lemma}
\begin{proof}
	We will prove this by induction. 
	
	Base case: $N=1$
\begin{align}
\begin{split}
	&\nabla_\Omega\nabla_A\Omega^{(1)}=\nabla_A\Omega^{(2)}-\We^{d}{}_{\Omega A\Omega}\nabla_d\Omega-\Co_{\Omega \Omega A}\Omega\in \cK^{(0)},\\
	&\nabla_\Omega\nabla_A\nabla_B\Omega=\nabla_A\nabla_\Omega\nabla_B\Omega-\We^{d}{}_{B\Omega A}\nabla_d\Omega-\Co_{B\Omega A}\Omega=-g^{\Omega\Omega}T_{BA}+\cK^{(0)},
	\end{split}
\end{align}
where we also use the fact that $\nabla_A$ preserves $\cK^{(0)}$. The proof of the base case of part b) of the lemma follows from Lemma~\bref{lemma-I0} since $\nu_B^{(0)}=\frac{\partial}{\partial\xi^A}$.

Now suppose that the statement of the lemma holds for all $i\leq N$. We prove it for $N+1$. Part a):
\begin{align}\label{th-proof-1}
	\nabla_\Omega^{N+1}\nabla_{A}\Omega^{(1)}=\nabla_\Omega^{N}[\nabla_\Omega,\nabla_A]\Omega^{(1)}+\nabla_\Omega^{N}\nabla_A\Omega^{(2)}
\end{align}
Consider the second term. By Proposition~\bref{utv-commutator} and the induction hypothesis 
\begin{align}\label{th-proof-2}
\nabla_\Omega^{N}\nabla_A\Omega^{(2)}=\sum_{i=0}^{N}C^{i}_{N}\mathcal{D}_{A}^{(N-i)}\Omega^{(i+2)}=\mathcal{D}_A^{(N)}\Omega^{(2)}+\cK^{(N)},
\end{align}
since the fact that $\nu_A^{(N-1)}$ preserves $\cK^{(N)}$ implies, in particular, that $\mathcal{D}_A^{(N-1)}\equiv[\nu_A^{(N-1)},Q]$ preserves this ideal.

Using $Q\Omega=\xi^{a}\nabla_a\Omega+\lambda\Omega$ a straightforward computation shows that
\begin{align}\label{th-proof-3}
\begin{split}
	&Q\nabla_b\Omega=\xi^{a}\nabla_a\nabla_b\Omega+C_{b}{}^{a}\nabla_a\Omega+\lambda_b\Omega+\lambda\nabla_b\Omega,\\
	&Q\nabla_c\nabla_b\Omega=\xi^{a}\nabla_a\nabla_c\nabla_b\Omega+C_{c}{}^{a}\nabla_a\nabla_b\Omega+C_{b}{}^{a}\nabla_c\nabla_a\Omega+\lambda^{a}\nabla_a\Omega g_{bc}+\lambda\nabla_c\nabla_b\Omega.
    \end{split}
\end{align}
Then, using the explicit formula for $\nu_A^{(N)}$ \eqref{nu-counted-bulk}
\begin{align}\label{th-proof-4}
\begin{split}
            \nu^{(N)}_{A}=&-(\nabla_\Omega^{N-1}\We^{b}{}_{c\Omega A}+(N-1)\mathcal{P}^{db}_{c\Omega}\nabla^{N-2}_{\Omega}\Co_{d\Omega A})\frac{\partial}{\partial C_{c}{}^{b}}-\\&-\nabla_\Omega^{N-1}\Co_{d\Omega A}\frac{\partial}{\partial\lambda_d}+(N-1)\nabla_\Omega^{N-2}\Co_{\Omega\Omega A}\frac{\partial}{\partial \lambda}+\\&+\sum_{i=1}^{N-1}C^{i}_{N-1}(\nabla_{\Omega}^{N-i-1}\We^{b}{}_{\Omega\Omega A}+\frac{N-i-1}{i+1}\nabla_{\Omega}^{N-i-2}\mathcal{P}^{db}_{\Omega\Omega}\Co_{{d\Omega A}})\nu^{(i-1)}_b, \quad N\geq1
\end{split}
\end{align}
we obtain 
\begin{align}
\begin{split}
	\mathcal{D}_A^{(N)}\Omega^{(2)}=\nu_A^{(N)}Q\Omega^{(2)}=-\tilde{\Lambda}\nabla_\Omega^{N-1}\Co_{\Omega \Omega A} +\cK^{(1)}.
	\end{split}
\end{align}
Substituting this result first into \eqref{th-proof-2}, and then into \eqref{th-proof-1} we get
\begin{align}
		\nabla_\Omega^{N+1}\nabla_{A}\Omega^{(1)}=\nabla_\Omega^{N}[\nabla_\Omega,\nabla_A]\Omega^{(1)}-\tilde{\Lambda}\nabla_\Omega^{N-1}\Co_{\Omega\Omega A}+\cK^{(N)}.
\end{align}
Moreover, 
\begin{align}
	\nabla_\Omega^{N}[\nabla_\Omega,\nabla_A]\Omega^{(1)}=-\nabla_\Omega^{N}(\Co_{\Omega\Omega A}\Omega)+\cK_B=-\nabla_\Omega^{N-1}\Co_{\Omega\Omega A}\tilde{\Lambda}+\cK^{(N-2)}
\end{align}
and finally we have
\begin{align}\label{th-proof-10}
			\nabla_\Omega^{N+1}\nabla_{A}\Omega^{(1)}=-2\tilde{\Lambda}\nabla_\Omega^{N-1}\Co_{\Omega\Omega A}+\cK^{(N)}=\frac{2\tilde\Lambda}{D-3}\nabla_\Omega^{N-1}\nabla^{e}\We_{e\Omega\Omega A}+\cK^{(N)}.
\end{align}
which is precisely the desired formula, since $\nabla^{e}\We_{e\Omega\Omega A}=-\nabla^{E}T_{EA}$.

The next generator of the ideal:
\begin{align}
	\nabla_\Omega^{N+1}\nabla_A\nabla_B\Omega=\nabla_\Omega^{N}[\nabla_\Omega,\nabla_A]\nabla_B\Omega+\nabla_\Omega^N\nabla_A\nabla_B\Omega^{(1)}.
\end{align}
As in the previous case, we consider the second term:
\begin{align}
	\nabla_\Omega^N\nabla_A\nabla_B\Omega^{(1)}=\sum_{i=0}^{N}C^{i}_N\mathcal{D}_A^{(N-i)}\nabla_{\Omega}^{i}\nabla_B\Omega^{(1)}=\mathcal{D}_A^{(N)}\nabla_B\Omega^{(1)}+\cK^{(N)}.
\end{align}
Using formulas \eqref{th-proof-3}, \eqref{th-proof-4} one verifies that $\mathcal{D}_A^{(N)}\nabla_B\Omega^{(i)}\in \cK^{(1)}$, and therefore
\begin{align}\label{th-proof-5}
\begin{split}
		\nabla_\Omega^{N+1}\nabla_A\nabla_B\Omega=\nabla_\Omega^{N}[\nabla_\Omega,\nabla_A]\nabla_B\Omega+\cK^{(N)}=\\= -\nabla_\Omega^{N}(\We^{d}{}_{B\Omega A}\nabla_d\Omega+\Co_{B \Omega A}\Omega)+\cK^{(N)}=\\=T^{(N)}_{BA}+N\tilde\Lambda\nabla_\Omega^{N-1}\Co_{B\Omega A}+\cK^{(N)}.
\end{split}
\end{align}
Using 
\begin{align}
	\Co_{B\Omega A}=-\frac{1}{D-3}\nabla^{e}\We_{eB\Omega A}=-\frac{1}{D-3}(\tilde{\Lambda}^{-1}T^{(1)}_{BA}+\nabla^{E}J_{EBA})+\cK_B
\end{align}
we arrive at the original statement.

Part b) of the lemma is proved by a direct calculation. Since $\nu_A^{(N)}$ has degree $-1$, it suffices to verify that $\nu_A^{(N-1)}\cK_B \subset \cK^{(N)}$. Let us act with \eqref{th-proof-4} on the generators of $\cK_B$ of degree 1 \eqref{OmegaQOmega}, \eqref{boundcond}:
\begin{flalign}
\nu_{A}^{(N-1)}Q\Omega
={}&
\frac{N-3}{2}\nabla_\Omega^{N-4}\Co_{\Omega\Omega A}\tilde{\Lambda}
\nonumber\\
&+
\sum_{i=2}^{N-2}C^{i}_{N-2}
\left(
\nabla_{\Omega}^{N-i-2}\We^{b}{}_{\Omega\Omega A}
+
\frac{N-i-2}{i+1}
\nabla_{\Omega}^{N-i-3}
\mathcal{P}^{db}_{\Omega\Omega}
\Co_{{d\Omega A}}
\right)
\nu^{(i-1)}_b\Omega
+\cK_B
\in\cK^{(N-2)},
&&
\end{flalign}

\begin{flalign}
\nu_A^{(N-1)}
\left(
C_{\Omega}{}^{\Omega}
-g^{\Omega\Omega}\xi^{b}\nabla_b\Omega^{(1)}
+\lambda
\right)
&\in\cK^{(N-2)},
&&
\\[0.4em]
\nu_A^{(N-1)}
\left(
C_{\Omega}{}^{B}
-\xi^{c}\nabla_c\nabla^{B}\Omega
\right)
&=
-(-T^{(N-2)|B}{}_{A}
-(N-2)\tilde\Lambda\nabla_\Omega^{N-3}\Co^{B}{}_{\Omega A})
+\cK^{(N-2)}
\in\cK^{(N-1)},
&&
\\[0.4em]
\nu_A^{(N-1)}
\left(
\lambda_\Omega
+g^{\Omega\Omega}\xi^{a}\nabla_a\Omega^{(2)}
\right)
&=
-\nabla_\Omega^{N-2}\Co_{\Omega\Omega A}
+\cK^{(N-2)}
\in\cK^{(N)},
&&
\\[0.4em]
\nu_A^{(N-1)}
\left(
C_{B}{}^{\Omega}
+g^{\Omega\Omega}\xi^{c}\nabla_c\nabla_B\Omega
\right)
&\in\cK^{(N)}.
&&
\end{flalign}
\end{proof}
It is obvious that the proof of the theorem almost follows from Lemma~\bref{th-lemma} and Proposition~\bref{mnogoidealov}. The only remaining issue is that at the moment as new generators of the ideal, we have all $\nabla_\Omega^{N}\nabla^{A}T_{AB}$, $N\geq 0$, while under the theorem only $N=D-3$ is singled out. We will show that the others can be discarded from the generating set. 

Consider one of the generators of the ideal $\cK^{(N+1)}$ (see \eqref{th-proof-5})
\begin{align}
	\nabla^{B}\nabla_\Omega^{N}[\nabla_\Omega,\nabla_A]\nabla_B\Omega=[\nabla^{B},\nabla_\Omega^{N}][\nabla_\Omega,\nabla_A]\nabla_B\Omega+\nabla_\Omega^{N}\nabla^{B}[\nabla_\Omega,\nabla_A]\nabla_B\Omega. 
\end{align}
As in the proof of the lemma, the first term belongs to $\cK^{(N)}$, since the vector fields $\mathcal{D}_{B}^{(i)}, i<N$, preserve $\cK^{(i+1)}$, and 
\begin{align}
    \mathcal{D}^{(N)|B}\nabla_\Omega\nabla_A\nabla_B\Omega\in \cK^{(1)}\,,
\end{align}
which is easy to prove by explicit calculation. Expanding the second term then gives
\begin{multline}\label{th-proof-6}
	\nabla^{B}\nabla_\Omega^{N}[\nabla_\Omega,\nabla_A]\nabla_B\Omega=-\nabla_\Omega^{N}\nabla^{B}(\We^{d}{}_{B\Omega A}\nabla_d\Omega+\Omega\Co_{B\Omega A})+\cK^{(N)}=\\=-\nabla_\Omega^{N}\nabla^{B}T_{BA}+\nabla_\Omega^{N}(\nabla^{B}\Co_{B\Omega A}\Omega)+\cK^{(N)}=\\=-\nabla_\Omega^{N}\nabla^{B}T_{BA}-\nabla_\Omega^{N}(\nabla^{\Omega}\Co_{\Omega\Omega A}\Omega)+\cK^{(N)}=\\=-\frac{D-3-N}{D-3}\nabla_\Omega^{N}\nabla^{E}T_{EA}+\cK^{(N)},
\end{multline}
where we used the consequence of the Bianchi identity \eqref{Bianchi} $\nabla^{a}\Co_{abc}=0$. For $N\neq D-3$ this means that part of the generators are already contained in the prolongations of the others with respect to  $\nabla_A$ and therefore can be excluded. Note that for $N=D-3$ this leads to:
\begin{align}\label{App-div-free}
    \nabla^{B}O_{AB}^{(D-3)}\in\cK^{(D-3)}.
\end{align}
\section[The action of Q on T]{Action of $Q$ on $ T^{(N)}_{AB}$}\label{sec:prof-com}
We first prove the following auxiliary lemma.
\begin{lemma}\label{Lemma-Gamma}
On $\Ee$ we have:
    \begin{align}
                [\Gamma^B,\nabla_\Omega^{N}]=\sum_{i=0}^{N-2}d^{i}_{N}P^{Bc}_{\Omega\Omega}\mathcal{D}^{(N-2-i)}_{c}\nabla^{i}_\Omega+N\nabla_\Omega^{N-1}(g^{BC}\tilde\Lambda\Delta^{\Omega}{}_{C}-\Delta^{B}{}_{\Omega})+\cK^{(0)}.
    \end{align}
\end{lemma}
\begin{proof}
    First, we show that
    \begin{align}\label{gamma-bulk}
        [\Gamma^B,\nabla_\Omega^{N}]=\sum_{i=0}^{N-2}d^{i}_{N}\mathcal{P}^{Bc}_{\Omega\Omega}\mathcal{D}^{(N-2-i)}_{c}\nabla^{i}_\Omega-N\mathcal{P}^{Bc}_{d\Omega}\nabla^{N-1}_\Omega\Delta^{d}{}_{c}
    \end{align}
    To prove this, similarly to $[\nabla_{\Omega}^{N},\nabla_{A}]$, we expand $[\nabla^N_{\Omega},\Gamma^{B}]$ as follows:
    \begin{align}
        [\nabla^N_{\Omega},\Gamma^{B}]=\sum^{N-1}_{i=0}C^{i}_{N}\underbrace{[\nabla_{\Omega},\dots,[\nabla_\Omega}_{N-i},\Gamma^{B}]\dots]\nabla^{i}_{\Omega}.
    \end{align}
Using the commutation relations \eqref{commutrelations}, we get
\begin{align}
\begin{split}    
    &[\nabla_{\Omega},\Gamma^{B}]=\mathcal{P}^{Bc}_{d\Omega}\Delta^{d}{}_{c},\\
    &\underbrace{[\nabla_{\Omega},\dots,[\nabla_\Omega}_{i},\Gamma^{B}]\dots]=\mathcal{P}^{Bc}_{\Omega\Omega}\mathcal{D}^{(i-2)}_{c},\quad i\geq2.
    \end{split}
\end{align}
Thus,
\begin{align}\label{pr}
    [\nabla^N_{\Omega},\Gamma^{B}]=\sum_{i=0}^{N-2}C^{i}_N \mathcal{P}^{Bc}_{\Omega\Omega}\mathcal{D}_{c}^{(N-i-2)}\nabla^{i}_{\Omega}+N\mathcal{P}^{Bc}_{d\Omega}\Delta^{d}{}_{c}\nabla_{\Omega}^{N-1}.
\end{align}
 In the second term, we want to move $\nabla_{\Omega}^{N-1}$ to the left, for which we need:
\begin{align}
    [\nabla_{\Omega}^{N-1},\Delta^{d}{}_{c}]=\sum_{i=0}^{N-2}C^{i}_{N-1}\underbrace{[\nabla_{\Omega},\dots,[\nabla_\Omega}_{N-i},\Delta^{d}{}_{c}]\dots]\nabla^{i}_{\Omega}
\end{align}
From the same commutators \eqref{commutrelations} we get
\begin{align}
        \underbrace{[\nabla_{\Omega},\dots,[\nabla_\Omega}_{i},\Delta^{d}{}_{c}]\dots]=\delta^{d}_\Omega \mathcal{D}^{i-1}_{c},\quad i\geq1.
\end{align}
Substituting these expressions into \eqref{pr} gives 
\begin{align}
     [\nabla^N_{\Omega},\Gamma^{B}]=\sum_{i=0}^{N-2} (C^{i}_{N}-NC^{i}_{N-1})\mathcal{P}^{Bc}_{\Omega\Omega}\mathcal{D}^{(N-2-i)}_c\nabla^{i}_{\Omega}+N \mathcal{P}^{Bc}_{d\Omega}\nabla^{N-1}_{\Omega}\Delta^{d}{}_{c}.
\end{align}
It is straightforward to check that 
\begin{align}
    C^{i}_{N}-NC^{i}_{N-1}=-(N-1-i)C^{i}_{N}=-d^i_N\,,
\end{align}
which completes the proof of \eqref{gamma-bulk}.
The formula stated in the lemma follows immediately from the definition of $\mathcal{P}^{ab}_{cd}=(-g^{ab}g_{cd}+\delta^{a}_{c}\delta^{b}_d+\delta^{a}_d\delta^{b}_c)$ and from the fact that $g_{C\Omega}$ and $(g_{\Omega\Omega}-\tilde{\Lambda})$ belong to $\cK^{(0)}$.
\end{proof}
We can now prove the following proposition.
\begin{prop}\label{App-QT}
On $\Ee$ we have
    \begin{align}
        QT^{(N)}_{BC}&=\xi^{A}\nabla_{A}T^{(N)}_{BC}+C_{B}{}^{A}T^{(N)}_{AC}+C_{C}{}^{A}T^{(N)}_{BA}+\lambda_{A}\Gamma^A T^{(N)}_{BC}-N\lambda T^{(N)}_{BC}+\cK^{(0)},\qquad \text{where}\\
        \Gamma_{A}T_{BC}^{(N)}&=\tilde\Lambda\sum_{i=0}^{N-2}d^{i}_N\mathcal{D}^{(N-2-i)}_AT^{(i)}_{BC}+N\tilde\Lambda( J^{(N-1)}_{ABC}+J_{ACB}^{(N-1)})+\cK^{(0)}.
    \end{align}
    \end{prop}
\begin{proof}
    Notice that, by \eqref{boundcond} and Lemma ~\bref{lemma-I0}, on $\Ee$ we have
\begin{multline}
    QT^{(N)}_{BC}=\xi^{a}\nabla_{a}T^{(N)}_{BC}+C_{d}{}^{e}\Delta^{d}{}_{e}T^{(N)}_{BC}+\lambda_{a}\Gamma^aT^{(N)}_{BC}+\lambda\Delta T^{(N)}_{BC}=\\=\xi^{A}\nabla_{A}T^{(N)}_{BC}+C_{D}{}^{E}\Delta^{D}{}_{E}T^{(N)}_{BC}+\lambda_{A}\Gamma^{A}T^{(N)}_{BC}+\lambda(\Delta-\Delta^{\Omega}{}_{\Omega})T^{(N)}_{BC}+\cK^{(0)}
\end{multline}
For $N=0$, the actions of $\Delta_{D}{}^{E}$, $\Gamma_{A}$, $\Delta$, $\Delta^{\Omega}{}_{\Omega}$ follows directly from the action of $Q$ on $\We_{bcde}$ \eqref{W-counted}, namely
\begin{align}
    QT_{BC}=Q W_{\Omega B\Omega C}=\xi^{a}\nabla_aT_{BC}+C_\Omega{}^{A}J_{ABC}+C_{B}{}^{A}T_{AC}+C_{\Omega}{}^AJ_{ACB}+C_{C}{}^{A}T_{BA}+2C_{\Omega}{}^{\Omega}T_{BC}+2\lambda T_{BC}.
\end{align}

For higher $N$, we must commute these vector fields with $\nabla_{\Omega}^{N}$ using the commutator formulas \eqref{commutrelations}. For example, 
\begin{align}
    \Delta^{D}{}_{E}\nabla_\Omega^{N}T_{AB}=\nabla_\Omega^{N}\Delta^{D}{}_{E}T_{AB}+[\Delta^{D}{}_{E},\nabla_\Omega^{N}]T_{AB}=\nabla_\Omega^{N}\Delta^{D}{}_{E}T_{AB}=\delta^{D}_AT_{EB}^{(N)}+\delta^{D}_BT^{(N)}_{AE}.
\end{align} The only nontrivial commutator here is $[\Gamma^{B},\nabla_{\Omega}^{N}]$, for which one should use Lemma~\bref{Lemma-Gamma} proved above.\end{proof}

Using Lemma \bref{Lemma-Gamma}, we can also prove the following
\begin{prop}\label{prop-gamma-conservation}
\begin{align}
    \Gamma_{C}\nabla_\Omega^{D-3}\nabla_{A}T^{A}{}_{B}=-\tilde\Lambda(D-3)\nabla_{\Omega}^{D-4}\nabla_{D}J^{D}{}_{CB}+\cK^{(D-4)}.
    \end{align}
    \end{prop}
    \begin{proof}
        First, observe that as a consequence of $C_{abc}\equiv -\frac{1}{D-3}\nabla_dW^{d}{}_{abc}$, we have 
        \begin{align}\label{proof-gamma-obstr}
            \nabla_{\Omega}^{D-3}\nabla_{A}T^{A}{}_{B}=(D-3)\nabla_\Omega^{D-3}C_{\Omega \Omega A} +\cK^{(0)}.
        \end{align}
    Using Lemma \bref{Lemma-Gamma} and $\Gamma^{c}C_{abd}=-W^{c}{}_{abd}$ we can compute
    \begin{multline}
        \Gamma^{C}\nabla_\Omega^{D-3}C_{\Omega\Omega A}=\sum_{i=0}^{D-5}d^{i}_{D-3}P^{Cd}_{\Omega\Omega}\cD_{d}^{(D-5-i)}\nabla_\Omega^{i}C_{\Omega\Omega A}+(D-3)\nabla_\Omega^{D-4}(g^{BC}\tilde{\Lambda}\Delta^{\Omega}{}_{B}-\Delta^{C}{}_{\Omega})C_{\Omega\Omega A}+\\+T^{(D-3)|C}{}_{A}+\cK^{(0)}.
    \end{multline}
    In the proof of Theorem~\bref{mnogoidealov}, it was shown (see \eqref{th-proof-6}) that for $i<D-3$ one has 
$\nabla_\Omega^{i}C_{\Omega\Omega A}\in \cK^{(i+1)}$. 
Using this and the fact that the corresponding operators $\cD_A^{(N)}$ preserve this ideal as a consequence of the theorem itself, 
we conclude that the entire sum $\sum_{i=0}^{D-5}d^{i}_{D-3}P^{Cd}_{\Omega\Omega}\cD_{d}^{(D-5-i)}\nabla_\Omega^{i}C_{\Omega\Omega A}$ belongs to $\cK^{(D-4)}$. Moreover, it is easy to see that $\Delta^{C}{}_{\Omega}C_{\Omega\Omega A}=0$ and that $g^{BC}\Delta^{\Omega}{}_{B}C_{\Omega\Omega A}=C^{C}{}_{\Omega A}+C_{\Omega}{}^{C}{}_{A}$. Hence
\begin{align}
    \Gamma^{C}\nabla_\Omega^{D-3}C_{\Omega\Omega A}=(D-3)\tilde\Lambda( \nabla_\Omega^{D-4}C^{C}{}_{\Omega A}+\nabla_\Omega^{D-4}C_{\Omega}{}^{C}{}_{A})+T^{(D-3)|C}{}_{A}+\cK^{(D-4)}.
\end{align}
Since, as a consequence of the Bianchi identities \eqref{Bianchi}, we have $\nabla_\Omega C_{\Omega AB}=\nabla_A C_{\Omega \Omega B}-\nabla_BC_{\Omega\Omega A}$, it is easy to see that $\nabla_{\Omega}^{D-4}C_{\Omega CA}\in \cK^{(D-4)}$. Then, substituting the formula \eqref{Cotton-T}, which expresses $C_{A\Omega B}$ in terms of the components of the Weyl tensor, we obtain
\begin{align}
     \Gamma^{C}\nabla_\Omega^{D-3}C_{\Omega\Omega A}=-\tilde\Lambda \nabla_\Omega^{D-4}\nabla_EJ^{EC}{}_A+\cK^{(D-4)}.
\end{align}
Substituting this into \eqref{proof-gamma-obstr}, we arrive at the desired formula.
\end{proof}
\section[Proof of the formulas for J and nu]{Proof of the formulas for $J^{(N)}_{ABC}$ and $\nu_A$}\label{sec:prof-BC}
We begin with the following lemma:
\begin{lemma}\label{th2-proof-lemma}
On $\Ee$, the following relation holds:
    \begin{align}
		J^{(N)}_{ABC}+N\tilde\Lambda\nabla_\Omega^{N-1}\Co_{CAB}\in\cK^{(N+1)}.
	\end{align}
\end{lemma}
\begin{proof}
	From \eqref{th-proof-0} we have $\nabla_{A}\nabla_\Omega^{N}\nabla_B\nabla_C\Omega\in \cK^{(N)}$. On the other hand, we have
	\begin{align}
		[\nabla_A,\nabla_\Omega^{N}]\nabla_B\nabla_C\Omega=-\sum_{i=0}^{N-1}C^{i}_{N}\mathcal{D}_{A}^{(N-i)}\nabla_\Omega^{i}\nabla_B\nabla_C\Omega\in \cK^{(N+1)}
	\end{align}
which follows from Theorem~\bref{prop:mnogoidealov}.
	Next, observe that
	\begin{align}
		\cK^{(N)}\ni \nabla_{A}\nabla_\Omega^{N}\nabla_B\nabla_C\Omega-\nabla_{B}\nabla_\Omega^{N}\nabla_A\nabla_C\Omega=\nabla_\Omega^{N}[\nabla_A,\nabla_B]\nabla_C\Omega+\cK^{(N+1)}.
	\end{align}
	Expanding the commutator yields
	\begin{align}
		\cK^{(N+1)}\ni\nabla_\Omega^{N}(\We^{d}{}_{CAB}\nabla_d\Omega+\Co_{CAB}\Omega)=J^{(N)}_{ABC}+N\tilde{\Lambda}\nabla_\Omega^{N-1}\Co_{CAB}+\cK^{(N-1)}.
	\end{align}
	\end{proof}

	We now prove the formula for $J^{(N)}_{ABC}$:
    \begin{prop}\label{App-J-counted}
On $\Ee$, the following holds:
    \begin{align}
    \begin{split}
    &J^{(1)}_{ABC}=\frac{\tilde{\Lambda}}{D-4}\nabla_EW^{E}{}_{CAB}+\cK^{(1)},\\
    &J^{(N)}_{ABC}=\frac{N}{N-1}\nabla_\Omega^{N-1}(\nabla_AT_{BC}-\nabla_B T_{AC})+\cK^{(N+1)},\qquad N>1.
    \end{split}
    \end{align}
    \end{prop}
    \begin{proof}
            Substituting two $\Omega$ indices into the Bianchi identity \eqref{Bianchi} yields
	\begin{align}\label{th2-proof-1}
	\begin{split}
				J^{(1)}_{ABC}=&\nabla_A T_{BC}-\nabla_B T_{AC}+6\Co_{\Omega [AB}g_{\Omega]C}-6\Co_{C[AB}g_{\Omega]\Omega}=\\=&\nabla_A T_{BC}-\nabla_B T_{AC}-\Co_{\Omega\Omega B}g_{AC}+\Co_{\Omega \Omega A}g_{BC}-\Co_{CAB}g_{\Omega\Omega}+I^0(g_{\Omega A}),
					\end{split}
	\end{align}
	where $I^0(g_{\Omega A})\subset \cK_B$ denotes the ideal generated by $g_{\Omega A}$. Note that this ideal is stable under $\nabla_a$, since $\nabla_ag_{bc}=0$. Expanding
	\begin{align}
		C_{CAB}=-\frac{1}{D-3}(\nabla_e \We^{e}{}_{CAB})=-\frac{1}{D-3}(g^{\Omega\Omega}J_{ABC}^{(1)}+\nabla_E \We^{E}{}_{CAB})+I^0(g_{\Omega A})
	\end{align}
	and moving the term with $J^{(1)}_{ABC}$ to the left-hand side, we obtain
	\begin{align}
		\frac{D-4}{D-3}J^{(1)}_{ABC}=\nabla_A T_{BC}-\nabla_B T_{AC}-\Co_{\Omega\Omega B}g_{AC}+\Co_{\Omega \Omega A}g_{BC}+\frac{g_{\Omega\Omega}}{D-3}\nabla_E\We^{E}{}_{CAB}+I^0(g_{\Omega A}).
	\end{align}
	Since $T_{AB}\in\cK^{(1)}, $ $\Co_{\Omega\Omega B}=\frac{1}{D-3}\nabla_ET^{E}{}_{B}\in\cK^{(1)}$,  we arrive at the desired formula (valid for $D>4$):
	\begin{align}
		 J^{(1)}_{ABC}=\tilde{\Lambda}\frac{1}{D-4}\nabla_E  \We^{E}{}_{CAB}+\cK^{(1)}.
	\end{align}
	Applying $\nabla_\Omega^{N-1}$ to $\eqref{th2-proof-1}$ and using Lemma~\bref{th2-proof-lemma} to express $\nabla_\Omega^{N-1}\Co_{CAB}$ we obtain
	\begin{align}
		\frac{N-1}{N}J^{(N)}_{ABC}=\nabla_\Omega^{N-1}(\nabla_AT_{BC}-\nabla_B T_{AC})-\nabla_\Omega^{N-1}\Co_{\Omega\Omega B}g_{AC}+\nabla_\Omega^{N-1}\Co_{\Omega\Omega A}g_{BC}+\cK^{(N+1)}.
	\end{align}
	Since (see \eqref{th-proof-10})  we have $\nabla_\Omega^{N-1}\Co_{\Omega\Omega B}\in \cK^{(N+1)}$,  we arrive at the desired formula.
    \end{proof}
\begin{prop}\label{App-nu-counted}
      The following is true on $\Ee$:
\begin{align}
            &\nu_{A}^{(N)}=-J^{(N-1),D}{}_{CA}\frac{\partial}{\partial C_{C}{}^D}+\dfrac{1}{\tilde\Lambda N} T^{(N)}_{CA}\frac{\partial}{\partial\lambda_C}-\frac{1}{N}\sum_{i=1}^{N-2} d^{i}_N  T^{(i),C}{}_{A}\nu_{C}^{(N-2-i)}+\cK^{(N+1)}\,
\end{align}
\end{prop}
\begin{proof}
First of all (see \eqref{commutrelations-pr-1}) we have
\begin{align}
    \nu_A^{(1)}=r_{\Omega A}=-W^{d}{}_{c\Omega A}\frac{\partial}{\partial C_{c}{}^{d}}-C_{c\Omega A}\frac{\partial}{\partial\lambda_c}.
\end{align}
By repeatedly applying $[\nabla_\Omega,\cdot]$ to this formula, we arrive at the following iterative relation:
\begin{multline}
       \nu^{(N)}_{A}=-(\nabla_\Omega^{N-1}\We^{b}{}_{c\Omega A}+(N-1)\mathcal{P}^{db}_{c\Omega}\nabla^{N-2}_{\Omega}\Co_{d\Omega A})\frac{\partial}{\partial C_{c}{}^{b}}-\\-\nabla_\Omega^{N-1}\Co_{d\Omega A}\frac{\partial}{\partial\lambda_d}+(N-1)\nabla_\Omega^{N-2}\Co_{\Omega\Omega A}\frac{\partial}{\partial \lambda}+\\+\sum_{i=1}^{N-1}C^{i}_{N-1}\Bigl(\nabla_{\Omega}^{N-i-1}\We^{b}{}_{\Omega\Omega A}+\frac{N-i-1}{i+1}\nabla_{\Omega}^{N-i-2}\mathcal{P}^{db}_{\Omega\Omega}\Co_{{d\Omega A}}\Bigr)\nu^{(i-1)}_b.
\end{multline}

Note that from the definition of $\mathcal{P}^{db}_{ca}$, we have
\begin{align}
	\mathcal{P}^{dB}_{C\Omega}=\delta^B_C\delta^d_\Omega+\cK_B,\quad 
 \mathcal{P}^{d\Omega}_{\Omega\Omega}=\delta^{d}_\Omega+\cK_B,\quad 
 \mathcal{P}^{dB}_{\Omega\Omega}=-g^{dB}\tilde{\Lambda}+\cK^{(0)}.
\end{align}

From formulas \eqref{th-proof-10} and \eqref{th-proof-5} it follows that
\begin{align}
\begin{split}
	&\nabla_\Omega^{N-1}\Co_{\Omega\Omega A}\in \cK^{(N+1)},\\
	&T_{DA}^{(N)}+N\tilde\Lambda\nabla_\Omega^{N-1}\Co_{D\Omega A}\in \cK^{(N+1)}.
\end{split}
\end{align}
Thus, the remaining coefficients are:
\begin{align}
	&\nabla_\Omega^{N-i-1}\We^{B}{}_{\Omega\Omega A}+\frac{N-1-i}{i+1}\nabla_\Omega^{N-i-2}\mathcal{P}^{dB}_{\Omega\Omega}\Co^{B}{}_{\Omega A}=\\
	&=-T^{(N-i-1)|B}{}_{A}+\frac{1}{i+1}T^{(N-i-1)|B}{}_{A}+\cK^{(N-i)}
	=-\frac{i}{i+1}T^{(N-i-1)|B}{}_{A}+\cK^{(N-i)}.
\end{align}

Consequently, the formula for $\nu_A^{(N)}$ takes the form
\begin{align}
\label{nu-count-1}
\begin{split}
	\nu_A^{(N)}=-J^{(N-1)|B}{}_{CA}\frac{\partial}{\partial C_{C}{}^B}
	+\frac{1}{N\tilde\Lambda}T^{(N)}_{AB}\frac{\partial}{\partial\lambda_B}
	-\sum_{i=1}^{N-1}C^{i}_{N-1}\frac{i}{i+1}T^{(N-i-1)|B}{}_{A}\nu_B^{(i-1)}
	+\cK^{(N+1)}.
\end{split}
\end{align}
Then, using the properties of the binomial coefficient and changing the summation index in the sum, we arrive at the desired formula.
\end{proof}

\section{Proof of Proposition \bref{utv-nechet}}
\label{sec:proof-TDJ}
\begin{proof} 
    Starting from 
\begin{equation}
        \hat J^{(0)}_{ABC}=0\,,\qquad \hat T^{(1)}_{BC}=0
\end{equation}
and $\nu_{A}^{(1)}=0$, which follows from \eqref{nu-count-1}, we proceed by induction on $i$. Assume that
\begin{align}
    \hat T^{(2k-1)}_{BC}=0 \quad \nu_{A}^{(2k-1)}=0,\quad \hat J^{(2k-2)}_{ABC}=0
\end{align}
for all  $1\leq k<\dfrac{D-4}{2}-1$. We show that
\begin{align}
    \hat T^{(2k+1)}_{BC}=0 \quad \nu_{A}^{(2k+1)}=0,\quad \hat J^{(2k)}_{ABC}=0.
\end{align}
To see this, let us compute  $\hat J^{(2k)}$.
\eqref{formuli-calculus} implies
\begin{align}\label{J2k}
    \frac{2k-1}{2k}\hat J^{(2k)}_{ABC}=\hat b^{*}\nabla_{\Omega}^{2k-1}(\nabla_{A}T_{BC}-\nabla_{B}T_{AC})=\hat b^{*}([\nabla^{2k-1}_{\Omega},\nabla_{A}]T_{BC}-[\nabla^{2k-1}_{\Omega},\nabla_{B}]T_{AC})=0,
\end{align}
where in the second equality $\nabla^{2k-1}_{\Omega}$ passes through $\nabla_{A}$ and we use $\hat T^{(2k-1)}_{BC}=0$. Then 
\begin{align}
   \hat  b^{*}[\nabla^{2k-1}_{\Omega},\nabla_{A}]T_{BC}=\sum_{i=0}^{2k-2}C^{i}_{2k-1} \mathcal{D}_{A}^{(2k-1-i)}\hat T^{(i)}_{BC}=\sum_{i=0}^{2k-2}C^{i}_{2k-1} [Q,\nu_{A}^{(2k-1-i)}]\hat T^{(i)}_{BC}=0
\end{align}
by assumption (the upper index of $\nu_A$ or $T_{AB}$ is odd).

Now consider $\hat T^{(2k+1)}_{AB}$. Equations ~\eqref{formuli-calculus} imply:
\begin{align}
    0=(D-4-2k)\hat T^{(2k+1)}_{BC}-(2k+1)\tilde{\Lambda}(\nabla^{A}\hat J^{(2k)}_{ABC}+\hat b^{*}[\nabla_{\Omega}^{2k},\nabla^{A}]J_{ABC})\rightarrow \hat T^{(2k+1)}_{BC}=0,
\end{align}
where the second term vanishes thanks to  \eqref{J2k}, while
\begin{align}
    \hat b^{*}[\nabla_{\Omega}^{2k},\nabla^{A}]J_{ABC}=0.
\end{align} 
by the assumption (each term vanishes separately because it contains either even $\hat J_{ABC}$ or odd  $\mathcal{D}_{A}$).

Finally, consider $\nu_{A}^{(2k+1)}$. Using \eqref{nu-count-1}  one gets: 
    \begin{align}
        \nu_{A}^{(2k+1)}=-\hat J^{(2k),D}{}_{CA}\frac{\partial}{\partial C_{C}{}^{D}}+\dfrac{1}{\tilde{\Lambda}(2k+1)}\hat T^{(2k+1)}_{CA}\frac{\partial}{\partial \lambda_{C}}-\sum_{i=1}^{2k} C^{i}_{2k}\frac{i}{i+1} \hat T^{(2k-i),C}{}_{A}\nu_{C}^{(i-1)}.
    \end{align}
   As we have seen, the first two terms vanish. The remaining terms also vanish by the induction assumption. 
\end{proof}
\section[d=8 Yang-Mills]{$d=8$ Yang-Mills}\label{App-YM9}
\begin{proof}
Analogously to the examples for $d=4$ \bref{YM-example-5D} and $d=6$ ~\bref{YM-example-7D} we have
\begin{align}\label{app-ym8-usef}
    \begin{split}
                \hat T^{(2)}_{BC}&=\frac{\tilde{\Lambda}^{2}}{2}\hat B_{BC},\qquad \hat J^{(1)}_{ABC}=-\tilde\Lambda \hat C_{CAB},\qquad \hat J^{(3)}_{ABC}=\frac{3\tilde\Lambda^{2}}{4}(\nabla_{A}\hat B_{BC}-\nabla_{B}\hat B_{AC})\,,\\
        \hat J^{(1)}_B&=\frac{\tilde\Lambda}{4}\nabla^{A}\hat F_{AB},\qquad \Gamma_C\hat J_B^{(1)}=\tilde\Lambda \hat F_{CB},\qquad  \nabla^{B}\hat J_{B}^{(3)}=0\,,\\
        \hat J^{(3)}_B&=\frac{3\tilde\Lambda^{2}}{8}( (\nabla_{A}\nabla^{A})\nabla^{C}\hat F_{CB}+8\hat C^{DA}{}_{B}\hat F_{DA}-2[\hat F_{AB},\nabla^{C}\hat F_C{}^{A}])\,,\\
        \hat F^{(2)}_{AB}&=\nabla_{A}\hat J^{(1)}_B-\nabla_{B}\hat J^{(1)}_A, \quad \Gamma_{C}\hat F^{(2)}_{AB}=\tilde\Lambda\nabla_{C}\hat F_{AB}+2\hat g_{CB}\hat J^{(1)}_{A}-2\hat g_{CA}\hat J^{(1)}_{B},\\
        \Gamma_{C}\hat J_B^{(3)}&=3\tilde\Lambda(2\nabla_{C}\hat J^{(1)}_B-\nabla_{B}\hat J_C^{(1)})\,.
        \end{split}
\end{align}
In this dimension, the functions $\cY_B$ \eqref{YM-obstr} are given by
        \begin{align}\label{App-9-bstr}
    \mathcal{Y}_B=\nabla^{A}\hat{F}_{AB}^{(4)}+\mathcal{D}^{(4)|A}\hat F_{AB}+6\mathcal{D}^{(2)|A}\hat F_{AB}^{(2)}\,.
    \end{align}
Using Theorem~\bref{YMbound-calculus} and \eqref{YMformuli-D}, the terms appearing in $\cY_B$ can be rewritten as
    \begin{align}\label{App-9-bstr-formulas}
        \begin{split}
            \mathcal{D}^{(4)|A}\hat F_{AB}&=-\hat J^{(3)|C}{}_{B}{}^{A}\hat F_{AC}+[\hat J^{(3)|A},\hat F_{AB}]-\frac{1}{4}d^{2}_4\hat T^{(2)|CA}\nabla_{C}\hat F_{AB}=\\&=-\frac{3\tilde\Lambda^2}{4}\nabla_{C}\hat B_{BA}\hat F^{AC}+[\hat J^{(3)|A},\hat F_{AB}]-\frac{3\tilde\Lambda^{2}}{4}\hat B^{CA}\nabla_{C}\hat F_{AB},\\
            \mathcal{D}^{(2)|A}\hat F^{(2)}_{AB}&=-\hat J^{(1)|C}{}_{B}{}^{A}\hat F_{AC}^{(2)}+[\hat J^{(1)|A},\hat F_{AB}^{(2)}]+\frac{\hat T^{(2)|AC}}{2\tilde\Lambda}\Gamma_{C}\hat F^{(2)}_{AB}=\\&=\tilde\Lambda\hat C^{AC}{}_{B}(\nabla_{A}\hat J^{(1)}_{C}-\nabla_{C}\hat J^{(1)}_{A})+[\hat J^{(1)|A},\hat F_{AB}^{(2)}]+\frac{\tilde\Lambda^2\hat B^{AC}}{4}\nabla_{C}\hat F_{AB}+\frac{\tilde\Lambda \hat B^{A}{}_{B}}{2}\hat J^{(1)}_{A},\\
            \nabla^A\hat F_{AB}^{(4)}&=\nabla^{A}\hat b^{*}\nabla_\Omega^{3}(\nabla_{A}\hat J_{B}-\nabla_{B}\hat J_A)=(\nabla_{A}\nabla^{A})\hat J^{(3)}_B-\nabla^{A}\nabla_{B}\hat J^{(3)}_{A}+3\nabla^{A}(\mathcal{D}^{(2)}_{A}\hat J_B^{(1)}-\mathcal{D}^{(2)}_{B}\hat J^{(1)}_{A}).
        \end{split}
    \end{align}
    Note that $\nabla^{B}\hat J_{B}^{(3)}=0$ implies
    \begin{multline}
    -\nabla^{A}\nabla_B\hat J_A^{(3)}=-[\nabla^{A},\nabla_B]\hat J^{(3)}_A=\hat C^{DA}{}_{B}\Gamma_{D}\hat J_A^{(3)}-[\hat F_{AB},\hat J^{(3)|A}]=\\=3\tilde\Lambda \hat C^{DA}{}_{B}(2\nabla_{D}\hat J_{A}^{(1)}-\nabla_{A}\hat J_{D}^{(1)})-[\hat F_{AB},\hat J^{(3)|A}].
    \end{multline}
Using this and $\mathcal{D}_{A}^{(2)}\hat J^{(1)}_{B}=\tilde\Lambda \hat C_{A}{}^{C}{}_{B}\hat J^{(1)}_C+\frac{\tilde\Lambda^2}{4}\hat B_{AC}\hat F^{C}{}_{B}+[\hat J^{(1)}_{A},\hat{J}^{(1)}_{B}]$ in the expression for $\nabla^A\hat F_{AB}^{(4)}$, after some algebra, we obtain
\begin{multline}
        \nabla^A\hat F_{AB}^{(4)}=(\nabla_{A}\nabla^{A})\hat J^{(3)}_B+6\tilde\Lambda\hat C^{CA}{}_{B}\nabla_{C}\hat J^{(1)}_{A}+\frac{3\tilde\Lambda^2}{4}\hat B_{AC}\nabla^{A}\hat F^{C}{}_{B}-\frac{3\tilde\Lambda^2}{4}\nabla^A\hat B_{BC}\hat F^{C}{}_{A}-\\-[\hat F_{AB},\hat J^{(3)|A}]+6[\hat J^{(1)}_A,\nabla^{A}\hat J^{(1)}_B]. 
\end{multline}
Therefore,  \eqref{App-9-bstr} takes the form

\begin{multline}
            \mathcal{Y}_B=(\nabla_{A}\nabla^{A})\hat J^{(3)}_B+\frac{3\tilde{\Lambda}^2}{2}(\hat B^{AC}\nabla_{C}\hat F_{AB}-\hat F_{CA}\nabla^{A}\hat B_{B}{}^{C}+\frac{1}{2}\hat B^{A}{}_{B}\nabla^{C}\hat F_{CA})+\\+6\tilde\Lambda \hat C^{AC}{}_{B}(2\nabla_{A}\hat J_C^{(1)}-\nabla_{C}\hat J^{(1)}_{A})+2[\hat J^{(3)|A},\hat F_{AB}]+6[\hat J^{(1)|A},2\nabla_A\hat J_B^{(1)}-\nabla_B\hat J_A^{(1)}].
\end{multline}
Using the formulas above, it is straightforward to compute $\mathcal{Y}$: 
\begin{multline}
            \mathcal{Y}=\nabla^{A}\cJ_A+\sum_{i=0}^4C^i_5 \mathcal{D}^{(5-i)|A}\hat J^{(i)}_{A}=\nabla^{A}\cJ_A+5\mathcal{D}^{(4)|A}\hat J^{(1)}_{A}+10\mathcal{D}^{(2)|A}\hat J^{(3)}_{A}=\\=\nabla^{A}\cJ_A-\frac{15\tilde\Lambda^{2}}{4}\hat B^{AC}\nabla_C\hat J_{A}^{(1)}+\frac{15\tilde\Lambda^{2}}{2}\hat B^{AC}(2\nabla_{C}\hat J^{(1)}_A-\nabla_A\hat J^{(1)}_C)-5[\hat J^{(3)|A},\hat J^{(1)}_A]=\\=\nabla^{A}\cJ_A+\frac{15\tilde\Lambda^{2}}{4}\hat B^{AC}\nabla_C\hat J_{A}^{(1)}-5[\hat J^{(3)|A},\hat J^{(1)}_A].
\end{multline}
\end{proof}

\section{Solution space of the boundary theory: the Yang-Mills case}\label{App-YM-sol}
 Let  $\sigma_G$ : $T[1]\Sigma\to \hat \Ee$ be a fixed, generic solution of the off-shell gravity boundary gPDE $(\hat E,Q)$ in a ``metric-like'' gauge, as described in Section~\bref{sec:gravity-sections}.
In particular,
$\sigma_G^{*}\hat\xi^{A}=\theta^A$, $\sigma_G^{*}(\hat  \lambda)=0$. As discussed in this section, the choice of $\sigma_G$ corresponds to the choice of a conformal metric on $\Sigma$ together with the gravity subleading $\mathbf{T}_{AB}(x)$ which is irrelevant for the discussion of Yang-Mills.

In this appendix, we will focus on solutions of the gPDE $(\hat \Ee^\YM,Q)$ on a fixed background $\sigma_G$, i.e., sections $\sigma: T[1]\Sigma \to \hat \Ee^\YM$ such that $\hat \pi^\YM \circ \sigma = \sigma_G$\footnote{Equivalently, they can be described as sections of the gPDE $(\hat \Ee^\YM|_{\sigma_G},Q, T[1]\Sigma)$.}.
 
From the action of $Q$
\begin{align}
    Q\hat{\mathcal C}
    = -\frac{1}{2}[\hat{\mathcal C},\hat{\mathcal C}]
      + \frac{1}{2}\,\hat{\xi}^{A}\hat{\xi}^{B}\hat F_{AB}
      \end{align}
one finds that 
\begin{align}
\sigma^{*}\hat F_{AB}
    = 2\partial_{[A} A_{B]} + [A_{[A},A_{B]}]\,,
\end{align}
where we have introduced the notation $\sigma^{*}\hat{\mathcal C}\equiv A_B(x)\theta^B$, and $A_B$ is $\mathfrak{g}$-valued. Identifying $A_B$ with a connection and introducing its curvature $F_{AB}[A]$, this formula takes the form  $\sigma^{*}\hat F_{AB}=F_{AB}[A]$.

Analogously to the gravitational case, we introduce the following vector field  along the section:
\begin{align}
    \nabla_C^{g}\sigma^{*}f=\dl{x^C} \sigma^*f- \Gamma_{AC}^{B} (\sigma^{*}\Delta^{A}{}_{B}f)-A^{I}_C(\sigma^{*}R_If)\,,
\end{align}
 which is defined even when $\sigma$ is not a solution. Analogously to the gravitational case, such $\nabla_C^{g}$ can, in a suitable coordinate system, be regarded as the Levi-Civita covariant derivative twisted by the connection $A$. We will adopt this interpretation henceforth.

Using the fact that for any solution the action of $Q$ on a degree $0$ function can be decomposed as
\begin{align}
    Qf=\xi^{A}\nabla_{A}f+C_{A}{}^{B}\Delta^{A}{}_{B}f+\lambda_{A}\Gamma^{A}f+\lambda \Delta f+\cC^{I}R_If\,,
\end{align}
we find (see Section~\bref{sec:gravity-sections})
\begin{align}
        \sigma^{*}\nabla_{C}f=\dl{x^C} f- \Gamma_{AC}^{B} (\sigma^{*}\Delta^{A}{}_{B}f)+P_{CA}(\sigma^{*}\Gamma^{A}f)-A^{I}_C(\sigma^{*}R_If)=\nabla_C^{g}\sigma^{*}f+P_{C}{}^{A}(\sigma^{*}\Gamma_Af)\,.
\end{align}
Furthermore, the formula for the Laplacian takes the same form as in \eqref{nabla-laplacian}.
\begin{multline}\label{YM-nabla-laplacian}
\sigma^{*}\nabla_{A}\nabla^{A}f=\nabla_A^g\nabla^{g|A}\sigma^{*}f+(\partial_AP)\sigma^{*}\Gamma^Af+2P^{AD}\nabla^{g}_{A}\sigma^{*}\Gamma_D f+P^{AB}P_{A}{}^{C}\sigma^{*}\Gamma_B\Gamma_Cf+\\+P\sigma^{*}(-\Delta^{A}{}_{A}+\hat\Delta) f\,.
\end{multline}

Then, repeating the analysis of Section~\bref{sec:gravity-sections} we see that the space of solutions is parametrized by two off-shell $\mathfrak{g}$-valued fields: $A_B$ and $\mathbf{J}_A(x)\equiv\sigma^{*} \cJ_A$. The on-shell conditions $\sigma^{*}\cY_B=\sigma^{*}\cY=0$ will lead to differential equations on these fields.

In the $d=4$ case (see Example~\bref{YM-example-5D}) using $\Gamma_A\hat F_{BC}=0$, $\Gamma_B\cJ_A=\tilde\Lambda \hat F_{BA}$ we have 
\begin{align}\label{App-5D-result}
\begin{split}
    &\sigma^{*}\cY_B=\nabla^{g}_AF^{A}{}_B,\\
    &\sigma^{*}\cY=\nabla_A^{g}\mathbf{J}^{A}.
    \end{split}
\end{align}
These are the Yang-Mills equation and the covariant conservation equation for the connection $A_B$ and the current $\mathbf{J}_A$.

For the $d=6$ case (see Example~\bref{YM-example-7D}) we need to note that from the boundary calculus Theorem~\bref{YMbound-calculus} we have $\hat J^{(1)}_B=\frac{\tilde\Lambda}{D-5}\nabla^{A}\hat F_{AB}$, $D>5$. Therefore, from formula \eqref{YM-Gamma} we have 
\begin{align}
    \Gamma_D\nabla^{A}\hat F_{AB}=(D-5)\hat F_{DB}.
\end{align}
Then, substituting the formula for Laplacian \eqref{YM-nabla-laplacian} into $\sigma^{*}\nabla_C\nabla^{C}\nabla^{A}\hat F_{AB}$, we obtain from \eqref{6deqs}\footnote{Recall that in Section~\bref{sec:gravity-sections} it was shown that $\sigma^{*}\hat C_{ABC}=C_{ABC}[g]$, $\sigma^{*}\hat B_{AB}=B_{AB}[g]$.}
\begin{align}\label{App-7D-result}
\begin{split}
        \frac{2}{\tilde\Lambda}\sigma^{*}\cY_B&= \Box^g \nabla^{g}_AF^{A}{}_{B}+2\partial_CP F^{C}{}_{B}+4P^{AC}\nabla^g_AF_{CB}-3P\nabla^{g}_AF^{A}{}_{B}+4C^{DA}{}_{B}F_{DA}-2[F_{AB},\nabla_C^gF^{CA}],\\
    \sigma^{*}\cY&=\nabla_{A}^g\mathbf{J}^{A}+\frac{3}{2}\tilde\Lambda^{2}P^{AC}(\nabla^g_C\nabla^{g|B}F_{BA}),
    \end{split}
\end{align}
where $\Box^g\equiv\nabla_A^g\nabla^{g|A}.$
The first formula here is the Weyl-invariant generalization of the Yang-Mills equations, coinciding with the one appearing in \cite{Gover:2023rch}. The second is a modified conservation law for the current $\mathbf{J}^{A}$, which reduces to the standard conservation law, for example, when $F_{BA}=0$.

The case $d=8$ (see Example~\bref{YM-example-9D}) is technically quite involved. We need to compute
\begin{multline}
                \sigma^{*}\mathcal{Y}_B=\sigma^{*}\Big( (\nabla_{A}\nabla^{A})\hat J^{(3)}_B+\frac{3\tilde{\Lambda}^2}{2}(\hat B^{AC}\nabla_{C}\hat F_{AB}-\hat F_{CA}\nabla^{A}\hat B_{B}{}^{C}+\frac{1}{2}\hat B^{A}{}_{B}\nabla^{C}\hat F_{CA})+\\+6\tilde\Lambda \hat C^{AC}{}_{B}(2\nabla_{A}\hat J_C^{(1)}-\nabla_{C}\hat J^{(1)}_{A})+2[\hat J^{(3)|A},\hat F_{AB}]+6[\hat J^{(1)|A},2\nabla_A\hat J_B^{(1)}-\nabla_B\hat J_A^{(1)}] \Big).
\end{multline}

To begin with, we introduce the notation 
\begin{align}
    j^{(1)}_B\equiv\frac{4}{\tilde\Lambda}\sigma^{*}\hat J^{(1)}_B\,,\quad j^{(3)}_B\equiv \frac{8}{3\tilde\Lambda^{2}}\sigma^{*}\hat J^{(3)}_B\,.
\end{align}
These quantities can be computed in terms of the connection $A_B$ similarly to \eqref{App-5D-result}, \eqref{App-7D-result} (only the coefficients differ). Applying $\sigma^{*}$ to \eqref{YM-example-9D-J1} we obtain the following:
\begin{align}\label{App-sol-j}
    \begin{split}
        &j^{(1)}_B=\nabla^g_AF^{A}{}_{B},\\
        &j^{(3)}_B=\Box^g\nabla^{g}_AF^{A}{}_{B}+4\partial_CP F^{C}{}_{B}+8P^{AC}\nabla^g_AF_{CB}-3P\nabla^{g}_AF^{A}{}_{B}+8C^{DA}{}_{B}F_{DA}-2[F_{AB},\nabla_C^gF^{CA}]\,.
    \end{split}
\end{align}
To compute $\sigma^{*}\nabla_{A}\nabla^{A}J^{(3)}_{B}$ using the formula for the Laplacian \eqref{YM-nabla-laplacian}, we will need:
\begin{align}
    \Gamma_D\hat J_B^{(3)}&=3\tilde\Lambda(2\nabla_D\hat J^{(1)}_B-\nabla_B\hat J^{(1)}_D)\,,\\
    \Gamma_{E}\Gamma_D\hat J_B^{(3)}&=3\tilde\Lambda(\delta^{C}_{E}g_{BD}+\delta ^{C}_{D}g_{EB}-5 \delta^{C}_Bg_{ED})\hat J^{(1)}_{C}+3\tilde\Lambda^{2}(2\nabla_D\hat F_{EB}-\nabla_B\hat F_{ED})\,,
\end{align}
where the first formula was obtained in \eqref{App-9-bstr-formulas}, and the second can be derived from the first using the commutator formula (similarly to \eqref{commutrelations}) $[\Gamma^{E},\nabla_B]=-\mathcal{P}^{EC}_{DB}\Delta^{D}{}_{C}+\delta^{E}_{B}\hat \Delta$ and $\Gamma_E\hat J^{(1)}_B=\tilde\Lambda \hat F_{EB}$. Finally, putting all the pieces together, 
\begin{multline}
                \frac{8}{3\tilde\Lambda^2}\sigma^{*}\mathcal{Y}_B=\Box^{g}j^{(3)}_B+2 (\partial^DP+2P^{AD}\nabla_A^{g})(2\nabla_D^g j^{(1)}_B-\nabla_B^{g}j ^{(1)}_D+8P_{D}{}^{C}F_{CB}-4P_{B}{}^{C}F_{CD})+ \\+2(-5P_{AD}P^{AD}j^{(1)}_B+8 P_{AE}P^{AD}\nabla^g_{D}F^{E}{}_{B}+2P_{AD}P^{A}{}_{B}j^{(1)|D})-5Pj^{(3)}_B+\\+4( B^{AC}\nabla^g_{C} F_{AB}- F^{CA}(\nabla^{g}_A B_{BC}-4P^{E}{}_{A}C_{BEC}-4P^{E}{}_{A}C_{CEB})+\frac{1}{2} B^{A}{}_{B}\nabla^{g}_{C} F^{C}{}_{A})+\\+4 C^{AC}{}_{B}(2\nabla^g_{A}j_C^{(1)}-\nabla^g_{C}j^{(1)}_{A}+8P_{A}{}^{D}F_{DC}-4P_{C}{}^{D}F_{DA})+\\+2[j^{(3)|A}, F_{AB}] +[ j^{(1)|A},2\nabla^g_A j_B^{(1)}-\nabla^g_Bj_A^{(1)}+8P_{A}{}^{C}F_{CB}-4P_{B}{}^{C}F_{CA}]\,,
                \end{multline}
where we also used $\Gamma_{A}\hat B_{BC}=-(D-5)(\hat C_{CAB}+\hat C_{BAC})$ to calculate $\sigma^{*}\nabla_A\hat B_{BC}$. 

Similarly, the modified conservation law for $\mathbf{J}^{A}$ can be computed:
\begin{multline}
                \sigma^{*}\mathcal{Y}=\nabla^g_A\mathbf{J}^A+\frac{5\tilde\Lambda^{3}}{4}P^{CB}(C_{C}{}^{D}{}_{B}j^{(1)}_D+B_{CD}F^{D}{}_{B}+3\nabla_{C}^{g}j_{B}^{(3)}+6P_{C}{}^{E}\nabla_{E}^gj^{(1)}_{B})+\\+\frac{15\tilde{\Lambda}^{3}}{16} B^{AC}(\nabla^g_{C} j^{(1)}_A)-\frac{15\tilde\Lambda^3}{32}[j^{(3)|A},j^{(1)}_A].
\end{multline}

\setlength{\itemsep}{0em}
\small
\providecommand{\href}[2]{#2}\begingroup\raggedright\endgroup

\end{document}